\documentclass[11pt,documentclass,onecolumn]{article}
\usepackage{amsmath}
\usepackage{amsthm}
\usepackage{amssymb}
\usepackage{algorithm}
\usepackage{subfig}
\usepackage{color}
\usepackage[english]{babel}
\usepackage{graphicx}
\usepackage{wrapfig,epsfig}
\usepackage{psfrag}
\usepackage{epstopdf}
\usepackage{url}
\usepackage{graphicx}
\usepackage{color}
\usepackage{epstopdf}
\usepackage{algpseudocode}
\usepackage{tikz}
\usetikzlibrary{arrows}
\usepackage[margin=1in]{geometry}
\author{Eric Price and Zhao Song
\\ ecprice@cs.utexas.edu zhaos@utexas.edu
\\ The University of Texas at Austin
}

\title{A Robust Sparse Fourier Transform in the Continuous Setting}
\newtheorem{theorem}{Theorem}[section]
\newtheorem{lemma}[theorem]{Lemma}
\newtheorem{definition}[theorem]{Definition}

\newtheorem{corollary}[theorem]{Corollary}

\newtheorem{claim}[theorem]{Claim}

\newcommand{\abs}[1]{|#1|}
\newcommand{\tabs}[1]{\left|#1\right|}
\newcommand{\norm}[1]{\|#1\|}

\newcommand{\wh}{\widehat}
\newcommand{\eps}{\epsilon}
\newcommand{\N}{\mathcal{N}}
\newcommand{\R}{\mathbb{R}}
\newcommand{\C}{\mathbb{C}}
\newcommand{\Z}{\mathbb{Z}}
\newcommand{\RN}[1]{%
  \textup{\uppercase\expandafter{\romannumeral#1}}%
}

\renewcommand{\i}{\mathbf{i}}
\DeclareMathOperator*{\E}{\mathbb{E}}
\DeclareMathOperator*{\median}{median}

\DeclareMathOperator{\supp}{supp}
\DeclareMathOperator{\poly}{poly}
\DeclareMathOperator{\sinc}{sinc}

\usepackage{lineno}

\usepackage{etoolbox}

\makeatletter

\newcommand{\define}[4][ignore]{%
  \ifstrequal{#1}{ignore}{}{
  \@namedef{thmtitle@#2}{#1}}%
  \@namedef{thm@#2}{#4}%
  \@namedef{thmtypen@#2}{lemma}%
  \newtheorem{thmtype@#2}[theorem]{#3}%
  \newtheorem*{thmtypealt@#2}{#3~\ref{#2}}%
}

\newcommand{\state}[1]{%
  \@namedef{curthm}{#1}
  \@ifundefined{thmtitle@#1}{
  \begin{thmtype@#1}
    }{
  \begin{thmtype@#1}[\@nameuse{thmtitle@#1}]
  }
    \label{#1}
    \@nameuse{thm@#1}
  \end{thmtype@#1}
  \@ifundefined{thmdone@#1}{
  \@namedef{thmdone@#1}{stated}%
  }{}
}

\newcommand{\restate}[1]{%
  \@namedef{curthm}{#1}
  \@ifundefined{thmtitle@#1}{
    \begin{thmtypealt@#1}
    }{
  \begin{thmtypealt@#1}[\@nameuse{thmtitle@#1}]
  }
    \@nameuse{thm@#1}
  \end{thmtypealt@#1}
  \@ifundefined{thmdone@#1}{
  \@namedef{thmdone@#1}{stated}%
  }{}
}

\newcommand{\thmlabel}[1]{
  \@ifundefined{thmdone@\@nameuse{curthm}}{\label{#1}
    }{\tag*{\eqref{#1}}}
}
\makeatother

\begin{document}

\begin{titlepage}
  \maketitle
  \begin{abstract}
    In recent years, a number of works have studied methods for
    computing the Fourier transform in sublinear time if the output is
    sparse.  Most of these have focused on the discrete setting, even
    though in many applications the input signal is continuous and
    naive discretization significantly worsens the sparsity level.

    We present an algorithm for robustly computing sparse Fourier
    transforms in the continuous setting.  Let $x(t) = x^*(t) + g(t)$,
    where $x^*$ has a $k$-sparse Fourier transform and $g$ is an
    arbitrary noise term.  Given sample access to $x(t)$ for some
    duration $T$, we show how to find a $k$-Fourier-sparse
    reconstruction $x'(t)$ with
    \[
    \frac{1}{T}\int_0^T \abs{x'(t) - x(t)}^2 \mathrm{d} t \lesssim \frac{1}{T}\int_0^T \abs{g(t)}^2 \mathrm{d}t.
    \]
    The sample complexity is linear in $k$ and logarithmic in the
    signal-to-noise ratio and the frequency resolution.  Previous
    results with similar sample complexities could not tolerate an
    infinitesimal amount of i.i.d. Gaussian noise, and even algorithms
    with higher sample complexities increased the noise by a
    polynomial factor.  We also give new results for how precisely the
    individual frequencies of $x^*$ can be recovered.
  \end{abstract}
  \thispagestyle{empty}
\end{titlepage}

\section{Introduction}

The Fourier transform is ubiquitous in digital signal processing of a
diverse set of signals, including sound, image, and video.  Much of
this is enabled by the Fast Fourier Transform (FFT)~\cite{CT65}, which
computes the $n$-point discrete Fourier transform in $O(n \log n)$
time.  But can we do better?

In many situations, much of the reason for using Fourier transforms is
because the transformed signal is \emph{sparse}---i.e., the energy is
concentrated in a small set of $k$ locations.  In such situations, one
could hope for a dependency that depends nearly linearly on $k$ rather
than $n$.  Moreover, one may be able to find these frequencies while
only sampling the signal for some period of time.  This idea has led
to a number of results on \emph{sparse Fourier transforms}, including
\cite{GGIMS,GMS,HIKP,IK}, that can achieve $O(k \log (n/k) \log n)$
running time and $O(k \log (n/k))$ sample complexity (although not
quite both at the same time) in a robust setting.


These works apply to the discrete Fourier transform, but lots of
signals including audio or radio originally come from a continuous
domain.  The standard way to convert a continuous Fourier transform
into a discrete one is to apply a window function then subsample.
Unfortunately, doing so ``smears out'' the frequencies, blowing up the
sparsity.  Thus, one can hope for significant efficiency gains by
directly solving the sparse Fourier transform problem in the
continuous setting.  This has led researchers to adapt techniques from
the discrete setting to the continuous both in
theory~\cite{BCGLS,TBSR,CF14,DB13} and in practice~\cite{MRS-sparse}.
However, these results are not robust to noise: if the signal is
sampled with a tiny amount of Gaussian noise or decays very slightly
over time, no method has been known for computing a sparse Fourier
transform in the continuous setting with sample complexity linear in
$k$ and logarithmic in other factors.  That is what we present in this
paper.

Formally, a vector $x^*(t)$ has a $k$-sparse Fourier transform if it
can be written as
\[
x^*(t) = \sum_{i=1}^k v_i e^{2\pi \i f_i t}
\]
for some \emph{tones} $\{(v_i, f_i)\}$.  We consider the problem
where we can sample some signal
\[
x(t) = x^*(t) + g(t)
\]
at any $t$ we choose in some interval $[0, T]$, where $x^*(t)$ has a
$k$-sparse Fourier transform and $g(t)$ is arbitrary noise.  As long
as $g$ is ``small enough,'' one would like to recover a good
approximation to $x$ (or to $x^*$, or to $\{(v_i, f_i)\}$) using
relatively few samples $t \in [0, T]$ and fast running time.  Our
algorithm achieves several results of this form, but a simple one is
an $\ell_2/\ell_2$ guarantee: we reconstruct an $x'(t)$ with
$k$-sparse Fourier transform such that
\[
\frac{1}{T}\int_0^T \abs{x'(t) - x(t)}^2 \mathrm{d} t \lesssim \frac{1}{T}\int_0^T \abs{g(t)}^2 \mathrm{d}t
\]
using a number of samples that is $k$ times logarithmic
factors\footnote{We use $f \lesssim g$ to denote that $f \leq C g$ for
  some universal constant $C$.}.  To the best of our knowledge, this
is the first algorithm achieving such a constant factor approximation
with a sample complexity sublinear in $T$ and the signal-to-noise
ratio.

Our algorithm also gives fairly precise estimates of the individual
tones $(v_i, f_i)$ of the signal $x^*$.  To demonstrate what factors
are important, it is helpful to think about a concrete setting.  Let
us consider sound from a simplified model of a piano.

\paragraph{Thought experiment: piano tuning} In a simplified model of
a piano, we have keys corresponding to frequencies over some range
$[-F, F]$.  The noise $g(t)$ comes from ambient noise and the signals
not being pure tones (because, for example, the notes might decay
slowly over time).  For concrete numbers, a modern piano has 88 keys
spaced from about 27.5 Hz to $F = 4200$Hz. The space between keys
ranges from a few Hz to a few hundred Hz, but most chords will have an
$\eta = 30$Hz or more gap between the frequencies being played.  One
typically would like to tune the keys to within about $\pm \nu =
1$Hz. And piano music typically has $k$ around $5$.

Now, suppose you would like to build a piano tuner that can listen to
a chord and tell you what notes are played and how they are tuned.
For such a system, how long must we wait for the tuner to identify the
frequencies?  How many samples must the tuner take?  And how robust is
it to the noise?

If you have a constant signal-to-noise ratio, you need to sample for a
time $T$ of at least order $1/\nu = 1$ second in order to get $1$Hz
precision---frequencies within $1$Hz of each other will behave very
similarly over small fractions of a second, which noise can make
indistinguishable.  You also need at least $\Omega(k
\log\frac{F}{k\nu}) \approx 50$ samples, because the support of the
signal contains that many bits of information and you only get a
constant number per measurement (at constant SNR).  At higher
signal-to-noise ratios $\rho$, these results extend to
$\Omega(\frac{1}{\nu \rho})$ duration and $\Omega(k \log_\rho
\frac{F}{k\nu})$ samples.  But as the signal-to-noise ratio gets very
high, there is another constraint on the duration: for $T <
\frac{1}{\eta} \approx 33$ milliseconds the different frequencies
start becoming hard to distinguish, which causes the robustness to
degrade exponentially in $k$~\cite{M15} (though the lower bound there
only directly applies to a somewhat restricted version of our
setting).

This suggests the form of a result: with a duration $T >
\frac{1}{\eta}$, one can hope to recover the frequencies to within
$\frac{1}{\rho T}$ using $O(k \log_\rho \frac{FT}{k})$ samples.  We
give an algorithm that is within logarithmic factors of this ideal:
with a duration $T > \frac{O(\log (k/\delta))}{\eta}$, we recover the
frequencies to within $O(\frac{1}{\rho T})$ using $O(k \log_\rho
(FT) \cdot \log(k/\delta) \log k)$ samples, where $\rho$ and
$1/\delta$ are (roughly speaking) the minimum and maximum
signal-to-noise ratios that you can tolerate, respectively.

Instead of trying to tune the piano by recovering the frequencies
precisely, one may simply wish to record the sound for future playback
with relatively few samples.  Our algorithm works for this as well:
the combination $x'(t)$ of our recovered frequencies satisfies
\[
\frac{1}{T}\int_0^T \abs{x'(t) - x(t)}^2 \mathrm{d}t \lesssim \frac{1}{T}\int_0^T \abs{g(t)}^2 \mathrm{d}t.
\]

Let us now state our main theorems.  The first shows how well we can
estimate the frequencies $f_i$ and their weights $v_i$; we refer to
this $(v_i, f_i)$ pair as a \emph{tone}.

\define[Tone estimation]{thm:main1}{Theorem}{%
Consider any signal $x(t): [0, T] \to \C$ of the form
\begin{align*}
x(t) = x^*(t) + g(t),
\end{align*}
for arbitrary ``noise'' $g(t)$ and an exactly $k$-sparse $x^* =
\sum_{i \in [k]} v_i e^{2\pi \i f_i t}$ with frequencies $f_i
\in [-F, F]$ and frequency separation $\eta = \min_{i \neq j} \abs{f_i
  - f_j}$.  For some parameter $\delta > 0$, define the ``noise level''
\[
\N^2 := \frac{1}{T}\int_0^T \abs{g(t)}^2   \mathrm{d}t + \delta \sum_{i=1}^k \abs{v_i}^2.
\]
We give an algorithm that takes samples from $x(t)$ over any duration
$T > O(\frac{\log(k/\delta)}{\eta})$ and returns a set of $k$
tones $\{(v'_i, f'_i)\}$ that approximates $x^*$ with error
proportional to $\N$.  In particular, every large tone is
recovered: for any $v_i$ with $\abs{v_i} \gtrsim \N$, we have for an
appropriate permutation of the indices that
\begin{align}
  \abs{f'_i - f_i} &\lesssim \frac{\N}{T\abs{v_i}}  & \text{and} &\quad&
  \abs{v'_i - v_i} \lesssim \N.\thmlabel{eq:1}
\end{align}
In fact, we satisfy a stronger guarantee that the \emph{total} error
is bounded:
\begin{align}
  \sum_{i=1}^k \frac{1}{T}\int_{0}^T \abs{v'_i e^{2\pi \i f'_i t} - v_i e^{2\pi \i f_i t}}^2 \mathrm{d}t \lesssim \N^2.\thmlabel{eq:2}
\end{align}
The algorithm takes $O(k\log (FT)\log(\frac{k}{\delta})
\log(k))$ samples and  $O(k\log (FT)\log(\frac{FT}{\delta})
\log(k))$ running time, and succeeds with
probability at least $1-1/k^c$ for an arbitrarily
large constant $c$.
}
\state{thm:main1}

We then show that the above approximation of the individual tones is
good enough to estimate the overall signal $x(t)$ to within constant
factors:

\define[Signal estimation]{thm:main2}{Theorem}{%
  In the same setting as Theorem~\ref{thm:main1}, if the duration is
  slightly longer at $T > O(\frac{\log(1/\delta) + \log^2 k}{\eta})$,
  the reconstructed signal $x'(t) = \sum_{i=1}^k v'_i e^{2\pi \i f'_i
    t}$ achieves a constant factor approximation to the complete
  signal $x$:
\begin{align}
  \frac{1}{T}\int_{0}^T \abs{x'(t) - x(t)}^2 \mathrm{d} t \lesssim \N^2.\thmlabel{eq:3}
\end{align}

The algorithm takes $O(k\log (FT)\log(\frac{k}{\delta})
\log(k))$ samples and $O(k\log (FT)\log(\frac{FT}{\delta})
\log(k))$ running time, and succeeds with
probability at least $1-1/k^c$ for an arbitrarily
large constant $c$.
}

\state{thm:main2}

The above theorems give three different error guarantees, which are
all in terms of a ``noise level'' $\N^2$ that is the variance of the
noise $g(t)$ plus $\delta$ times the energy of the signal.  The
algorithm depends logarithmically on $\delta$, so one should think of
$\N^2$ as being the variance of the noise, e.g. $\sigma^2$ if samples
have error $N(0, \sigma^2)$.

\paragraph{Error guarantees.} Our algorithm does a good job of
estimating the signal, but how exactly should we quantify this?
Because very few previous results have shown robust recovery in the
continuous setting, there is no standard error measure to use.  We
therefore bound the error in three different ways: the maximum error
in the estimation of any tone; the weighted total error in the
estimation of all tones; and the difference between the
reconstructed signal and the true signal over the sampled interval.
The first measure has been studied before, while the other two are to
the best of our knowledge new but useful to fully explain the
robustness we achieve.

The error guarantee~\eqref{eq:1} says that we achieve good recovery of
any tones with magnitude larger than $C\N$ for some constant $C$.
Note that such a requirement is necessary: for tones with $\abs{v_i}
\leq \N$, one could have $g(t) = -v_ie^{2\pi\i f_it}$, completely
removing the tone $(v_i, f_i)$ from the observed $x(t)$ and making it
impossible to find.  For the tones of sufficiently large magnitude, we
find them to within $\frac{\N}{T \abs{v_i}}$.  This is always less
than $1/T$, and converges to $0$ as the noise level decreases.  This
is known as \emph{superresolution}--one can achieve very high
frequency resolution in sparse, nearly noiseless settings.  Moreover,
by Lemma~\ref{lem:lower} our ``superresolution'' precision $\abs{f'_i
  - f_i} \lesssim \frac{\N}{T \abs{v_i}}$ is optimal.


While the guarantee of~\eqref{eq:1} is simple and optimal given its
form, it is somewhat unsatisfying.  It shows that the maximum error
over all $k$ tones is $\N$, while one can hope to bound the
\emph{total} error over all $k$ tones by $\N$.  This is
precisely what Equation~\eqref{eq:2} does.  The guarantee~\eqref{eq:1}
is the precision necessary to recover the tone to within $O(\N)$
average error in time, that is~\eqref{eq:1} is equivalent to
\[
\frac{1}{T}\int_{0}^T \abs{v'_i e^{2\pi \i f'_i t} - v_i e^{2\pi \i f_i t}}^2 \mathrm{d} t \lesssim \N^2 \qquad \forall i \in [k].
\]
In~\eqref{eq:2}, we show that this bound holds even if we sum the left
hand side over all $i \in [k]$, so the \emph{average} error is a
factor $k$ better than would be implied by~\eqref{eq:1}.  It also means that
the total mass of all tones that are not recovered to within
$\pm 1/T$ is $O(\N)$, not just that every tone larger than
$O(\N)$ is recovered to within $\pm 1/T$.

The stronger bound~\eqref{eq:2} can be converted to the
guarantee~\eqref{eq:3}, which is analogous to the $\ell_2/\ell_2$
recovery guarantee standard in compressive sensing.  It bounds the
error of our estimate in terms of the input noise, on average over the
sampled duration.  The standard form of the $\ell_2/\ell_2$ guarantee
in the discrete sparse Fourier transform setting~\cite{GGIMS, GMS,
  HIKP,IKP,IK} compares the energy in frequency domain rather than
time domain.  This cannot be achieved directly in the continuous
setting, since frequencies infinitesimally close to each other are
indistinguishable over a bounded time interval $T$.  But if the signal
is periodic with period $T$ (the case where a discrete sparse Fourier
transform applies directly), then~\eqref{eq:3} is equivalent to the
standard guarantee by Parseval's identity.  So~\eqref{eq:3} seems like
the right generalization of $\ell_2/\ell_2$ recovery to our setting.

\paragraph{Other factors.} Our algorithm succeeds with high
probability in $k$, which could of course be amplified by repetition
(but increasing the sample complexity and running time).  Our running
time, at $O(k \log (FT) \log (FT/\delta) \log k)$, is (after
translating between the discrete and continuous setting) basically a
$\log k$ factor larger than the fastest known algorithm for discrete
sparse Fourier transforms~(\cite{HIKP}).  But since that result only
succeeds with constant probability, and repetition would also increase
that by a $\log k$ factor, our running time is actually equivalent to
the fastest known results that succeed with high probability in $k$.

Our sample complexity is $O(\log (k/\delta) \log k)$ worse than the
presumptive optimal, which is known to be achievable in the discrete
setting~(\cite{IK}).  However, the techniques that~\cite{IK} uses to
avoid the $\log (k/\delta)$ factor seem hard to adapt to the
continuous setting without losing some of the robustness/precision
guarantees.

Another useful property of our method, not explicitly stated in the
theorem, is that the sampling method is relatively simple: it chooses
$O(\log (FT) \log k)$ different random arithmetic sequences of points,
where each sequence has $O(k \log (k/\delta))$ points spread out over
a constant fraction of the time domain $T$.  Thus one could implement
this in hardware using a relatively small number of samplers, each of
which performs regular sampling at a rate of $\frac{k \log
  (k/\delta)}{T}$.  This is in contrast to the Nyquist rate for
non-sparse signals of $2F$ --- in the piano example, each sampler has
a rate on the order of $50$Hz rather than $8000$Hz.

The superresolution frequency is optimal because two signals of
magnitude $v_i$ and frequency separation $\nu < 1/T$ will differ by
$O(\nu^2 T^2 \abs{v_i}^2)$ over the duration $T$, so for $\nu$ below
our threshold the difference is just $\N^2$.  Hence if the observed
signal $x(t)$ looks identical to $(v_i, f_i)$, it might actually be
$(v_i, f'_i)$ with noise equaling the difference between the two.

The only previous result known in the form of~\eqref{eq:1}
was~\cite{M15}, which lost a $\poly(k, \frac{\max \abs{v_i}}{\min
  \abs{v_i}}, \eta)$ factor in noise tolerance and also did not
optimize for sample complexity.

\subsection{Comparison to Naive Methods}

This section compares our result to some naive ways one could try to
solve this problem by applying algorithms not designed for a
continuous, sparse Fourier transform setting.  The next section will
compare our result to algorithms that are designed for such a setting.

\paragraph{Nyquist Sampling.}  The traditional theory of band-limited
signals from discrete samples says that, from samples taken at the
Nyquist rate $2F$, one can reconstruct an $F$-band-limited signal
exactly.  The Whittaker-Shannon interpolation formula then says that
\begin{align}
  x^*(t) = \sum_{i=-\infty}^\infty x^*(2Fi) \sinc (2Ft - i)\label{e:WS}
\end{align}
where $\sinc(t)$ is the normalized $\sinc$ function $\frac{\sin(\pi
  t)}{\pi t}$.  This is for the band-limited ``pure'' signal $x^*$,
but one could then get a relationship for samples of the actual signal
$x(t)$.  This has no direct implications for learning the tones
(e.g. our~\eqref{eq:1} or~\eqref{eq:2}), but for learning the signal
(our~\eqref{eq:3}) there is also an issue.  Even in the absence of
noise and for $k=1$, this method will have error polynomially rather
than exponentially small in the number of samples.

That is, if there is no noise the method has zero error given
infinitely many samples.  But we only receive samples over the
interval $[0, T]$, leading to error.  Consider the trivial setting of
$x(t) = 1$.  The partial sum of~\eqref{e:WS} at a given $t$ will be
missing terms for $i > 2Ft$ and $i < 0$, which (for a random $t$ in
$[0, T/2]$) have magnitude at most $1/(Ft)$.  The terms alternate in
sign, so the sum has error approximately $1/(Ft)^2$.  This means that
the error over the first $1/F$ time is a constant, leading to average
error of $\frac{1}{FT}$.  This is with an algorithm that uses $FT$
samples and time.  By contrast, our algorithm in the noiseless setting
has error exponentially small in the samples and running time.

\paragraph{Discrete Sparse Fourier Transforms.}  An option related to
the previous would be to discretize very finely, then apply a discrete
sparse Fourier transform algorithm to keep the sample complexity and
runtime small.  The trouble here is that sparse Fourier transforms
require sparsity, and this process decreases the sparsity.  In
particular, this process supposes that the signal is periodic with
period $T$, so one can analyze this process as first converting the
signal to one, equivalent over $[0, T]$, but with frequency spectrum
only containing integer multiples of $1/T$.  This is done by
convolving each frequency $f_i$ with a sinc function (corresponding to
windowing to $[0, T]$) then restricting to multiples of $1/T$
(corresponding to aliasing).  The result is that a one-sparse signal
$e^{2 \pi \i f_i t}$ is viewed as having Fourier spectrum
\[
\wh{x''}[j] = \sinc(f_iT - j)
\]
for $j \in \Z$.  When $f_i$ is not a multiple of $1/T$, this means the
signal is not a perfectly sparse signal.  And this is true regardless
of the discretization level, which only affects the subsequent error
from aliasing $j \in \Z$ down to $\Z_n$.  To have error proportional to $\delta
\norm{\wh{x}^*}_2$, one would need to run such methods for a sparsity
level of $k/\delta$.  Thus, as with Nyquist sampling, the sample and
runtime will be polynomial, rather than logarithmic, in $\delta$.

The above discussion refers to methods for learning the signal
(our~\eqref{eq:3}).  In terms of learning the tones, one could
run the algorithm for sparsity $O(k)$ so that $\delta$ is a small
constant, which would let one learn roughly where the peaks are and
get most of the frequencies to the nearest $1/T$.  This would give a
similar bound to our~\eqref{eq:1}, but without the superresolution
effect as the noise becomes small.  On the plus side, the duration
could be just $O(1/\eta)$---which is sufficient for the different
peaks to be distinguishable---rather than $O(\frac{\log k}{\eta})$ as
our method would require, and the time and sample complexities could
save a $\log k$ factor (if one did not want to recover \emph{all} the
tones, just most of them).

Essentially, this algorithm tries to round each frequency to the
nearest multiple of $1/T$, which introduces noise that is a constant
fraction of the signal.  If the signal-to-noise ratio is already low,
this does not increase the noise level by that much so such an
algorithm will work reasonably well.  If the signal-to-noise ratio is
fairly high, however, then the added noise leads to much worse
performance.  Getting a constant factor approximation to the whole
signal is only nontrivial for high SNR, so such a method does badly in
that setting.  For approximating the tones, it is comparable to
our method in the low SNR setting but does not improve much as SNR
increases.


\subsection{Previous Work In Similar Settings}
There have been several works that recover continuous frequencies from
samples in the time domain.  Some of these are in our setting where
the samples can be taken at arbitrary positions over $[0, T]$ and
others are in the discrete-time (DTFT) setting where the samples must
be taken at multiples of the Nyquist frequency $\frac{1}{2F}$.

The results of~\cite{TBSR,CF14,YX15,DB13} show that a convex program
can solve the problem in the DTFT setting using $O(k \log k \log
(FT))$ samples if the duration is $T > O(\frac{1}{\eta})$, in the
setting where $g(t) = 0$ and the coefficients of $x^*$ have random
phases.  The sample complexity can be one log factor better than ours,
which one would expect for the noiseless setting.  None of these
results show robustness to noise, and some additionally require a
running time polynomial in $FT$ rather than $k$.

The result of~\cite{BCGLS} is in a similar setting to our paper, using
techniques of the same lineage.  It achieves very similar sample
complexity and running time to our algorithm, and a guarantee similar
in spirit to~\eqref{eq:1} with some notion of robustness.  However,
the robustness is weaker than ours in significant ways.  They consider
the noise $g(t)$ in frequency space (i.e. $\wh{g}(f)$), then require
that $\wh{g}(f)$ is zero at any frequency within $\eta$ of the signal
frequencies $f_i$, and bound the error in terms of $\N' =
\norm{\wh{g}}_1/k$ instead of $\norm{g}_2$.  This fails to cover
simple examples of noise, including i.i.d. Gaussian noise $g(t) \sim
N(0, \sigma^2)$ and the noise one would get from slow decay of the
signal over time (e.g. $x(t) = x^*(t) e^{-\frac{t}{100T}}$.).  Both
types of noise violate both assumptions on the noise: $\wh{g}(f)$ will
be nonzero arbitrarily close to each $f_i$ and $\norm{\wh{g}}_1$ will
be unbounded.  Their result also requires a longer duration than our
algorithm and has worse precision for any fixed duration.

The result of \cite{M15} studies noise tolerance in the DTFT setting,
ignoring sample complexity and running time.  It shows that the matrix
pencil method~\cite{HS90}, using $FT$ samples, achieves a guarantee of
the form~\eqref{eq:1}, except that the bounds are an additional
$\poly(FT, k, \delta)$ factor larger.  Furthermore, it shows a sharp
characterization of the minimal $T$ for which this is possible by any
algorithm: $T = (1 \pm o(1))\frac{2}{\eta}$ is necessary and
sufficient.  It is an interesting question whether the lower bound
generalizes to our non-DTFT setting, where the samples are not
necessarily taken from an even grid.

Lastly,~\cite{MRS-sparse} tries to apply sparse Fourier transforms to
a domain with continuous signals.  They first apply a discrete sparse
Fourier transform then use hill-climbing to optimize their solution
into a decent set of continuous frequencies.  They have interesting
empirical results but no theoretical ones.

\section{Algorithm Overview}

At a high level, our algorithm is an adaptation of the framework used
by~\cite{HIKP} to the continuous setting.  However, getting our result
requires a number of subtle changes to the algorithm.  This section
will describe the most significant ones.  We assume some familiarity
with previous work in the area~\cite{GGIMS,CCF02,GMS,GLPS,HIKP}.

First we describe a high-level overview of the structure. The
algorithm proceeds in $\log k$ stages, where each stage attempts to
recover each tone with a large constant probability (e.g. 9/10).
In each stage, we choose a parameter $\sigma \approx \frac{T}{k \log
  (k/\delta)}$ that we think of as ``hashing'' the frequencies into
random positions.  For this $\sigma$, we will choose about $\log (FT)$
different random ``start times'' $t_0$ and sample an arithmetic
sequence starting at $t_0$, i.e. observe
\[
x(t_0), x(t_0 + \sigma), x(t_0 + 2\sigma), \dotsc, x(t_0 + (k \log (k/\delta)) \sigma)
\]
We then scale these observations by a ``window function,'' which has
specific properties but among other things scales down the values near
the ends of the sequence, giving a smoother transition between the
time before and after we start/end sampling.  We alias this down to
$B=O(k)$ terms (i.e. add together terms $1, B+1, 2B+1, \dotsc$ to get
a $B$-dimensional vector) and take the $B$-dimensional DFT.  This
gives a set of $B$ values $\wh{u}_i$.  The observation made in
previous papers is that $\wh{u}$ is effectively a \emph{hashing} of
the tones of $\wh{x}$ into $B$ buckets, where $\sigma$ defines a
permutation on the frequencies that affects whether two different
tones land in the same bucket, and $\wh{u}_j$ approximately
equals the sum of all the tones that land in bucket $j$, each
scaled by a phase shift depending on $t_0$.

Because of this phase shift, for each choice of $t_0$ the value of
$\wh{u}_j$ is effectively a sample from the Fourier transform of a
signal that contains only the tones of $\wh{x}^*$ that land in bucket
$j$, with zeros elsewhere.  And since there are $k$ tones and
$O(k)$ buckets, most tones are alone in their bucket.  Therefore
this sampling strategy reduces the original problem of $k$-sparse
recovery to one of $1$-sparse recovery---we simply choose $t_0$
according to some strategy that lets us achieve $1$-sparse recovery,
and recover a tone for each bin.

\paragraph{One-sparse recovery.} The algorithm for one-sparse recovery
in~\cite{HIKP} is a good choice for adaptation to the continuous
setting.  It narrows down to the frequency in a locality-aware way,
maintaining an interval of frequencies that decreases in size at each
stage (in contrast to the method in~\cite{GMS}, which starts from the
least significant bit rather than most significant bit).

If a frequency is perturbed slightly in time (e.g., by multiplying by
a very slow decay over time) this will blur the frequency slightly
into a narrow band.  The one-sparse recovery algorithm of~\cite{HIKP}
will proceed normally until it gets to the narrow scale, at which
point it will behave semi-arbitrarily and return something near that
band.  This gives a desired level of robustness---the error in the
recovered frequency will be proportional to the perturbation.

Still, to achieve our result we need a few changes to the one-sparse
algorithm.  One is related to the duration $T$: in the very last stage
of the algorithm, when the interval is stretched at the maximal
amount, we can only afford one ``fold'' rather than the typical
$O(\log n)$.  The only cost to this is in failure probability, and
doing it for one stage is fine---but showing this requires a
different proof.  Another difference is that we need the final
interval to have precision $\frac{1}{T \rho}$ if the signal-to-noise
ratio is $\rho$---the previous analysis showed $\frac{1}{T
  \sqrt{\rho}}$ and needed to be told $\rho$, but (as we shall see) to
achieve an $\ell_2/\ell_2$ guarantee we need the optimal
$\rho$-dependence and for the algorithm to be oblivious to the value
of $\rho$.  Doing so requires a modification to the algorithm and
slightly more clever analysis.

\paragraph{$k$-sparse recovery.} The changes to the $k$-sparse
recovery structure are broader.  First, to make the algorithm simpler
we drop the~\cite{GLPS}-style recursion with smaller $k$, and just
repeat an $O(k)$-size hashing $O(\log k)$ times.  This loses a $\log
k$ factor in time and sample complexity, but because of the other
changes it is not easy to avoid, and at the same time improves our
success probability.

The most significant changes come because we can no longer measure the
noise in frequency space or rely on the hash function to randomize the
energy that collides with a given heavy hitter.  Because we only look
at a bounded time window $T$, Parseval's identity does not hold and
the energy of the noise in frequency space may be unrelated to its
observed energy.  Moreover, if the noise consists of frequencies
infinitesimally close to a true frequency, then because $\sigma$ is
bounded the true frequency will always hash to the same bin as the
noise.  These two issues are what drive the restrictions on noise in
the previous work~\cite{BCGLS}---assuming the noise is bounded in
$\ell_1$ norm in frequency domain and is zero in a neighborhood of the
true frequencies fixes both issues.  But we want a guarantee in terms
of the average $\ell_2$ noise level $\N^2$ in time domain over the
observed duration.  If the noise level is $\N^2$, because we cannot
hash the noise independently of the signal, we can only hope to
guarantee reliable recovery of tones with magnitude larger than
$\N^2$.  This is in contrast to the $\N^2/k$ that is possible in the
discrete setting, and would naively lose a factor of $k$ in the
$\ell_2/\ell_2$ approximation.

The insight here is that, even though the noise is not distributed
randomly across bins, the total amount of noise is still bounded.  If
a heavy hitter of magnitude $v^2$ is not recovered due to noise, that
requires $\Omega(v^2)$ noise mass in the bin that is not in any other
bin.  Thus the total amount of signal mass not recovered due to noise
is $O(\N^2)$, which allows for $\ell_2/\ell_2$ recovery.

This difference is why our algorithm only gets a constant factor
approximation rather than the $1+\eps$ guarantee that hashing
techniques for sparse recovery can achieve in other settings.  These
techniques hash into $B = O(k/\eps)$ bins so the average noise per bin
is $O(\frac{\eps}{k}\N^2)$.  In our setting, where the noise is not
hashed independently of the signal, this would give no benefit.

Another difference arises in the choice of the parameter $\sigma$,
which is the separation between samples in the arithmetic sequence
used for a single hashing, and gives the permutation during hashing.
In the discrete setting, one chooses $\sigma$ uniformly over $n$,
which in our setting would correspond to a scale of $\sigma \approx
\frac{1}{\eta}$.  Since the arithmetic sequences have $O(k \log
(k/\delta))$ samples, the duration would then become at least $\frac{k
  \log (k/\delta)}{\eta}$ (which is why~\cite{BCGLS} has this
duration).  What we observe is that $\sigma$ can actually be chosen at
the scale of $\frac{1}{k \eta}$, giving the desired
$O(\frac{\log(k/\delta)}{\eta})$ duration.  This causes frequencies at
the minimum separation $\eta$ to always land in bins that are a
constant separation apart.  This is sufficient because we
use~\cite{HIKP}-style window functions with strong isolation
properties (and, in fact,~\cite{HIKP} could have chosen $\sigma
\approx n/B$); it would be an issue if we were using the window
functions of~\cite{GMS,IK} that have smaller supports but less
isolation.

\paragraph{Getting an $\ell_2$ bound} Lastly, converting the
guarantee~\eqref{eq:2} into~\eqref{eq:3} is a nontrivial task that is
trivial in the discrete setting.  In the discrete setting, it follows
immediately from the different frequencies being orthogonal to each
other.  In our setting, we use that the recovered frequencies should
themselves have $\Omega(\eta)$ separation, and that well-separated
frequencies are nearly orthogonal over long enough time scales $T \gg
1/\eta$.

This bears some similarity to issues that arise in sparse recovery
with overcomplete dictionaries.  It would be interesting to see
whether further connections can be made between the problems.



\section{Proof outline}

In this section we present the key lemmas along the path to producing the algorithm. The full proof are presented in appendix.

\paragraph{Notation.} First we define the notation necessary to
understand the lemmas.  The full notation as used in the proofs
appears in Section~\ref{sec:notation}.

The algorithm proceeds in stages, each of which hashes the frequencies
to $B$ bins.  The hash function depends on two parameters $\sigma$ and
$b$, and so we define it as $h_{\sigma, b}(f): [-F, F] \to [B]$.

A tone with a given frequency $f$ can have two ``bad events''
$E_{\mathit{coll}}(f)$ or $E_{\mathit{off}}(f)$ hold for a given
hashing.  These correspond to colliding with another frequency of
$x^*$ or landing within an $\alpha$ fraction of the edge,
respectively; they each will occur with small constant probability.

For a given hashing, we will choose a number of different offsets $a$
that let us perform recovery of the tones that have neither bad
event in this stage.

We use $f \lesssim g$ to denote that there exists a constant $C$ such
that $f \leq C g$, and $f \eqsim g$ to denote $f \lesssim g \lesssim
f$.

\paragraph{Key Lemmas} First, we need to be able to compare the
distance between two pure tone signals in time domain to their
differences in parameters.  The relation is as follows:

\define{lem:two_close_signal}{Lemma}{%
  Let $(v,f)$ and $(v',f')$ denote any two tones, i.e., (magnitude,
  frequency) pairs.  Then for
\begin{align*}
\mathrm{dist}\left((v,f),(v',f')\right)^2 := \frac{1}{T} \int_0^T \left| v e^{2\pi f t\i} - v' e^{2\pi f' t\i } \right|^2\mathrm{d} t,
\end{align*}
we have
\begin{align*}
 \mathrm{dist}\left((v,f),(v',f')\right)^2 &\eqsim (|v|^2 + |v'|^2)\cdot \min (1,T^2\abs{f - f'}^2) + |v -v'|^2,
\end{align*}
and
\begin{align*}
\mathrm{dist}\left((v,f),(v',f')\right) &\eqsim |v| \cdot \min (1,T\abs{f - f'}) + |v -v'|.
\end{align*}
}
\state{lem:two_close_signal}

The basic building block for our algorithm is a function
$\mathsf{HashToBins}$, which is very similar to one of the same name
in~\cite{HIKP}.

The key property of $\mathsf{HashToBins}$ is that, if neither ``bad''
event holds for a frequency $f$ (i.e. it does not collide or land near the boundary
 of the bin), then for the bin $j = h_{\sigma, b}(f)$ we
have that $\abs{\wh{u}_j} \approx \abs{\widehat{x^*}(f)}$ with a phase
depending on $a$.

How good is the approximation?  In the discrete setting, one can show
that each tone has error about $\N^2/B$ in expectation.  Here, because
the hash function cannot randomize the noise, we instead show that the
total error over all tones is about $\N^2$:

\begin{lemma}\label{lem:hbinserror}
  Let $\sigma\in [\frac{1}{B\eta},\frac{2}{B\eta}]$ uniformly at
  random, then $b\in [0,\frac{\lceil F/\eta\rceil}{\sigma B}]$, $a \in
  [0, \frac{cT}{\sigma}]$ be sampled uniformly at random for some
  constant $c > 0$. Let the other parameters be arbitrary in
  $\widehat{u} = \mathsf{HashToBins} (x,P_{\sigma,a,b},
  B,\delta,\alpha)$, and consider $H = \{ f \in \supp(\wh{x}^*) \mid
  \text{neither} ~ E_{\mathit{coll}}(f) \text{~nor~} E_{\mathit{off}}(f) \text{~holds}\}$ and $I =
  [B] \setminus h_{\sigma, b}(\supp(\wh{x}^*))$ to be the bins that
  have no frequencies hashed to them.  Then
\begin{equation*}
\E_{\sigma, b, a}\left[ \sum_{f \in H}\left| \widehat{u}_{h_{\sigma,b}(f)} - \widehat{x^*}(f) e^{a\sigma  2\pi f \i}\right|^2 + \sum_{j \in I} \wh{u}_j^2 \right]  \lesssim \N^2
\end{equation*}
\end{lemma}

We prove Lemma~\ref{lem:hbinserror} by considering the cases of $x^* =
0$ and $g = 0$ separately; linearity then gives the result.  Both
follow from properties of our window functions.

\define{lem:expectation_lemma_when_x_star_is_zero}{Lemma}{
If $x^*(t) = 0, \forall t \in [0,T]$, then 

\begin{equation*}
\underset{\sigma,a,b}{ \mathbb{E} } [  \sum_{j=1}^B |\widehat{u}_j|^2] \lesssim \frac{1}{T} \int_0^T |g(t)|^2 \mathrm{d} t.
\end{equation*}
}
\state{lem:expectation_lemma_when_x_star_is_zero}

\define{lem:expectation_lemma_when_g_is_zero}{Lemma}{
If $g(t) = 0, \forall t \in [0,T] $. Let $H$ denote a set of frequencies, $H=\{f \in \supp(\wh{x}^*) \mid \mathrm{neither} ~
  E_{\mathit{coll}}(f) $ $~\mathrm{nor}~ E_{\mathit{off}}(f) \mathrm{~holds}\}$. 
Then,
\begin{equation*}
\underset{\sigma,a,b}{ \mathbb{E} } [\sum_{f\in H}  \left| \widehat{u}_{h_{\sigma,b}(f)} - \widehat{x^*}(f) e^{a\sigma 2\pi f \i}\right|^2] \leq \delta \| \widehat{x^*}\|_1^2. 
\end{equation*}
}
\state{lem:expectation_lemma_when_g_is_zero}

Lemma~\ref{lem:hbinserror} is essentially what we need for $1$-sparse
recovery.  We first show a lemma about the inner call, which narrows
the frequency from a range of size $\Delta l$ to one of size
$\frac{\Delta l}{\rho s t}$ for some parameters $\rho s t$.  This
gives improved performance (superresolution) when the signal-to-noise
ratio $\rho$ within the bucket is high.  The parameter $s$ and $t$
provide a tradeoff between success probability, performance,
running time, and duration.

\begin{lemma}
  Consider any $B, \delta, \alpha$.  Algorithm $\mathsf{HashToBins}$
  takes $O(B\log (k/\delta))$ samples and runs in $O(\frac{B}{\alpha}\log
  (k/\delta) + B\log B)$ time.
\end{lemma}

\define{lem:correctness_of_locateinner}{Lemma}{
  Given $\sigma$ and $b$, consider any frequency $f$ for which neither 
   $E_{\mathit{coll}}(f)$ nor $E_{\mathit{off}}(f)$ holds, and let $j =
  h_{\sigma,b}(f)$. Let $\mu^2(f) = \mathbb{E}_{a}[|\widehat{u}_j -
  \widehat{x^*}(f)e^{a\sigma 2\pi f\i} |^2 ]$ and $\rho^2 =
  |\widehat{x^*}(f) |^2 / \mu^2(f)$. 
  For sufficiently large $\rho$, and $\forall 0<s<1,t \geq 4$, consider any run of
  $\mathsf{LocateInner}$ with $f \in [l_j -\frac{\Delta l}{2}, l_j+
  \frac{\Delta l}{2}]$. It takes $O(R_{loc})$ random $(\gamma,\beta)
  \in [\frac{1}{2},1] \times [\frac{s t}{4 \sigma \Delta l},\frac{s
     t}{2 \sigma \Delta l}]$ samples over duration $\beta \sigma =
  \Theta(\frac{st}{\Delta l})$, runs in $O(st R_{loc})$ time, to
  learn $f$ within a region that has length $ \Theta( \frac{ \Delta l}{t} )$ 
with failure probability at most $(\frac{4}{s\rho})^{R_{loc}} + t\cdot (60 s)^{R_{loc}/2}$.
}
\state{lem:correctness_of_locateinner}

By repeating this inner loop, we can recover the tones in almost
every bin that does not have the ``bad'' events happen, so we recover a
large fraction of the heavy hitters in each stage. 

\define{lem:correctness_of_locateouter}{Lemma}{
  Algorithm $\mathsf{LocateKSignal}$ takes $O(k\log_{C} (FT)
  \log(k/\delta))$ samples over $O( \frac{\log(k/\delta)}{\eta})$
  duration, runs in $O(k \log_{C} (FT) \log(FT/\delta) )$ time, and
  outputs a set $L \subset [-F, F]$ of $O(k)$ frequencies with minimum
  separation $\Omega(\eta)$.

  Given $\sigma$ and $b$, consider any frequency $f$ for which neither
  of $E_{\mathit{coll}}(f)$ or $E_{\mathit{off}}(f)$ hold. Let $j = h_{\sigma,b}(f)$,
  $\mu^2(f) = \mathbb{E}_{a}[|\widehat{u}_j -
  \widehat{x^*}(f)e^{a\sigma 2\pi f\i} |^2 ]$, and $\rho^2 =
  |\widehat{x^*}(f) |^2 / \mu^2(f)$. If $\rho > C$, then with an
  arbitrarily large constant probability there exists an $f' \in L$
  with
  \[
  \abs{f - f'} \lesssim \frac{1}{T \rho}.
  \]
}
\state{lem:correctness_of_locateouter}

Combining this with estimation of the magnitudes of recovered
frequencies, we can show that the total error over all bins without
``bad'' events---that is, bins with either one well placed frequency
or zero frequencies---is small.  At this point we give no guarantee
for the (relatively few) bins with bad events; the recovered values
there may be arbitrarily large.

\define{lem:whole_of_stage}{Lemma}{
  Algorithm $\mathsf{OneStage}$ takes $O(k\log_{C} (FT)
  \log(k/\delta))$ samples over $O( \frac{\log(k/\delta)}{\eta})$
  duration, runs in $O(k (\log_{C} (FT) \log(FT/\delta) ))$ time, and
  outputs a set of $\{ (v_i',f_i')\}$ of size $O(k)$ with $\min_{i
    \neq j} \abs{f'_i - f'_j} \gtrsim \eta$.  Moreover, one can
  imagine a subset $S \subseteq [k]$ of ``successful'' recoveries,
  where $\mathsf{Pr} [i \in S] \geq \frac{9}{10}~ \forall i \in [k]$
  and for which there exists an injective function $\pi : [k]
  \rightarrow [O(k)]$ so that
  \begin{equation*}
  \underset{\sigma,b}{\mathbb{E}}\left[ \sum_{i\in S} \frac{1}{T} \int_0^T \left|v_i' e^{2\pi f_i' t \i } - v_{\pi(i)} e^{2\pi f_{\pi(i)}t\i } \right|^2 \mathrm{d} t  \right] \lesssim C^2 {\cal N}^2.
  \end{equation*}
  with $1 - 1/k^c$ probability for an arbitrarily large constant $c$.
}

\state{lem:whole_of_stage}

We can repeat the procedure for $O(\log k)$ stages and merge the
results, getting a list of $O(k)$ tones that includes $k$ tones that
match up well to the true tones.  However, we give no guarantee for
the rest of the recovered tones at this point---as far as the
analysis is concerned, mistakes from bins with collisions may cause
arbitrarily large spurious tones.
\define{lem:onemerged}{Lemma}{ Repeating algorithm $\mathsf{OneStage}$
  $O(\log k)$ times, $\mathsf{MergedStages}$ returns a set $\{
  (v_i',f_i')\}$ of size $O(k)$ with $\min_{i \neq j} \abs{f'_i -
    f'_j} \gtrsim \eta$
  that can be indexed by $\pi$ such that
\begin{equation*}
\sum_{i=1}^k \frac{1}{T} \int_0^T \left|v_i' e^{2\pi f_i' t \i } - v_{\pi(i)} e^{2\pi f_{\pi(i)}t\i } \right|^2 \mathrm{d} t \lesssim C^2 {\cal N}^2.
\end{equation*}
with probability $1-1/k^c$ for an arbitrarily large constant $c$.
}
\state{lem:onemerged}

To address the issue of spurious tones, we run the above
algorithm twice and only take the tones that are recovered in
both stages.  We show that the resulting $O(k)$ tones are
together a good approximation to the vector.
\define{lem:prune_twice}{Lemma}{
  If we run $\mathsf{MergedStages}$ twice and take the
  tones $\{(v'_i, f'_i)\}$ from the first result that have
  $f'_i$ within $c\eta$ for small $c$ of some frequency in the second
  result, we get a set of $k''=O(k)$ tones that can be indexed by 
  some permutation $\pi$ such that
\begin{align}
\sum_{i=1}^k \frac{1}{T} \int_0^T \left|v_i' e^{2\pi f_i' t \i } - v_{\pi(i)} e^{2\pi f_{\pi(i)}t\i } \right|^2 \mathrm{d} t + \sum_{i=k+1}^{k''} \abs{v'_i}^2 \lesssim C^2 {\cal N}^2.\thmlabel{eq:prune_twice}
\end{align}
}
\state{lem:prune_twice}

Simply picking out the largest $k$ recovered tones then gives
the result~\eqref{eq:2}.
\define{thm:get_eq2}{Theorem}{
  Algorithm $\mathsf{ContinuousFourierSparseRecovery}$ returns a set $\{ (v_i',f_i')\}$ of
  size $k$ with $\min_{i \neq j} \abs{f'_i - f'_j} \gtrsim \eta$ for which
\begin{equation*}
\sum_{i=1}^k \frac{1}{T} \int_0^T \left|v_i' e^{2\pi f_i' t \i } - v_{\pi(i)} e^{2\pi f_{\pi(i)}t\i } \right|^2 \mathrm{d} t \lesssim C^2 {\cal N}^2
\end{equation*}
with probability $1 - 1/k^c$ for an arbitrarily large constant $c$.
}
\state{thm:get_eq2}

By only considering the term in the sum corresponding to tone $i$
and applying Lemma~\ref{lem:two_close_signal}, we get result~\eqref{eq:1}:

\begin{corollary}\label{cor:max_error}
  With probability $1-1/k^c$ for an arbitrarily large constant $c$, 
  we recover a set of
  tones $\{ (v_i',f_i') \}$ such that, for any $v_i$ with
  $|v_i|\gtrsim \N$, we have for an appropriate permutation of the
  indices that
\begin{align}
  \abs{f'_i - f_i} &\lesssim \frac{\N}{T\abs{v_i}}  & \text{and} &\quad&
  \abs{v'_i - v_i} \lesssim \N.
\end{align}
\end{corollary}

We then show that~\eqref{eq:2} implies~\eqref{eq:3} for sufficiently
long durations $T$.  A long duration helps because it decreases the
correlation between $\eta$-separated frequencies.

\define{lem:choose_longer_duration}{Lemma}{
  Let $\{(v_i,f_i)\}$ and $\{(v'_i,f'_i)\}$ be two sets of $k$ tones
  for which $\min_{i\neq j} |f_i-f_j| \geq \eta$ and $\min_{i\neq j}
  |f'_i - f'_j| \gtrsim \eta$ for some $\eta > 0$.  Suppose that $T >
  O(\frac{\log^2 k}{\eta})$.  Then these sets can be
  indexed such that
\begin{equation}
\frac{1}{T} \int_0^T | \sum_{i=1}^k  (v'_i e^{2\pi \i f'_i t} - v_i e^{2\pi \i f_i t})|^2 \mathrm{d} t \lesssim \sum_{i=1}^k \frac{1}{T}\int_{0}^T \abs{v'_i e^{2\pi \i f'_i t} - v_i e^{2\pi \i f_i t}}^2\mathrm{d}t.
\end{equation}
}
\state{lem:choose_longer_duration}

Combining Theorem \ref{thm:get_eq2} and Lemma \ref{lem:choose_longer_duration} immediately implies

\begin{theorem}
  Suppose we sample for a duration $T>O(\frac{\log(1/\delta)+\log^2
    k}{\eta})$.  Then the reconstructed signal $x'(t) = \sum_{i=1}^k
  v'_i e^{2\pi \i f'_i t}$ achieves a constant factor
  approximation to the complete signal $x$:
\begin{equation}
\frac{1}{T} \int_0^T |x'(t) -x(t)|^2 \mathrm{d} t \lesssim C^2 \N^2.
\end{equation}
The algorithm takes $O(k\log \frac{F}{\eta}\log(\frac{k}{\delta})
\log(k))$ samples, runs in $O(k\log
\frac{F}{\eta}\log(\frac{FT}{\delta}) \log(k))$ time, and succeeds with
probability  at least $1-1/k^c$ for an arbitrarily
large constant $c$.
\end{theorem}

That finishes the proof of our main theorem.  We also show that our
``superresolution'' precision from~\eqref{eq:1} is optimal, which is a
simple corollary of Lemma~\ref{lem:two_close_signal}.

\define{lem:lower}{Lemma}{
  There exists a constant $c > 0$ such that, for a given sample duration
  $T$, one cannot recover the frequency $f$ to within 
  \[
  c\frac{\N}{T \abs{\wh{x}^*(f)}}
  \]
  with $3/4$ probability, for all $\delta > 0$, even if $k=1$.
}
\state{lem:lower}
\newpage
\bibliographystyle{alpha}
\bibliography{ref} 

\newpage
\appendix

\section{Notation and Definitions: Permutation, Hashing, Filters}\label{sec:notation}

This section gives definitions about the permutation, hashing, and
filters that are used throughout the proofs. Let $[n]$ denote the set $\{1,2,\cdots,n-1,n\}$. $\mathbb{R}$ denotes the real numbers, $\mathbb{Z}$ denotes the integer numbers and $\mathbb{C}$ denotes the complex numbers. The convolution of two continuous functions $f$ and $g$ is written as $f * g$,
\begin{equation*}
(f* g) (t) := \int_{-\infty}^{+\infty} f(\tau) g(t-\tau) \mathrm{d} \tau
\end{equation*}
and the discrete convolution of $f$ and $g$ is given by,
\begin{equation*}
(f*g) [n] := \sum_{m=-\infty}^{+\infty} f[m] g[n-m]
\end{equation*}
Let $\i$ denote $\sqrt{-1}$, and $e^{\i \theta} = \cos(\theta) +\i \sin(\theta)$. For any complex number $z\in \mathbb{C}$, we have $ z = a+ \i b$, where $a,b\in \mathbb{R}$. Define $\overline{z} = a - \i b$, $|z|^2 = z \overline{z} = a^2 + b^2$ and let $\phi(z)$ be the phase of $z$. Let $\supp(f)$ denote the support of function/vector $f$, and $\| f\|_0 = |\supp(f)|$. For any $p\in [1,\infty]$, the $\ell_p$ norm of a vector of $x$ is $\| x \|_p = (\sum_i |x_i|^p )^\frac{1}{p}$, defined to be $\max_i |x_i|$ for $p=\infty$.  Let $k$ denote the sparsity of frequency domain. All the frequencies $\{f_1,f_2,\cdots,f_k\}$ are from $[-F, F]$. Let $B=O(k)$ denote the number of hash bins in our algorithm.

We translate the ``permutation'' $P_{\sigma, a, b}$ of \cite{HIKP}
from the DFT setting to the DTFT setting.

\begin{definition}
$(P_{\sigma,a,b} x^*)(t) = x^*(\sigma(t-a)) e^{-2\pi \i\sigma bt}$.
\end{definition}

\begin{lemma}
$\widehat{P_{\sigma,a,b} x^*} (\sigma(f-b)) = e^{-2\pi f \sigma a \i } \widehat{x}^*(f)$
\end{lemma}

\begin{proof}
  The time domain representation of the given Fourier definition would be
  \begin{align*}
    x'(t) &= \sum_{f \in \supp(\wh{x}^*)} e^{-2\pi \i \sigma a f} \wh{x}^*(f) e^{2\pi \i \sigma(f - b) t}\\
    &= \sum_{f \in \supp(\wh{x}^*)} e^{-2\pi \i \sigma b t} \wh{x}^*(f) e^{2\pi \i \sigma(t-a)f}\\
    &= e^{-2\pi \i \sigma b t}  x^*(\sigma(t-a))
  \end{align*}
  which matches, so the formula is right.
\end{proof}

We also extend the flat window function for the DFT setting \cite{HIKP}, \cite{HIKP12}
to the DTFT setting:

\begin{definition}
Let $M=O(B\log \frac{k}{\delta})$. We say that $(G, \widehat{G'}) = (G_{B,\delta,\alpha}, \widehat{G'}_{B,\delta,\alpha}) \in \mathbb{R}^{M} \times \mathbb{R}^{[-F,F]}$ is a flat window function with parameters $B\geq 1$, $\delta >0$, and $\alpha>0$. For simplicity, let's say $B$ is a function of $\alpha$. Define $|\mathrm{supp}(G)| =M $ and $\widehat{G'}$ satisfies
\begin{itemize}
\item $G_i = \frac{\sin(i\frac{1}{B} )}{i} \cdot e^{-\frac{i^2}{2\sigma^2}}$, where $\sigma = \Theta(B \sqrt{\log(k/\delta)})$.
\item $\widehat{G} (f) =  \overset{M} { \underset  {i=1}{ \sum} } G_i e^{f\cdot\frac{i}{M} 2\pi \i}$.
\item $\mathrm{supp}(\widehat{G'} ) \subset [-\frac{2\pi}{2B}, \frac{2\pi}{2B}]$.
\item $\widehat{G'}(f) =1$ for all $f \in [-\frac{(1-\alpha)2\pi}{2B}, \frac{(1-\alpha)2\pi}{2B}]$.
\item $\widehat{G'}(f) =0$ for all $ |f| \geq \frac{2\pi}{2B}$.
\item $\widehat{G'}(f) \in [0,1]$ for all $f$.
\item $\| \widehat{G'} -\widehat{G} \|_\infty^2 < \delta/k$.
\end{itemize}
\end{definition}

\begin{claim}\label{cla:l2_of_discrete_G}
$ \overset{M} { \underset  {i=1}{ \sum} } G_i^2 \eqsim \frac{1}{B}$, where $M = O(B\log k/\delta)$.
\end{claim}

\begin{proof}
By definition of $G_i$, we have
\begin{equation*}
\sum_{i=1}^M G_i^2  = 2 \sum_{i=1}^{M/2} \frac{\sin^2(i \frac{1}{B})}{ (i)^2}  (  e^{- \frac{ (i)^2}{2\sigma^2}} )^2.
\end{equation*}
There exists some constant $c\in [0,2\pi)$, such that $\frac{\sin(i/B) }{i} \eqsim \frac{1}{B} $ if $i/B < c\pi$.
\begin{eqnarray*}
\sum_{i=1}^M G_i^2 & = &  2 \sum_{i=1}^{\lfloor Bc\pi \rfloor } \frac{\sin^2(i \frac{1}{B})}{ (i)^2}  (  e^{- \frac{ (i)^2}{2\sigma^2}} )^2 + 2 \sum_{i= \lceil Bc\pi \rceil }^{M/2} \frac{\sin^2(i \frac{1}{B})}{ (i)^2}  ( e^{- \frac{ (i)^2}{2\sigma^2}} )^2 \\
& \leq & 2 \sum_{i=1}^{\lfloor Bc\pi \rfloor } \frac{\sin^2(i \frac{1}{B})}{ (i)^2} \cdot 1 + 2 \sum_{i=\lceil Bc\pi \rceil }^{M/2} \frac{\sin^2(i \frac{1}{B})}{ (i)^2} \cdot 1\\
& \lesssim  & \sum_{i=1}^{\lfloor Bc\pi \rfloor }  \frac{1}{B^2} + \sum_{i=\lceil Bc\pi \rceil }^{M/2} \frac{1}{i^2} \\
& \lesssim & \frac{1}{B}.
\end{eqnarray*}
Thus, we show an upper bound. It remains to prove the lower bound.
\begin{eqnarray*}
\sum_{i=1}^M G_i^2 & \geq &  2 \sum_{i=1}^{\lfloor Bc\pi \rfloor } \frac{\sin^2(i \frac{1}{B})}{ (i)^2}  (  e^{- \frac{ (i)^2}{2\sigma^2}} )^2 \\
& \geq &  2 \sum_{i=1}^{\lfloor Bc\pi \rfloor } \frac{\sin^2(i \frac{1}{B})}{ (i)^2}  (  e^{- \frac{ (Bc\pi)^2}{2\sigma^2}} )^2 \\
& \gtrsim &  2 \sum_{i=1}^{\lfloor Bc\pi \rfloor } \frac{1}{B^2}  (  e^{- \frac{ (Bc\pi)^2}{2\sigma^2}} )^2 \\
& \gtrsim &  \frac{1}{B} (  e^{- \frac{ (Bc\pi)^2}{2\sigma^2}} )^2 \\
& \gtrsim &  \frac{1}{B} (  e^{-c_0} )^2 .
\end{eqnarray*}
The last inequality follows by there exists some universal constant $c_0>0$ such that $-\frac{1}{\log(k/\delta)} \gtrsim - c_0$. Thus, we show $\sum_{i=1}^M  G_i^2 \gtrsim \frac{1}{B}$.
\end{proof}

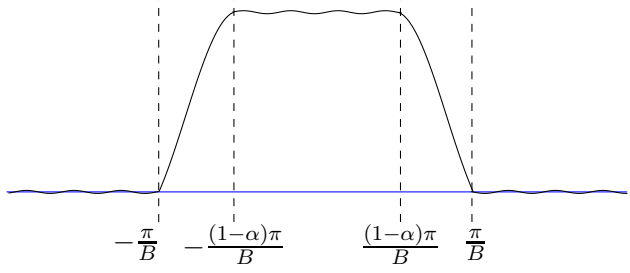
\begin{figure}
\centering
\begin{tikzpicture}
\def\PI{3.1415926}
\def\s{3}
\def\height{0.8}
\def\rh{\height*sin(\s* 0.05 r)/0.05 }
\def\rc{1}
\def\rightposition{0.05+\rc+ 2*\PI/40}
\def\leftposition{ -0.05-\rc- 2*\PI/40}

    	\draw [blue] (-4.25,0) --(4,0);
    	\draw [dashed] (\leftposition-0.02 , -0.4) -- (\leftposition-0.02 , 2.6);
    	\draw [dashed] (\rightposition-0.22 , -0.4) -- (\rightposition-0.22 , 2.6);

    	\draw [dashed] (\leftposition-0.02 -1.0 , -0.4) -- (\leftposition-0.02 -1.0, 2.6);
    	\draw [dashed] (\rightposition-0.22 +0.95 , -0.4) -- (\rightposition-0.22 +0.95, 2.6);

    	\node at (\leftposition-0.02,-0.7) {$-\frac{(1-\alpha)\pi}{B}$};
    	\node at (\rightposition-0.22,-0.7) {$\frac{(1-\alpha)\pi}{B}$};

    	\node at (\leftposition-0.02-\rc-0.3,-0.7) {$-\frac{\pi}{B}$};
    	\node at (\rightposition-0.22+\rc,-0.7) {$\frac{\pi}{B}$};
    	\draw [domain={\leftposition-0.04}:{\rightposition-0.2},variable=\t,smooth,samples=500]
    	plot ( {\t },{\rh + 0.02* sin( 10* \t r)});
    	\draw [domain=-0.05:{-\PI/\s},variable=\t,smooth,samples=500]
    	plot ( {\t - \leftposition-0.07},{\height*sin(\s* \t r)/\t});
    	\draw [domain=0.05:{\PI/\s},variable=\t,smooth,samples=500]
    	plot ( {\t + \rightposition-0.3},{\height*sin(\s* \t r)/\t});

    	\draw [domain={-\PI/\s-\leftposition-1.7}:{-\PI/\s-\leftposition+0.3},variable=\t,smooth,samples=500]
    	plot ( {\t-0.37 },{ 0.02* sin( 10* \t r)});

    	\draw [domain={\PI/\s+\rightposition+2}:{\PI/\s+\rightposition-0.08},variable=\t,smooth,samples=500]
    	plot ( {\t -0.25 },{ 0.02* sin( 10* \t r)});

\end{tikzpicture}
\caption{Filter $\widehat{G'}(f)$ }\label{fig:filter}
\end{figure}

To analyze the details of our algorithm, we explain some lower-level definitions and claims first. Here we give the definition of three notations that related to hash function.

\begin{definition}
$\pi_{\sigma,b} (f) = 2\pi \sigma(f-b) \pmod {2\pi }$. We maps frequency to a circle $[0,2\pi)$, since our observation of sample is the phase of some complex number, which also belongs to $[0,2\pi)$.
\end{definition}

\begin{definition}
$h_{\sigma,b}(f) = \mathrm{round}( \pi_{\sigma,b} (f) \cdot \frac{B}{2\pi})$. $h_{\sigma,b}(f)$ is a ``hash function'' that hashes frequency $f$ into one of the $B$ bins. The motivation is, it is very likely that each bin only has 1 heavy hitters if we choose large enough $B$. Then, for each bin, we can run a $1$-sparse algorithm to recover the frequency.
\end{definition}

\begin{definition}
$o_{\sigma,b}(f) = \pi_{\sigma,b} (f)- \frac{2\pi}{B} \cdot h_{\sigma,b}(f)$. Offset $o_{\sigma,b}(f)$ denotes the distance from $\pi_{\sigma,b}(f)$ to the center of the corresponding bin that frequency $f$ was hashed into.
\end{definition}

Then we define some events that might happen after applying hash function to the entire frequency domain.
\begin{definition}
``Collision'' event $E_{\mathit{coll}}(f)$: holds iff $h_{\sigma,b}(f) \in h_{\sigma,b}(\mathrm{supp}(\widehat{x^*}) \backslash \{f\})$. The ``collision'' event happening means there exists some other frequency $f'$ such that both $f$ and $f'$ are hashed into the same bin. Once two frequencies are colliding in one bin, the algorithm will not be able to recover them.
\end{definition}

\begin{definition}
``Large offset'' event $E_{\mathit{off}}(f)$: holds iff $|o_{\sigma,b}(f)| \geq (1-\alpha) \frac{2\pi}{2B}$. 
The event holds if frequency $f$ is not within factor $1-\alpha$ of the radius close to the center of that hash bin. It causes the frequency to be in the intermediate regime of filter and not recoverable, see Figure \ref{fig:filter}.
\end{definition}

\begin{definition}\label{def:main_sample}
We sample $\sigma$ uniformly at random from $[\frac{1}{B\eta}, \frac{2}{B\eta}]$. Conditioning on $\sigma$ is chosen first, we sample $b$ uniformly at random from $[0, \frac{\lceil F/\eta \rceil}{B\sigma}]$. Then we sample $\gamma $ uniformly at random from $[\frac{1}{2},1]$ and $\beta$ uniformly at random from $[\widehat{\beta},2\widehat{\beta}]$, where $\widehat{\beta}$ is dynamically changing during our algorithm(The details of setting $\widehat{\beta}$ are explained in Lemma \ref{lem:correctness_of_locateinner}).
For $P_{\sigma,\gamma,b}$ and $P_{\sigma,\gamma+\beta,b}$, we take the following two sets of samples over time domain,
\begin{eqnarray*}
&&x(\sigma(1-\gamma)), x(\sigma (2- \gamma)) , x(\sigma(3-\gamma)), \cdots, x( \sigma( B\log(k/\delta) - \gamma) ) \\
&&x(\sigma(1-\gamma -\beta)), x(\sigma (2- \gamma- \beta)) , x(\sigma(3-\gamma -\beta)), \cdots, x( \sigma( B\log(k/\delta) - \gamma- \beta) ) 
\end{eqnarray*}
\end{definition}


Conditioning on drawing $\sigma,b$ from some distribution, we are able to show that the probability of ``Collision'' and ``Large offset'' event holding are small.

\begin{lemma}\label{lem:wrapping}
For any $\widetilde{T}$, and $ 0 \leq \widetilde{\epsilon}, \widetilde{\delta} \leq \widetilde{T}  $, if we sample $\widetilde{\sigma}$ uniformly at random from $[A,2A]$, then
\begin{equation}
\frac{2\widetilde{\epsilon} }{ \widetilde{T} } - \frac{2\widetilde{\epsilon} }{ A } \leq \mathsf{Pr} \left[ \widetilde{\sigma}  {\pmod {\widetilde{T}} } \in [ \widetilde{\delta} - \widetilde{\epsilon}, \widetilde{\delta} + \widetilde{\epsilon} ~] \right] \leq \frac{2\widetilde{\epsilon} }{ \widetilde{T} } + \frac{4\widetilde{\epsilon} }{ A }.
\end{equation} 
\end{lemma}

\begin{proof}
Since we sample $\widetilde{\sigma}$ uniformly at random from $[A,2A]$,  let $I$ denote a set of candidate integers, then the smallest one is $\lfloor A \rfloor$ and the largest one is $\lceil 2A \rceil$. Thus, the original probability equation is equivalent to
\begin{equation}
\mathsf{Pr} \left[ \widetilde{\sigma}  \in [ s\cdot \widetilde{T} + \widetilde{\delta} - \widetilde{\epsilon},  s\cdot \widetilde{T}+ \widetilde{\delta} + \widetilde{\epsilon} ~] ~\exists s\in I\right],
\end{equation}
where $I = \{ \lfloor A  / \widetilde{T} \rfloor ,\cdots, \lceil 2A   / \widetilde{T} \rceil \}$.

Consider any $s\in I$, the probability of $\widetilde{\sigma} $ belonging to the interval $[s\cdot \widetilde{T} +\widetilde{\delta} -\widetilde{\epsilon}, s\cdot \widetilde{T} +\widetilde{\delta} +\widetilde{\epsilon}]$ is 
\begin{equation*}
\mathsf{Pr}\left[ \widetilde{\sigma}  \in [ s\cdot \widetilde{T} + \widetilde{\delta} - \widetilde{\epsilon},  s\cdot \widetilde{T}+ \widetilde{\delta} + \widetilde{\epsilon} ~] \right] = \frac{2\widetilde{\epsilon}}{A }.
\end{equation*}
Taking the summation over all $s\in I$, we obtain
\begin{eqnarray*}
& & \mathsf{Pr} \left[ \widetilde{\sigma}  \in [ s\cdot \widetilde{T} + \widetilde{\delta} - \widetilde{\epsilon},  s\cdot \widetilde{T}+ \widetilde{\delta} + \widetilde{\epsilon} ~] ~\exists s\in I\right] \\
& = & \sum_{s\in I} \mathsf{Pr}\left[ \widetilde{\sigma}  \in [ s\cdot \widetilde{T} + \widetilde{\delta} - \widetilde{\epsilon},  s\cdot \widetilde{T}+ \widetilde{\delta} + \widetilde{\epsilon} ~] \right] \\
& = & \sum_{s\in I} \frac{2\widetilde{\epsilon} }{A  } \\
& = & \frac{2\widetilde{\epsilon} |I| }{A  }.
\end{eqnarray*}
It remains to bound $\frac{2\widetilde{\epsilon} |I| }{A  }$.
Since $|I| = \lceil 2A  /\widetilde{T} \rceil - \lfloor A  /\widetilde{T} \rfloor +1$, then we have an upper bound for $I$,
\begin{equation*}
|I| \leq A  /\widetilde{T} +2.
\end{equation*}
On the other side, we have an lower bound,
\begin{equation*}
|I| \geq A  /\widetilde{T} -1.
\end{equation*}
Plugging upper bound of $I$ into $\frac{2\widetilde{\epsilon} |I|}{A }$, 
\begin{equation*}
\frac{2\widetilde{\epsilon} |I|}{A } \leq \frac{2\widetilde{\epsilon} }{A } (A /\widetilde{T} + 2) = \frac{2\widetilde{\epsilon} }{\widetilde{T}} + \frac{4\widetilde{\epsilon} }{ A  }.
\end{equation*}
Using the lower bound of $|I|$, we have 
\begin{equation*}
\frac{2\widetilde{\epsilon} |I|}{A } \geq \frac{2\widetilde{\epsilon} }{A }  (A  /\widetilde{T}-1) = \frac{2\widetilde{\epsilon} }{\widetilde{T}} - \frac{2\widetilde{\epsilon}}{A }.
\end{equation*}
Thus, we complete the proof.
\end{proof}
The following corollary will be used many times in this paper. The proof directly follows by Lemma \ref{lem:wrapping}.
\begin{corollary}\label{cor:wrapping}
For any $\widetilde{T}$, $\Delta f$, and $ 0 \leq \widetilde{\epsilon}, \widetilde{\delta} \leq \widetilde{T}  $, if we sample $\widetilde{\sigma}$ uniformly at random from $[A,2A]$, then
\begin{equation} 
\frac{2\widetilde{\epsilon} }{ \widetilde{T} } - \frac{2\widetilde{\epsilon} }{ A \Delta f } \leq \mathsf{Pr} \left[ \widetilde{\sigma} \Delta f {\pmod {\widetilde{T}} } \in [ \widetilde{\delta} - \widetilde{\epsilon}, \widetilde{\delta} + \widetilde{\epsilon} ~] \right] \leq \frac{2\widetilde{\epsilon} }{ \widetilde{T} } + \frac{4\widetilde{\epsilon} }{ A \Delta f }.
\end{equation}
\end{corollary}
\begin{proof}
Since $\widetilde{\sigma}$ is sampled uniformly at random from $[A,2A]$, then $\widetilde{\sigma}\Delta f$ is sampled uniformly at random from $[A \Delta f,2A \Delta f]$. Now applying Lemma \ref{lem:wrapping} by only replacing $A\Delta f$ by $A$.
\end{proof}

\begin{claim}\label{cla:sample_sigma}
Let $\sigma$ be sampled uniformly at random from $[\frac{1}{B\eta},\frac{2}{B\eta}]$ and $ \underset{{i\neq j } }{\min} |f_i -f_j| > \eta$.  $\forall i,j \in [k]$, if $ i \neq j$,  then $\mathsf{Pr} [h_{\sigma,b} (f_i) = h_{\sigma,b}(f_j)] \lesssim \frac{1}{B} $.
\end{claim}

\begin{proof}
To simplify the proof, define $\Delta f = |f_i - f_j|$.
We consider two cases: (\RN{1}) $\eta \leq |f_i -f_j| < \frac{(B-1)\eta}{2}$, (\RN{2}) $\frac{(B-1)\eta}{2} \leq |f_i - f_j|$.

(\RN{1})
If $\Delta f = \eta$, then $2\pi \sigma \Delta f$ is at least $\frac{2\pi}{\eta B} \cdot \eta = \frac{2\pi}{B}$, which means two frequencies have to go to different bins after hashing. If $\Delta f = \frac{(B-1)\eta}{2}$, then $2\pi \sigma \Delta f$ is at most $\frac{4\pi}{B\eta} \cdot \frac{(B-1)\eta}{2} = (1-1/B)2\pi$. In order to make two frequencies collide, $2\pi \sigma \Delta f$ should belong to $[(1-1/B)2\pi, (1+1/B)2\pi)$. Since for any  $ \Delta f \in [\eta, \frac{(B-1)\eta}{2} ) $, we have $2\pi \sigma \Delta f \in [\frac{1}{B}2\pi, (1-1/B)2\pi)$, which does not intersect interval $[(1-1/B)2\pi, (1+1/B)2\pi)$. Thus,
\begin{equation*}
\underset{\sigma,b}{\mathsf{Pr}} [h_{\sigma,b} (f_i) = h_{\sigma,b}(f_j)] =0.
\end{equation*}

(\RN{2}) 
We apply Corollary \ref{cor:wrapping} by setting $\widetilde{T} = 2\pi$, $\widetilde{\sigma} =2\pi \sigma$, $\widetilde{\delta} = 0$, $\widetilde{\epsilon} = \frac{2\pi}{2B}$, $A = 2\pi \frac{1}{B\eta}$. Then we have
\begin{eqnarray}\label{eq:prob_sigma_b_with_s}
& &\underset{\sigma,b}{\mathsf{Pr}} [h_{\sigma,b} (f_i) = h_{\sigma,b}(f_j)] =\underset{\sigma,b}{\mathsf{Pr}} \left[2\pi \sigma \Delta f \in [s\cdot 2\pi  -\frac{2\pi}{2B}, s\cdot 2\pi + \frac{2\pi}{2B}] ~\exists~s~\in~I \right],
\end{eqnarray}
where 
\begin{equation*}
I = \{  \lfloor  \frac{1}{B\eta} \Delta f \rfloor,\cdots, \lceil   \frac{2}{B\eta} \Delta f \rceil \}.
\end{equation*}
By upper bound of Corollary \ref{cor:wrapping}, Equation (\ref{eq:prob_sigma_b_with_s}) is at most
\begin{equation*}
\frac{1}{\frac{2\pi}{B\eta} \Delta f} \cdot \frac{2\pi}{B}  \cdot ( \frac{1 }{B\eta} \Delta f +2) =\frac{1}{B} + \frac{2\eta}{\Delta f} \leq \frac{1}{B} + \frac{4}{B-1} \lesssim \frac{1}{B},
\end{equation*}
where the first inequality follows by $\frac{(B-1)\eta}{2} \leq |f_i - f_j|$ which is the assumption of part (\RN{2}).

\end{proof}




\begin{claim}\label{cla:sample_b}
 $\forall f$, $\forall ~0<\alpha<1$, $\underset{\sigma,b}{\mathsf{Pr}} \left[ |o_{\sigma,b}(f)| \leq (1-\alpha) \frac{2\pi}{2B} \right]  \geq 1- O(\alpha)$.
\end{claim}

\begin{proof}
Since we draw $\sigma$ uniformly at random from $[\frac{1}{B\eta}, \frac{2}{B\eta}]$, then $2\pi \sigma (f-b) \in [\frac{2\pi}{B\eta}(f-b), \frac{4\pi}{B\eta}(f-b)]$ uniformly at random. The probability is equal to
\begin{equation*}
\underset{\sigma,b}{\mathsf{Pr}} \left[2\pi \sigma(f-b) \in [\cdot \frac{ 2\pi}{B} - (1-\alpha) \frac{2\pi}{2B}, s\cdot \frac{2\pi}{B} + (1-\alpha) \frac{2\pi}{2B}] ~\exists  s\in I \right],
\end{equation*}
where
\begin{equation*}
I = \{ \lfloor \frac{2\pi}{B\eta}(f-b) \frac{B}{2\pi} - \frac{1-\alpha}{2} \rfloor ,\cdots, \lceil \frac{4\pi}{B\eta}(f-b) \frac{B}{2\pi} + \frac{1-\alpha}{2} \rceil \}.
\end{equation*}
We apply Corollary \ref{cor:wrapping} by setting $\widetilde{T} = \frac{2\pi}{B}$, $\widetilde{\sigma} =2\pi \sigma$, $\widetilde{\delta} = 0$, $\widetilde{\epsilon} = (1-\alpha)\frac{2\pi}{2B}$, $A = 2\pi \frac{1}{B\eta}$, $\Delta f = |f-b|$.

By lower bound of Corollary \ref{cor:wrapping}, we have
\begin{equation*}
\mathsf{Pr} \left[ |o_{\sigma,b}(f)| \leq (1-\alpha) \frac{2\pi}{2B} \right] \geq (1-\alpha) - (1-\alpha) \cdot \frac{\eta}{|f-b|}.
\end{equation*}

Since $\alpha \cdot \frac{\eta}{|f-b|} >0$, then
\begin{equation*}
\mathsf{Pr} \left[ |o_{\sigma,b}(f)| \leq (1-\alpha) \frac{2\pi}{2B} \right] \geq (1-\alpha) -   \frac{\eta}{|f-b|}.
\end{equation*}


Recall that we sample $\sigma $ uniformly at random from $ [\frac{1}{B\eta},\frac{2}{B\eta}]$ and sample $b$ uniformly at random from $[0, \frac{\lceil F/\eta\rceil }{B\sigma}]$.
Since $f\in [-F,F]$ and $b$ is uniformly chosen from range $(0,\frac{\lceil F/\eta\rceil }{B\sigma}]$, thus 
for any $C>0$, $\underset{b}{\mathsf{Pr}} [ |f-b| \leq C \eta] \lesssim  \frac{C\eta}{F}$. Replacing $C$ by $1/\alpha$, we have $\mathsf{Pr} [ \frac{\eta}{|f-b|} \leq \alpha] \geq 1-\Theta(\frac{\eta}{\alpha F})$. Compared to $\frac{\eta}{F}$, $\alpha$ is just a constant. Thus, we finish the proof.



\end{proof}

\section{Proofs of basic hashing-related lemmas}
\restate{lem:two_close_signal}
\begin{proof}
Define $\nu=|f - f'|$.  First, let's show the first upper bound.  We have that
\begin{align*}
  \mathrm{LHS} &= \frac{1}{T} \int_0^T \left| v e^{2\pi f t\i} - v' e^{2\pi f' t\i } \right|^2 \mathrm{d} t\\
    &\leq 2 \cdot \frac{1}{T} \int_0^T \tabs{v e^{2\pi f t\i} - v e^{2\pi f' t\i }}^2 + \tabs{v e^{2\pi f' t\i} - v'e^{2\pi f' t\i }}^2 \mathrm{d} t\\
    &= 2\abs{v}^2 \cdot \frac{1}{T} \int_0^T \tabs{e^{2\pi \nu t\i} - 1}^2 \mathrm{d} t + 2\abs{v - v'}^2\\
    &\leq 2\abs{v}^2 \cdot \frac{1}{T} \int_0^T \min(2, 2\pi \nu t)^2 \mathrm{d} t + 2\abs{v - v'}^2\\
    &\leq 2 \abs{v}^2 \cdot \min(4, \frac{4\pi^2}{3} \nu^2 T^2) + 2\abs{v - v'}^2,
  \end{align*}
  as desired.

  Now consider the lower bound.  First we show this in the setting
  where $\abs{v} = \abs{v'}$.  Suppose $v' = v e^{-\theta \i}$.  Then we
  want to bound
  \begin{align*}
    \mathrm{LHS} &= \frac{1}{T} \int_0^T \left| v e^{2\pi f t\i} - v e^{(2\pi f' t - \theta)\i } \right|^2 \mathrm{d} t\\
    &= \abs{v}^2 \frac{1}{T} \int_0^T \left| e^{(2\pi \nu t + \theta)\i} - 1 \right|^2 \mathrm{d} t,
  \end{align*}
  as being at least $\Omega(\abs{v}^2 (\min(1, \nu^2 T^2) +
  \theta^2))$.  In the case that $\nu T < 1/10$, then
  \[
  \left| e^{(2\pi \nu t + \theta)\i} - 1 \right| \gtrsim \abs{2\pi\nu t + \theta},
  \]
  and
  \[
  \E_t[(2\pi\nu t + \theta)^2] \geq (\theta - \frac{2\pi \nu T}{2})^2 + \E[(2 \pi \nu (t - T/2))^2] \gtrsim \nu^2T^2 + (\theta - \frac{2\pi \nu T}{2})^2 \gtrsim \nu^2T^2 + \theta^2.
  \]
  On the other hand, if $\nu T > 1/10$, then $2\pi\nu t - \theta$ is
  $\Omega(1)$ for at least a constant fraction of the $t$, giving that
  \[
  \left| e^{(2\pi \nu t + \theta)\i} - 1 \right| \gtrsim 1.
  \]
  Hence the lower bound holds whenever $\abs{v} = \abs{v'}$.

  Finally, consider the lower bound for $\abs{v} \neq \abs{v'}$.
  Without loss of generality assume $\abs{v'} \geq \abs{v}$, and define
  $v^* = \frac{\abs{v}}{\abs{v'}}v''$.  For any two angles $\theta,
  \theta'$ we have that
  \[
  \abs{v e^{\theta \i} - v' e^{\theta' \i}}^2 \geq \abs{v e^{\theta \i} - v^* e^{\theta' \i}}^2 + \abs{v^* - v'}^2, 
  \]
  because the angle $\angle v v^* v'$ is obtuse.  Therefore
  \begin{align*}
    \mathrm{LHS} &\geq \frac{1}{T} \int_0^T \left| v e^{2\pi f t\i} - v^* e^{2\pi f' t\i } \right|^2 + \abs{v^* - v'}^2 \mathrm{d} t\\
    &\gtrsim  \abs{v}^2 \min(1, \nu^2 T^2) + \abs{v - v^*}^2 + \abs{v^* - v'}^2\\
    &\gtrsim \abs{v}^2 \min(1, \nu^2 T^2) + \abs{v - v'}^2.
  \end{align*}
  Now, if $\abs{v'}^2 \leq 2\abs{v}^2$, this gives the desired bound.
  But otherwise, it also gives the desired bound because $\abs{v -
    v'}^2 \gtrsim \abs{v'}^2$.  So we get the bound in all settings.

  The second equation follows from the first, using that $(a + b)^2
  \eqsim a^2 + b^2$ for positive $a, b$ to show
  \[
  \mathrm{dist}\left((v,f),(v',f')\right) \eqsim (|v| + |v'|)\cdot \min (1,T\abs{f - f'}) + |v - v'|.
  \]
  We can then replace $|v| + |v'|$ with $|v|$ because either they are
  equivalent up to constants or $|v - v'|$ is within a constant factor
  of $|v| + |v'|$.
\end{proof}

\restate{lem:expectation_lemma_when_x_star_is_zero}
\begin{proof}
Since $\widehat{u}_j = \text{FFT} (u_j)$, then $\sum_{j=1}^B |\widehat{u}_j|^2 = B\sum_{j=1}^B |u_j|^2$.
Recall that $(P_{\sigma,a,b}x)(t) =  x(\sigma(t-a))e^{2\pi \sigma b t\i }$, $u_j = \sum_{i=1}^{\log (k/\delta)} y_{j+B i}$ and $y_j = G_j \cdot (P_{\sigma,a,b} g)_j = G_j \cdot g(\sigma(j-a) )e^{2\pi\i \sigma b j}  $. Then

\begin{eqnarray*}
 \underset{\sigma,a,b}{\mathbb{E}}\sum_{j=1}^B |u_j|^2 &= & \underset{\sigma,a,b}{\mathbb{E}} \sum_{j=1}^B \left|\sum_{i=1}^{\log (k/\delta) } y_{j+ Bi} \right|^2 \\
&=& \underset{\sigma,a}{\mathbb{E} } \sum_{j=1}^B \underset{b}{\mathbb{E}} \left|\sum_{i=1}^{\log (k/\delta) } y_{j+ Bi} \right|^2 \\
&=& \underset{\sigma,a}{\mathbb{E} } \sum_{j=1}^B \underset{b}{\mathbb{E}} \left( \sum_{i=1}^{\log (k/\delta) } y_{j+ Bi} \right) \left( \sum_{i'=1}^{\log (k/\delta) } \overline{ y_{j+ Bi'} } \right) \\
&=& \underset{\sigma,a}{\mathbb{E} } \sum_{j=1}^B \underset{b}{\mathbb{E}} \left( \sum_{i=1}^{\log (k/\delta) } y_{j+ Bi} \overline{ y_{j+ Bi} }  + \sum_{i\neq i'}^{\log (k/\delta)} y_{j+Bi} \overline{ y_{j+ Bi'} }  \right). \\
\end{eqnarray*}
For any $(i,j) \in [\log (k/\delta)] \times [B]$, let $S_{i,j} = G_{j+ Bi}  g(\sigma(j+Bi-a)) = y_{j+Bi}e^{-2\pi \i \sigma b (j + Bi)}$, then
\begin{eqnarray*}
 \underset{\sigma,a,b}{\mathbb{E}}\sum_{j=1}^B |u_j|^2 &= &\underset{\sigma,a}{\mathbb{E} } \sum_{j=1}^B \underset{b}{\mathbb{E}} \left( \underbrace{ \sum_{i=1}^{\log (k/\delta)} |S_{i,j}|^2}_{C_1} + \underbrace{ \sum_{i\neq i'}^{\log (k/\delta)} S_{i,j} \overline{S_{i',j}} e^{2\pi\i \sigma b B (i-i')} }_{C_2} \right). \\
\end{eqnarray*}
Consider the expectation of $C_2$:
\begin{eqnarray}\label{eq:reason_for_sample_b}
\underset{b}{\mathbb{E}} C_2 & = & \underset{b}{\mathbb{E}}  \sum_{i\neq i'}^{\log (k/\delta) } S_{i,j} \overline{ S_{i',j}} e^{2\pi\i\sigma b B(i-i' )} \notag \\
 &  = &\sum_{i\neq i'}^{\log (k/\delta) } S_{i,j} \overline{ S_{i',j}} \underset{b}{\mathbb{E}} e^{2\pi\i\sigma b B(i-i' )} \notag \\
& = & 0 ~\quad \text{by Definition \ref{def:main_sample}}
\end{eqnarray}

Note that term $C_1$ is independent of $b$ which means $\underset{b}{\mathbb{E}} C_1 = C_1 $. Thus, we can remove the expectation over $b$. Then,
\begin{eqnarray}
 \underset{\sigma,a,b}{\mathbb{E}}\sum_{j=1}^B | u_j |^2 &=& \underset{\sigma,a}{\mathbb{E}} \sum_{j=1}^B \sum_{i=1}^{\log (k/\delta) } |G_{j + Bi}|^2 \cdot |g(\sigma(j+Bi-a))|^2 \notag \\
 & = & \underset{\sigma,a}{\mathbb{E}} \sum_{i=1}^{B\log (k/\delta) } |G_{i}|^2 \cdot |g(\sigma(i-a)) |^2\notag .
\end{eqnarray}

Now, the idea is to replace the expectation term $\mathbb{E}_{a}$ by
an integral term $\int_{a\in A} (\star) \mathrm{d} a$. Then, replace it by another integral term $\int_{0}^T (\star) \mathrm{d} t$. 
 Let
$A$ denote a set of intervals that we will sample $a$ from. It is easy
to verify that $|A| \lesssim T/\sigma$, since $(\sigma(i-a))$ is
sampled from $[0,T]$.  If we choose $T$ to be a constant factor larger
than $\sigma \abs{\supp(G)}$, then we also have $\abs{A} \gtrsim
T/\sigma$.

 \begin{eqnarray*}
\underset{\sigma}{\mathbb{E}} \underset{a}{\mathbb{E}}\sum_{i=1}^{B\log (k/\delta) } |G_{i}|^2 \cdot |g(\sigma(i-a))|^2  & = & \underset{\sigma}{\mathbb{E}} \frac{1}{|A|} \int_{a\in A} \sum_{i=1}^{B\log (k/\delta) } |G_i|^2 \cdot |g(\sigma(i-a))|^2 \mathrm{d} a \notag \\
  & = & \underset{\sigma}{\mathbb{E}}  \sum_{i=1}^{B\log (k/\delta) } |G_i|^2 \cdot \frac{1}{|A|} \int_{a\in A} |g(\sigma(i-a))|^2 \mathrm{d} a \notag \\
    & = & \underset{\sigma}{\mathbb{E}}  \sum_{i=1}^{B\log (k/\delta) } |G_i|^2 \cdot \frac{1}{\sigma|A|} \int_{a\in A} |g(\sigma(i-a))|^2 \mathrm{d} \sigma a \notag \\
    &\lesssim & \underset{\sigma}{\mathbb{E}} \| G \|_2^2 \cdot \frac{1}{T} \int_0^T |g(t)|^2 \mathrm{d}t .
 \end{eqnarray*}
By Claim \ref{cla:l2_of_discrete_G}, we know that $\| G\|_2^2 \eqsim \frac{1}{B}$. Combining $\sum_{j=1}^B |\widehat{u}_j|^2 = B\sum_{j=1}^B |u_j|^2$ and $\| G\|_2^2 \eqsim \frac{1}{B}$ gives the desired result. 
\end{proof}

\restate{lem:expectation_lemma_when_g_is_zero}
\begin{proof}
For simplicity, let $G =G_{B,\delta,\alpha}$ and $\widehat{G'} = \widehat{G'}_{B,\delta,\alpha}$. we have 
\begin{eqnarray*}
\widehat{y} &= &\widehat{G \cdot P_{\sigma,a,b} x} \\
& = &  \widehat{G} * \widehat{P_{\sigma,a,b}x } \\
& = & \widehat{G'}* \widehat{P_{\sigma,a,b} x} + (\widehat{G} - \widehat{G'}) *\widehat{P_{\sigma,a,b} x}.
\end{eqnarray*}

The $\ell_{\infty}$ norm of second term can be bounded :
\begin{equation*}
\| (\widehat{G} - \widehat{G'} ) * \widehat{P_{\sigma,a,b}x }\|_{\infty} \leq \| \widehat{G} - \widehat{G'}\|_{\infty} \| \widehat{P_{\sigma,a,b}x } \|_1 \leq \sqrt{\delta/k} \| \widehat{x^*} \|_1.
\end{equation*}

Thus, consider the $j$th term of $\widehat{u}$,
\begin{eqnarray*}
\widehat{u}_j & = & \widehat{y}_{jF/B}\\
& = & \sum_{|l|< F/(2B)} \widehat{G'}_{-l} (\widehat{P_{\sigma,a,b} x})_{jF/B +l} \pm \sqrt{\delta/k} \| \widehat{x^*} \|_1 \\
& = & \sum_{| \pi_{\sigma,b}(f) - jF/B | < F/(2B)} \widehat{G'}_{jF/B - \pi_{\sigma,b}(f)} \widehat{P_{\sigma,a,b} x}_{\pi_{\sigma,b}(f) } \pm \sqrt{\delta/k} \| \widehat{x^*} \|_1 \\
& = & \sum_{h_{\sigma,b}(f) = j} \widehat{G'}_{-o_{\sigma,b}(f)} \widehat{x^*}(f) e^{2\pi f \sigma a\i} \pm  \sqrt{ \delta/k} \| \widehat{x^*} \|_1.
\end{eqnarray*}

If neither $E_{ \mathit{coll}}(f)$ nor $E_{\mathit{off}}(f)$ happens, then we know that frequency $f$ is the only heavy hitter hashed into bin $j$ and $\widehat{G'}_{-o_{\sigma,b}(f) }=1$ for frequency $f$. Thus,

\begin{eqnarray*}
\underset{\sigma,a,b}{\mathbb{E}} [\left| \widehat{u}_{h_{\sigma,b}(f)} - \widehat{x^*}(f) e^{a\sigma 2\pi f\i}\right|^2] \leq \delta/k \| \widehat{x^*} \|_1^2.
\end{eqnarray*}
Since the above equation holds for all $f\in H$, we get
\begin{eqnarray*}
\sum_{f\in H} \underset{\sigma,a,b}{\mathbb{E}} [\left| \widehat{u}_{h_{\sigma,b}(f)} - \widehat{x^*}(f) e^{a\sigma 2\pi f\i}\right|^2] \leq k  \delta/k \| \widehat{x^*} \|_1^2 =\delta \| \widehat{x^*} \|_1^2.
\end{eqnarray*}
\end{proof}

\section{Proofs for one stage of recovery}
\paragraph{ Binary search of one-sparse algorithm}
We first explain a simple, clean, but not optimal one-sparse algorithm, then we try to optimize the algorithm step by step. Let ``heavy'' frequency $f\in [-F,F]$, we can split the frequency interval into two regions: left region $[-F,0)$ and right region $[0,F]$. We can observe $ \theta \equiv 2\pi \beta f {\pmod {2\pi} }$ by checking the phase difference between using $P_{1,a,b}$ and $P_{1,a+\beta,b}$, where  $\beta$ is uniformly at random sampled from some suitable range $[\widehat{\beta}, 2\widehat{\beta}]$. For each observation $\theta$, we can guess $m$ different possibilities for $f$, say $\theta_1,\theta_2, \cdots, \theta_m$, if $\theta_i$ belong to the left region, we add a vote to that, otherwise we add a vote to the right region.  After taking enough samples, we choose the region that has the largest vote. This decision will let us narrow down the searching range of the true frequency by half with some good probability. Suppose we decide to choose the right region $[0,F]$, then we can just repeat the above binary search over $[0,F]$ again to get into a region that has size $F/2$. Repeating it $D$ times, we can learn the frequency with a region that has size at most $2F/{2^D}$. But the binary search is not the best approach, actually, we can do much better with using $t$-ary search.

\paragraph{$k$-sparse}

To locate those $k$ heavy signals in the frequency domain, we need to consider the ``bins'' computed by $\mathsf{HashToBins}$ with $P_{\sigma,a,b}$. One of the main difference from previous work \cite{HIKP} is, instead of permuting the discrete coordinates according to $P_{\sigma,a,b}$ and partitioning the coordinates into $B=O(k)$ bins, we permute the continuous frequency domain and partition the frequency domain into $B=O(k)$ bins. With a large constant probability, we can obtain for each heavy signal $f$ that neither $E_{\mathit{coll}}$ nor $E_{\mathit{off}}$ happens. After splitting those $k$ frequencies into different bins, then we can run the one-sparse algorithm for all the bins simultaneously.

\paragraph{$t$-ary search}
To argue the final succeed probability of our algorithm, we need to take the union bound for each array/region in each round and also take the union bound over each round. There are several benefits of changing binary search to $t$-ary search, (1) to reach the same accuracy, the number of rounds $D$ for $t$-ary search is smaller than the number of rounds for binary search; (2) our searching procedure is a ``noisy'' searching problem, for the noiseless version of the searching problem, we do not need to take care of the union bound argument. Having a parameter for the number of arrays/regions is important to optimize the entire procedure.

Recall the traditional binary search problem(noiseless version), given a list of sorted numbers $a[1,2\cdots, n]$ in increasing order. We want to determine if some number $x$ belongs to $a[1,2,\cdots,n]$. When we compare some $a[i]$ with $x$, we will know the true answer.

But the noisy binary search problem is slightly harder. $a[1,2,\cdots,n]$ is still a list of sorted numbers in increasing order, we want to determine if some number $x$ belongs to $a[1,2,\cdots,n]$. But when we compare the $a[i]$ with $x$, we will know the true answer with $9/10$ probability, and get the false answer with $1/10$ probability. In this case, we can not finish the task by following the procedure of traditional binary search algorithm, e.g. making the decision by just comparing $a[i]$ with $x$ once. The reason is after taking the union bound of $\log n$ rounds, the failure probability can be arbitrarily large. One idea to fix this issue is independently comparing $a[i]$ with $x$ multiple times, e.g. $\log \log n$ times. Then we can amplify the succeed probability of each round from $9/10$ to $1-\frac{1}{\poly(\log n)}$, after taking the union bound over $\log n$ rounds, we still have $1-\frac{1}{\poly(\log n)}$ succeed probability.

Our problem is still more complicated than the above noisy binary search problem. In each round of algorithm $\mathsf{LocateInner}$, we do a $t$-array search by splitting the candidate frequency region(that has length $\Delta l$) into $t$ consecutive regions, $Q_1, Q_2, \cdots, Q_t$, each of them has the equal size $\frac{\Delta l}{t}$. By using the hash values of $P_{\sigma,a,b}$ and $P_{\sigma,a+\beta,b}$ (Line 26-27 in Algorithm \ref{alg:noisy_k_sparse2}), we can have an observation over $[0,2\pi)$. For each such observation over $[0,2\pi)$, it was in fact scaled by $ 2\pi \sigma \beta$ and rounded over $[0,2\pi)$. Thus the corresponding frequency location of this observation might belong to $m = \Theta (\sigma \beta \Delta l)$ different possible regions. Since we do not know which one, we just add a vote to all of the possible regions.  To understand $t$-ary search, let's consider this example. Suppose we split the frequency into $20$ regions, the true frequency belongs to region $9$ and each observation is correct with probability $4/5$. For each observation, we will add a vote to a batch of roughly evenly spaced regions.
\begin{enumerate}
\item For the observation 1, we add a vote to region $1,5,9,13,17$.
\item For the observation 2, we add a vote to region $3,9,15$.
\item For the observation 3, we add a vote to region $2,9,16$.
\item For the observation 4, we add a vote to region $4,9,14,19$.
\item For the observation 5, we add a vote to region $2,6,10,14,20$.
\end{enumerate}
where the first four observations are the correct observations and the last one is wrong. Then the true region will have more than half of the $R_{loc}=5$ votes with some good probability.

\paragraph{Details of adding vote}
In previous description, to let people have a better understanding of $t$-ary search, we simplify the step of adding vote. In fact, adding a vote to each of the $m$ candidate region is not enough. The right thing to do is,  not only adding a vote to those $m$ regions, but also adding a vote to a constant number $c_n$ of neighbors of each possible region (e.g. two neighbors nearby that region, line 33 in Algorithm \ref{alg:noisy_k_sparse2}). The reason for doing this is, if the frequency $f$ is located very close to the boundary of two regions, then neither of them will get more than $R_{loc}/2$ votes. After we already have $R_{loc}$ independent observations, we just choose any region that contains more than $R_{loc}/2$ votes and enlarge it by the same constant factor $c_n$ to be the next candidate frequency region for $t$-array search, and plugging the new parameters into algorithm $\mathsf{LocateInner}$ and running it again.

\paragraph{The slightly different last round}
Recall that, in the previous description of each round of $t$-ary search, we're able to decrease the searching range of frequency geometrically. To achieve this progress, we have to increase the sample duration of  $\beta$ is geometrically. At some point, the sample duration will reach $T$. Suppose the sample duration of $\beta$ is geometrically increasing during the first $D-1$ rounds (Line 13-15 in Algorithm \ref{alg:noisy_k_sparse2}) and it becomes $T$ after $D-1$ rounds. If we want to use the same duration $T$ to perform one more round again, can we get some benefit for learning the frequency by doing some tricks? (1) Suppose in all the first $D-1$ rounds, we use $R_{loc} (D-1)$ samples. Then we can use $R_{loc} (D-1)$ for the last round, since it does not increase the sample complexity. This way allows us to learn frequency within $\frac{1}{C T}$. But can we do better than that? (2) Using more samples is a nice observation, another right thing to do is changing the algorithm from reporting region to reporting frequency. In the previous $D-1$ rounds, as the description of $t$-ary search, we just report the region that has more than $R_{loc}/2$ votes. But, we can take the median over all the values that were assigned to any region that has votes more than $R_{loc}/2$. This way can actually allow us to learn frequency location more accurately ( $\eqsim \frac{1}{\rho T}$) without increasing the resolution for $\beta$! Another question people might ask is, if we take median in the last round, why not just do it in every round? The answer is, in the first $D-1$ rounds, our goal is just narrowing down the search range of frequency location and we do not need to report a frequency location. Another reason is reporting a candidate region  needs less samples and has higher succeed probability, but taking the median is more expensive than reporting a candidate region. We cannot pay for taking the median in every round.

To prove main Lemma \ref{lem:correctness_of_locateinner} for one stage recovery, we introduce Lemma \ref{lem:good_vote}. Before explaining the proof, we give some definitions. 
Consider a frequency $f$, and define $j = h_{\sigma,b}(f)$ to be the bin that frequency $f$ was hashed into. Define $\theta = f - b \pmod F$. Define 
\begin{equation*}
\widehat{u} = \mathsf{HashToBins}(x,P_{\sigma,\gamma,b},B,\delta,\alpha) \quad \mathrm{and} \quad \widehat{u}'  = \mathsf{HashToBins}(x,P_{\sigma,\gamma +\beta,b},B,\delta,\alpha).
\end{equation*}
Note that ``bin'' and ``region'' are representing different things in this paper. ``bin'' is related to hash function $h_{\sigma,b}(f)$. ``region'' is only used in the algorithm of one stage recovery in this section. 
Define ``true'' region to be the region that contains frequency $f$. Define ``wrong'' region to be the region that is not within a constant number $c_n$ of neighbors of the ``true'' region.
Part (\RN{1}) of Lemma \ref{lem:good_vote} shows that for each observation generated by a batch of samples drawn from time domain, the counter $v_{j,q'}$ corresponding to the true region will increase by one, with some ``good'' probability. On the other side, Part (\RN{2}) of Lemma \ref{lem:good_vote} shows that the counter $v_{j,q}$ corresponding to wrong region will not increase by one, with some ``good'' probability. Note that, \cite{HIKP} proved the case when $c_n=6$ under discrete setting, we translate it into continuous setting, here.

\begin{lemma}\label{lem:good_vote}
Given $\sigma$ and $b$. Assume $f\in \mathrm{region}(j,q')$ and $\underset{\gamma}{\mathbb{E}} [ \left| \widehat{u}_j - e^{2\pi \gamma\sigma \theta \i }\widehat{x}(\theta) \right|^2 ]  \leq \frac{1}{\rho^2}  |\widehat{x}(\theta)|^2 $. $\forall 0<s<1$, for each two samples $(\gamma,0)$ and $(\gamma,\beta)$, where $\gamma$ is sampled from $[\frac{1}{2},1]$ uniformly at random and $\beta$ is sampled from $[\frac{st}{4\sigma \Delta l}, \frac{st}{2\sigma\Delta l}]$ uniformly at random, we have

$($\RN{1}$)$ for the $q'$, with probability at least $1-(\frac{2}{\rho s})^2$, $v_{j,q'}$ will increase by one.

$($\RN{2}$)$ for any $q$ such that $|q-q'|>3$, with probability at least $1-15s$, $v_{j,q}$ will not increase.
\end{lemma}

\begin{proof} Part (\RN{1}) of Lemma \ref{lem:good_vote}.
We have,
\begin{equation*}
\underset{\gamma}{\mathbb{E}} [ \left| \widehat{u}_j - e^{2\pi \gamma\sigma \theta \i }\widehat{x}(\theta) \right|^2 ]  \leq \frac{1}{\rho^2}  |\widehat{x}(\theta)|^2.
\end{equation*}
By Chebyshev's Inequality, we have $\forall ~g>0$, with probability $1-g$ we have
\begin{equation*}
\begin{split}
\left| \widehat{u}_j - e^{2\pi\gamma\sigma \theta\i } \widehat{x}(\theta)\right| & \leq \frac{1}{\rho} \sqrt{\frac{1}{g}} \left| \widehat{x}(\theta) \right| \\
\| \phi(\widehat{u}_j) -  ( \phi(\widehat{x}(\theta)) - 2\pi\gamma\sigma\theta) \|_\bigcirc &\leq \sin^{-1} (\frac{1}{\rho} \sqrt{\frac{1}{g}}), \\
\end{split}
\end{equation*}
where $\| x-y\|_\bigcirc = \underset{z \in \mathbb{Z} }{\min} |x- y + 2\pi z| $ denote the ``circular distance'' between $x$ and $y$. 
Similarly, replacing $\gamma$ by $\gamma+\beta$, with probability $1-g$, we also have that
\begin{equation*}
\| \phi(\widehat{u}'_j) - (\phi(\widehat{x}(\theta)) - 2\pi(\gamma+\beta )\sigma \theta ) \|_\bigcirc \leq \sin^{-1} (\frac{1}{\rho} \sqrt{\frac{1}{ g}}).
\end{equation*}
Define $c_j = \phi(\widehat{u}_j / \widehat{u}'_j )$. Combining the above two results, with probability $1-2g$ we have 
\begin{eqnarray}\label{eq:cj_within_rho}
\| c_j - 2\pi \beta \sigma\theta \|_\bigcirc & =& \| \phi(\widehat{u}_j) - \phi(\widehat{u}'_j) - 2\pi \beta \sigma\theta \|_\bigcirc\notag \\
&= & \| (\phi(\widehat{u}_j) -  ( \phi(\widehat{x}(\theta)) - 2\pi \gamma \sigma\theta )) )- (\phi(\widehat{u}'_j) - (\phi(\widehat{x}(\theta)) - 2\pi (\gamma+\beta)\sigma\theta ))) \|_\bigcirc \notag \\
&\leq & \| \phi(\widehat{u}_j) -  ( \phi(\widehat{x}(\theta)) - 2\pi \gamma \sigma\theta  ) \|_\bigcirc+ \| \phi(\widehat{u}'_j) - (\phi(\widehat{x}(\theta)) - 2\pi (\gamma + \beta )\sigma\theta ) \|_\bigcirc\notag \\
&\leq & 2  \sin^{-1} (\frac{1}{\rho} \sqrt{\frac{1}{ g} }) \notag.
\end{eqnarray} 
Here, we want to set $g = (\frac{4}{s \pi \rho})^2 $, thus, with probability at least $1- ( \frac{2}{s\rho} )^2$
\begin{equation}\label{eq:cj_is_close_to_2pibetasigmatheta}
\| c_j - 2\pi \beta \sigma\theta \|_\bigcirc < s\pi/2.
\end{equation}
The above equation shows that $c_j$ is a good estimate for $2\pi\beta \sigma\theta$ with good probability. We will now show that this means the true region $Q_{q'}$ gets a vote with large probability.

For each $q'$ with $f \in [l_j- \frac{\Delta l}{2} + \frac{q'-1}{t}\Delta l, l_j -\frac{\Delta l}{2} + \frac{q'}{t} \Delta l] \subset [-F,F]$, we have that $m_{j,q'}= l_j -\frac{\Delta l}{2} + \frac{q'-0.5}{t} \Delta l$ and $\theta_{j,q'} = m_{j,q'} - b \pmod F$ satisfies
\begin{equation*}
|f - m_{j,q'}| \leq \frac{\Delta l}{2t} ~ \mathrm{and} ~|\theta - \theta_{j,q'}| \leq \frac{\Delta l}{2t}.
\end{equation*} 

Since we sample $\beta$ uniformly at random from $[\frac{st}{4\sigma \Delta l}, \frac{st}{2\sigma \Delta l}]$, then $\beta \leq \frac{s t }{2 \sigma \Delta l}$, which implies that $2\pi \beta \sigma \frac{\Delta l}{2t} \leq \frac{s\pi}{2}$. Thus, we can show the observation $c_j$ is close to the true region in the following sense,
\begin{equation*}
\begin{split}
& \| c_j - 2\pi \beta \sigma\theta_{j,q'} \|_\bigcirc \\
\leq & \| c_j - 2\pi \beta \sigma\theta \|_\bigcirc + \| 2\pi \beta\sigma\theta - 2\pi \beta \sigma\theta_{j,q'} \|_\bigcirc \text{~by~triangle~inequality}\\
<& \frac{s\pi}{2} +  2\pi \| \beta\sigma\theta - \beta \sigma\theta_{j,q'} \|_\bigcirc \text{~by~Equation~(\ref{eq:cj_is_close_to_2pibetasigmatheta})}\\
\leq &\frac{s\pi}{2} + 2\pi \beta\sigma \frac{\Delta l}{2t}\\
= & \frac{s\pi}{2} + \frac{s\pi}{2}\\
\leq & s\pi .
\end{split}
\end{equation*}

Thus, $v_{j,q'}$ will increase in each round with probability at least $1-(\frac{2}{\rho s})^2$.
\end{proof}


\begin{proof} Part (\RN{2}) of Lemma \ref{lem:good_vote}.

Consider $q$ with $|q-q'|> 3$. Then $|f - m_{j,q}| \geq \frac{7\Delta l}{2t}$, and (assuming $\beta \geq \frac{s t }{4\sigma \Delta l}$) we have
\begin{equation}\label{eq:lower_bound_beta_xi_k}
2\pi \beta \sigma |f - m_{j,q}| \geq 2\pi \frac{s t }{ 4\sigma \Delta l} \sigma |f - m_{j,q}| =  \frac{s\pi t }{ 2 \Delta l}  |f - m_{j,q}| \geq \frac{7 s\pi}{4} > \frac{3s\pi}{2}.
\end{equation}
There are two cases: $|f - m_{j,q}| \leq \frac{\Delta l}{s t}$ and $|f - m_{j,q}| > \frac{\Delta l}{st}$.

First, if $|f - m_{j,q}| \leq \frac{\Delta l}{s t}$.
In this case, from the definition of $\beta$ it follows that
\begin{equation}\label{eq:upper_bound_beta_xi_k}
2\pi \beta \sigma |f - m_{j,q}| \leq \frac{s\pi t}{\sigma \Delta l} \sigma|f - m_{j,q}| \leq \pi.
\end{equation}

Combining equations (\ref{eq:lower_bound_beta_xi_k}) and (\ref{eq:upper_bound_beta_xi_k}) implies that
\begin{equation*}
\mathsf{Pr} [2\pi \beta \sigma (f - m_{j,q}) \mod 2\pi \in [-\frac{3s}{4}2\pi,\frac{3s}{4}2\pi ] ] =0.
\end{equation*}

Second, if $|f - m_{j,q}| > \frac{\Delta l}{st}$. We show this claim is true:
$ \mathsf{Pr} [2\pi \beta \sigma (f -m_{j,q}) \pmod {2\pi} \in [-\frac{3s}{4} 2\pi, \frac{3s}{4} 2\pi] ]  \lesssim s$.
To prove it, we apply Corollary \ref{cor:wrapping} by setting $\widetilde{T} = 2\pi$, $\widetilde{\sigma} = 2\pi\sigma \beta$, $\widetilde{\delta} = 0$, $\widetilde{\epsilon} = \frac{3s}{4}2\pi$, $A= 2\pi \sigma \widehat{\beta}$, $\Delta f= |f - m_{j,q}|$. By upper bound of Corollary \ref{cor:wrapping}, the probability is at most
\begin{equation*}
\frac{2\widetilde{\epsilon} }{\widetilde{T} } + \frac{4\widetilde{\epsilon} }{A \Delta f} = \frac{3s}{2} + \frac{ 3s}{\sigma \widehat{\beta} \Delta f} \leq \frac{3s}{2} + \frac{3s}{\sigma \frac{st}{4\sigma \Delta l} \frac{\Delta l}{st}} < 15 s.\\
\end{equation*}


Then in either case, with probability at least $1-15s$, we have
\begin{equation*}
\| 2\pi \beta \sigma m_{j,q} - 2\pi \beta \sigma f \|_\bigcirc > \frac{3s}{4} 2\pi.
\end{equation*}
which implies that $v_{j,q}$ will not increase.
\end{proof}

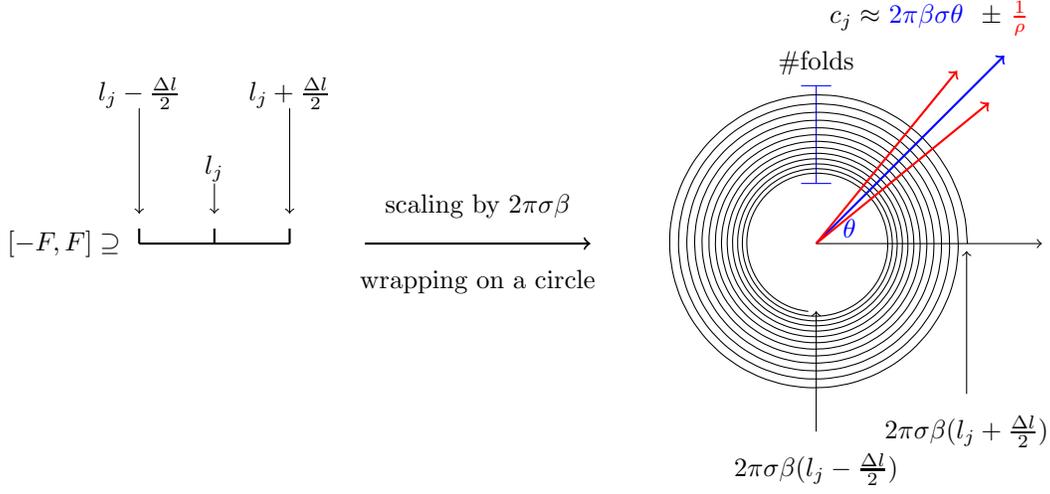
\begin{figure}
\centering{\small
\begin{tikzpicture}[scale = 1.0]
\def\PI{3.1415926}

	\node at (-10,0) {$[-F,F] \supseteq$};

	\draw [->] (0,-2.5) -- (0,-0.9);
	\draw [->] (2,-2) -- (2,-0.1);
	\node at (0,-3) {$2\pi\sigma\beta(l_j-\frac{\Delta l}{2})$};
	\node at (2,-2.5) {$2\pi\sigma\beta(l_j+\frac{\Delta l}{2})$};
	
	\draw [thick] (-9,0) -- (-7,0);
	\node at (-9,2) {$l_j - \frac{\Delta l}{2}$};
	\draw [->](-9,1.8) -- (-9, 0.4);

	\node at (-7,2) {$l_j + \frac{\Delta l}{2}$};
	\draw [->](-7,1.8) -- (-7, 0.4);

	\node at (-8,1) {$l_j$};
	\draw [->] (-8,0.8) -- (-8,0.4);
	
	\draw [thick] (-9,0) -- (-9,0.2);
	\draw [thick] (-8,0) -- (-8,0.2);
	\draw [thick] (-7,0) -- (-7,0.2);

	\node  at (-4.5,0.5)  {scaling by $2\pi \sigma \beta$};
	\node at (-4.5,-0.5) {wrapping on a circle};
	\draw[thick,->] (-6,0) -- (-3,0);

    	\node at (0,2.4) {$\#$folds};
    	\draw [blue] (-0.2,0.8) -- (0.2,0.8);
    	\draw [blue] (-0.2,2.1) -- (0.2,2.1);
    	\draw [blue] (0,0.8) --(0,2.1);
    	\draw [domain=0:80,variable=\t,smooth,samples=500]
        plot ({\t r}: {0.01+2*exp(-0.01*\t)});

        \draw [->,thick,blue] (0,0) -- (2.5,2.5) ;
        \draw [->] (0,0) -- (3,0) ;
        \draw [->,thick,red] (0,0) -- ( { 3*0.76604 }, { 3*0.62478 } );
        \draw [->,thick,red] (0,0) -- ( { 3*0.62478 }, { 3*0.76604 } );
        \node at (1.5,3) {$c_j \approx$ {\color{blue} $2\pi\beta \sigma \theta$ } $\pm$ {\color{red} $\frac{1}{\rho}$} };
        \node at (0.5,0.2) {{\color{blue} $\theta$ } };
\end{tikzpicture}
\caption{For an arbitrary frequency interval $[ l_j - \frac{\Delta l}{2}, l_j+ \frac{\Delta l}{2}]$, we scale it by $2\pi\sigma\beta$ to get a longer interval $[2\pi\sigma\beta (l_j +\frac{\Delta l}{2}) , 2\pi\sigma\beta (l_j - \frac{\Delta l}{2})]$. Then, we wrap the longer interval on a circle $[0,2\pi)$. The number of folds after wrapping is $\lceil \sigma\beta \Delta l\rceil$. For any random sample, the observation $c_j$ is close to the true answer within $1/\rho$ with some ``good'' probability.  }
}
\end{figure}
\restate{lem:correctness_of_locateinner}
\begin{proof}

Let $t$ denote the number of regions, and $[l_j- \frac{\Delta l}{2}, l_j +\frac{\Delta l}{2}]$ be the interval that contains frequency $f$. Let $Q_q$ denote a region that is $[l_j- \frac{\Delta l}{2} + (q-1)\frac{\Delta l}{t},l_j- \frac{\Delta l}{2} + q \frac{\Delta l}{t}]$. Let $\theta = f- b \pmod F$. Recall that we sample $\sigma$ uniformly at random from $[\frac{1}{B\eta}, \frac{2}{B\eta}]$. Then we sample $\gamma$ uniformly at random from $[\frac{1}{2}, 1]$ and sample $\beta$ uniformly at random from $[\frac{ s  t}{4\sigma\Delta l},\frac{ s  t}{2\sigma\Delta l}]$.  Define $c_j = \phi(\widehat{u}_j / \widehat{u}'_j)$. Let $m$ denote the number of folds, which is equal to $\lceil \sigma \beta \Delta l \rceil$. Let $v_{j,q}$ denote the vote of region$(j,q)$.


We hope to show that in any round $r$, each observed $c_j$ is close to $2\pi\sigma\beta  \theta$ with good probability. On the other hand, for each observed $c_j$, we need to assign it to some regions and increase the vote of the corresponding region. The straightforward way is just checking all possible $t$ regions, which takes $O(t)$ time. In fact, there are only $\Theta(m)$ regions close enough to the observation $c_j$, where $m= \Theta(\frac{2\pi\sigma\beta \Delta l}{2\pi}) = \Theta(s t)$. The reason is $2\pi \sigma \beta$ will scale the original length $\Delta l$ interval to a new interval that has length $2\pi\sigma \beta \Delta l$. This new interval can only wrap around circle $[0,2\pi)$ at most $\lceil \sigma \beta \Delta l \rceil$ times. The running time is $O(st R_{loc})$, since we take $R_{loc}$ independent observations.

For each observed $c_j$: Part (\RN{1}) of Lemma \ref{lem:good_vote} says, we assign it to the true region with some good probability; Part (\RN{2}) of Lemma \ref{lem:good_vote} says, we do not assign it to the wrong region with some good probability. Thus taking $R_{loc}$ independent $c_j$, we can analyze the failure probability of this algorithm based on these three cases. 

(\RN{1}) What's the probability of true fold fails?
\begin{eqnarray*}
& & \mathsf{Pr}[ \text{True ~ fold ~ fails}] \\
&=& \mathsf{Pr} [ \text{True ~region ~fails} ~ \geq R_{loc}/2 ~\text{times}] \\
&\leq & 2 \cdot {R_{loc} \choose R_{loc}/2} \cdot \left( \mathsf{Pr} [ \text{True~region~fails~once} ] \right)^{R_{loc}/2} \\
&\leq & 2 \cdot {R_{loc} \choose R_{loc}/2} \cdot ( \frac{2}{s\rho} )^{2R_{loc}/2} \quad \text{by~Lemma~\ref{lem:good_vote}}\\
&\leq & (  \frac{4}{s\rho}  )^{R_{loc}}.
\end{eqnarray*}

(\RN{2}) What if the region that is ``near''(within $c_n$ neighbors) true region becomes true?

Any of those region gets a vote only if true region also gets a vote. Since our algorithm choosing any region that has more than $R_{loc}/2$ votes and enlarging the region size by containing $c_n$ nearby neighbors of that chosen region, then the new larger region must contain the ``real'' true region. 

(\RN{3}) What if the region that is ``far away''(not within $c_n$ neighbors) from true region becomes true?

By Part (\RN{2}) of Lemma \ref{lem:good_vote}, the probability of one such wrong region gets a vote is at most $15s$. Thus, one of the wrong region gets more than $R_{loc}/2$ votes is at most $(60 s)^{R_{loc}/2}$. By taking the union bound over all $t$ regions, we have the probability of existing one wrong region getting more than $R_{loc}/2$ vote is at most  $t \cdot (60s)^{R_{loc}/2}$.

Thus, if we first find any region $Q_q$ that has more than $R_{loc}/2$ votes, and report a slightly larger region $[l_j - \frac{\Delta l}{2} + (q-1)\frac{\Delta l}{t} -  \frac{c_n}{2}  \frac{\Delta l}{t}, l_j - \frac{\Delta l}{2} + q \frac{\Delta l}{t} +  \frac{c_n}{2}  \frac{\Delta l}{t} ]$, it is very likely this large region contains the frequency $f$. Finally, the failure probability of this algorithm is at most $\Theta( (\frac{4}{s\rho})^{R_{loc}} + t \cdot (60s)^{R_{loc}/2} )$.
\end{proof}

\begin{lemma}\label{lem:learn_f_within}
Taking the median of values belong to any region getting at least $\frac{1}{2}R_{loc}$ votes, then we can learn frequency $f$ within $\Theta(\frac{\Delta l}{\rho s t})$ with probability $1- \exp(-\Omega(R_{loc}))$.
\end{lemma}

\begin{proof}
Let $\text{region}(j,q')$ be the region that getting at least $\frac{1}{2}R_{loc}$ votes. Let $R = |\text{region}(j,q')|$ denote the number of observations/votes assigned to $\text{region}(j,q')$. Since this region getting at least $\frac{1}{2}R_{loc}$ votes, then $\frac{1}{2} R_{loc} \leq R \leq R_{loc}$. Using Equation (\ref{eq:cj_within_rho}) in Lemma \ref{lem:good_vote}, $\forall g >0 $, we have

\begin{equation*}
\| c_j -2\pi \beta \sigma \theta \|_{\bigcirc} \leq 2 \left( \sin^{-1} (\frac{1}{\rho} \sqrt{\frac{1}{g}} ) \right),
\end{equation*}
holds with probability $1-2g$. Choosing $g= \Theta(1)$, we have with constant success probability $p > \frac{1}{2}$,
\begin{equation*}
\| c_j -2\pi \beta \sigma \theta \|_{\bigcirc} \lesssim \frac{1}{\rho},
\end{equation*}
holds.

Taking the median over all the observations that belong to $\text{region}(j,q')$ gives 

\begin{eqnarray}\label{eq:median_of_true_bin}
\mathsf{Pr}\left( \left| \underset{r \in \text{region}(j,q')}{ \text{median}}  c_j^r - 2\pi \beta^r \sigma \theta \right| \gtrsim \frac{1}{\rho} \right) & < &  \sum_{i=  R/2 }^{R_{loc}} { R_{loc} \choose i} (1-p)^i p^{R_{loc}-i} \notag \\
& = & \sum_{i = R_{loc}/2}^{R_{loc }} {R_{loc} \choose i }(1-p)^i p^{R_{loc} - i} + \sum_{i=R/2}^{R_{loc}/2} {R_{loc} \choose i }(1-p)^i p^{R_{loc} - i} \notag \\
& < & \sum_{i=  R_{loc}/2 }^{R_{loc}} { R_{loc} \choose R_{loc}/2} (1-p)^i + \sum_{i=R/2}^{R_{loc}/2} { R_{loc} \choose R_{loc}/2} (1-p)^i  \notag \\
& < & 2 { R_{loc} \choose R_{loc}/2} (1-p)^{R/2} \notag \\
& \leq &  2 (2e)^{R_{loc}/2}  (1-p)^{ R_{loc}/4}  \notag\\
& \leq &  e^{-cR_{loc} },
\end{eqnarray}
where the the second inequality follows by ${R \choose i} \leq {R\choose R/2}$ and $p^{R_{loc}-i} < 1$, $\forall i$ ; the fourth inequality follows by ${n \choose k }\leq (ne/k)^k$; the last inequality follows by choosing some $p$ such that $\frac{1}{2}\log_{2e} \frac{1}{1-p}  - 1> 2c$ where $c>0$ is some constant.
Equation (\ref{eq:median_of_true_bin}) implies that
\begin{equation*}
\mathsf{Pr}\left( \left| \underset{r \in \text{region}(j,q')}{ \text{median}}  \theta^r -\theta \right| < \frac{1}{\rho 2\pi \sigma \widehat{\beta} } \right) > 1- \exp(-\Omega( R_{loc} ) ),
\end{equation*}
where $\forall r\in [R_{loc}]$, $\beta^r $ is sampled uniformly at random from $[\widehat{\beta},2\widehat{\beta}] = [\frac{s t}{4\sigma \Delta l}, \frac{st}{2\sigma \Delta l}]$. Thus, we can learn $f$ within $\Theta(\frac{1}{\rho 2\pi \sigma \widehat{\beta}}) = \Theta(\frac{\Delta l}{\rho s t})$.

\end{proof}

\begin{figure}
\centering
\begin{tikzpicture}
\def\dx{8}
\draw (-\dx,0) -- (\dx,0) ;
\draw (\dx,0) -- (\dx,0.2);
\node at (-\dx,-0.5) {$l_j -\frac{\Delta l}{2}$};

\node at (\dx,-0.5) {$l_j + \frac{\Delta l}{2}$};

\node at (0,-0.5) {$l_j$};

\node at (3,0.5) {$l_j - \frac{\Delta l}{2} + (q'-1)\frac{\Delta l}{t}$};
\draw [red,thick] (2.5,0) -- (2.5,0.3);
\draw [red,thick] (3.0,0) -- (3.0,0.3);

\node at (4,-0.5) {$l_j - \frac{\Delta l}{2} + q'\frac{\Delta l}{t}$};
\draw (-\dx+1,0) -- (-\dx+1,0.2);
\draw (-\dx+2,0) -- (-\dx+2,0.2);
\draw (-\dx+3,0) -- (-\dx+3,0.2);

\draw (-\dx,-2) -- (\dx,-2) ;
\draw (-\dx,-3) -- (\dx,-3) ;
\foreach \i in {-\dx,...,7}
{
	\draw (\i,0) -- (\i,0.2);
	\draw (\i+0.5,0) -- (\i+0.5,0.2);
	\draw (\i,-2) -- (\i,-3);
	\draw (\i+0.5,-2) -- (\i+0.5,-3);
}
\draw (\dx,-2) -- (\dx,-3);

\foreach \i in {1,...,8}
{
	\draw [fill=green,ultra thick,blue!50] (-\dx-2+0.5+\i*4*0.5,-3) rectangle (-\dx-2+0.5+0.5+\i*4*0.5,-2);
}

\draw [->,thick] (\dx-0.25,-3.5) -- (\dx-0.25,-3);
\node at (\dx-0.25,-3.8) {$\text{region}(j,t)$};
\node at (\dx-0.25,-2.5) {$v_{j,t}$};

\draw [->,thick] (-\dx+0.25,-3.5) -- (-\dx+0.25,-3);
\node at (-\dx+0.25,-3.8) {$\text{region}(j,1)$};
\node at (-\dx+0.25,-2.5) {$v_{j,1}$};

\node at (-\dx+0.25+0.5,-2.5) {$v_{j,2}$};

\draw [->,thick] (2.5+0.25,-3.5) -- (2.5+0.25,-3);
\node at (2.5+0.25,-3.8) {$\text{region}(j,q')$};
\node at (2.5+0.25,-2.5) {$v_{j,q'}$};

\end{tikzpicture}
\caption{The number of ``blue'' regions is equal to the number of folds $m$. The number of total regions is $t$. For each observed $c_j$, instead of checking all the $t$ regions, we only assign vote to the these ``blue'' regions. Since only these $m$ ``blue'' regions can be the candidate region that contains frequency $f$.}
\end{figure}
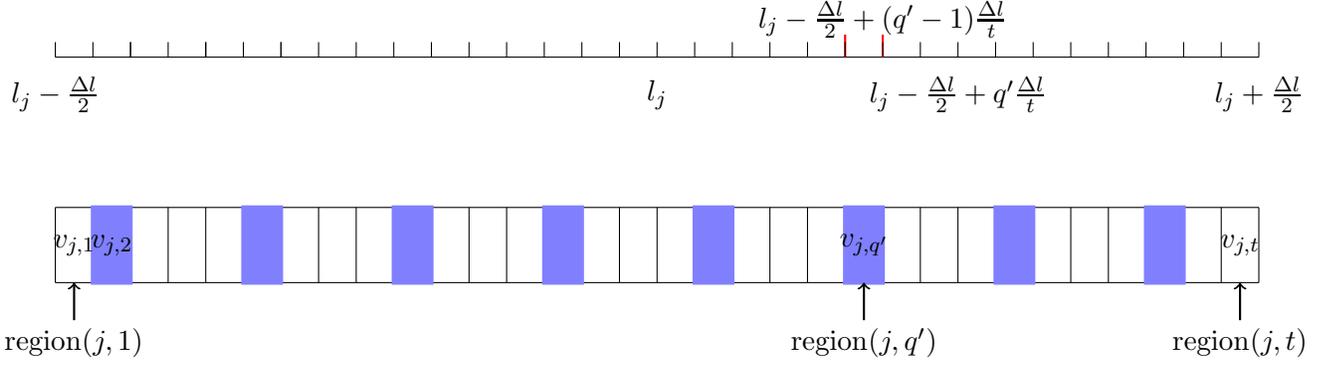

\restate{lem:correctness_of_locateouter}

\begin{proof}

Algorithm $\mathsf{LocateKSignal}$ rerun procedure $\mathsf{LocateInner}$ $D$ times. For the first $D -1$ rounds, the sampling range for $\beta$ is increased by $t$ every time. For the last round, the sampling range for $\beta$ is not increasing any more. On the other hand, the sampling range for $\beta$ for $D -1$ round and the last round are the same.

Recall that, we sample $\sigma$ uniformly at random from $[ \frac{1}{ B \eta}, \frac{ 2}{ B \eta}]$. Then we sample $\gamma$ uniformly at random from $[\frac{1}{2},1]$ and $\beta$ uniformly at random from $[\frac{st}{4\sigma \Delta l}, \frac{st}{2\sigma \Delta l}]$. $c_n$ is some constant for the number of neighbor regions nearby ``true'' region. In the first $D -1$ rounds, we set $s=1/\sqrt{C}$, $\Delta l = F/ (t')^{i-1}_0, \forall i \in [D -1]$, $t\eqsim \log(FT)$, and $t' = \frac{t}{c_n+1}$. For the last round, we set $s \eqsim 1/C$ , $\Delta l \eqsim st/T$ and $t\eqsim \log(FT)/s$. $C$ is known as ``approximation'' factor. Finally, we need to choose depth $D = \log_{t'} (FT/st) \eqsim \log_{t'}(FT)$ and set $R_{loc} = O( \log_{C} (tC) )$ for all the rounds. Define $ \mathit{DR} =\sum_{i=1}^D R_{loc}^i$, which is $ (D-1)R_{loc} + R_{loc}^{D} = O(\log_C (FT))$. 

After setting all the parameters for the first $D-1$ rounds and the last round, we explain some intuitions and motivations for setting the last round in a different way of the first $D-1$ rounds. For the first $D-1$ rounds, it is not acceptable to have constant failure probability for each round, since we need to take the union bound over $D-1$ rounds. But, for the last round, it is acceptable to allow just constant failure probability, since it is just a single round. That's the reason for setting $C$ in a different way.

We have the following reason for choosing the number of regions($=t$) in the last round larger than that of first $D-1$ rounds. For the first $D-1$ rounds, we do not need to learn frequency within $\eqsim \frac{1}{T\rho}$. It is enough to know which region does frequency belong to, although the diameter of the region is large at the beginning. The algorithm is making progress round by round, since the diameter of each region is geometrically decreasing while $\widehat{\beta}$ is geometrically increasing.
For the last round, by Lemma \ref{lem:learn_f_within}, we can learn $f$ within $\Theta(\frac{\Delta l}{\rho s t})$. Since after the last round, we hope to learn frequency within $\frac{1}{T\rho}$, thus we need to choose some $s$, $t$ and $\Delta l$ such that $\frac{1}{T} \eqsim  \frac{\Delta l }{st}$. To get more accuracy result in the last round, we'd like to choose a larger $t$. But there is no reason to increase $\widehat{\beta}$ again, since the  $\widehat{\beta}$ of the $(D-1)$th rounds can tolerance the $t$ we choose at the last round.

To show the constant succeed probability of this Lemma, we separately consider about the failure probability of the first $D-1$ rounds and the last round.
By the union bound, the probability of existing one of the first $D-1$ rounds is failing is,
\begin{eqnarray*}
& &(D-1)  \left( (\frac{4}{s\rho})^{R_{loc}} + t\cdot (60 s)^{R_{loc}/2} \right) \\
&\leq & D  \left( (\frac{4}{s\rho})^{R_{loc}} + t\cdot (60 s)^{R_{loc}/2} \right) \\
&\leq & D  \left( (\frac{4}{sC})^{R_{loc}} + t\cdot (60 s)^{R_{loc}/2} \right) ~{\text{by}~\rho>C>1}\\
&=& D  \left( (\frac{4}{\sqrt{C} })^{R_{loc}} + t\cdot ( \frac{60}{\sqrt{C}})^{R_{loc}/2} \right) ~{\text{by~setting}~s\eqsim 1/\sqrt{C}}\\
&\leq & D \cdot \frac{1}{ (C t)^{c}}   ~{\text{by~setting} ~R_{loc}=O(\log_C(tC))} ,
\end{eqnarray*}
where $c$ is some arbitrarily large constant. Using $t>D$, we can show that failure happening in any of the first $D-1$ rounds is small. Then, we still need to show that the probability of the last round is failing is also small,
\begin{eqnarray*}
& & (\frac{4}{s\rho})^{R_{loc}} + t\cdot (60 s)^{R_{loc}/2} + e^{ -\Theta( R_{loc} ) }\\
& \leq & ( \Theta( \frac{C}{\rho} ) )^{R_{loc}} + t\cdot ( \Theta( \frac{1}{C} ) )^{R_{loc}/2} + e^{ -\Theta( R_{loc} ) } ~{\text{by~setting}~s\eqsim 1/{C}} \\
& \leq & \frac{1}{c_1} + \frac{1}{(tC)^{c_2}} + e^{ -\Theta( R_{loc} ) } ~{\text{by~setting}   ~R_{loc}=O(\log_C(tC))} \\
& \leq & \frac{1}{c_1} + \frac{1}{(tC)^{c_2}} +  \frac{1}{c_3},
\end{eqnarray*}
where in the first line, the first two terms are from Lemma \ref{lem:correctness_of_locateinner} and the third term comes from Lemma \ref{lem:learn_f_within}; in the last line $c_1$, $c_2$ and $c_3$ are some arbitrarily large constants.

The expected running time includes the following part: Running $\mathsf{HashtoBins}$ algorithm $O(\mathit{DR})$ times, each run takes $O(\frac{B}{\alpha}\log\frac{k}{\delta} + B\log B)$. Updating the counter $v$, which takes $O(DR \cdot B t)$ time. The total running time should be 
\begin{eqnarray*}
& & O(\mathit{DR} ( \frac{B}{\alpha}\log\frac{k}{\delta} + B\log B ) + (D R_{loc} Bt)  ) \\
& = & O( \mathit{DR} B \log (k/\delta \cdot B \cdot FT)) \\
& = & O(B\log_{C} (FT) \log (\frac{k}{\delta} B FT)) \\
& = & O(B\log_{C} (FT) \log ( FT/\delta)) \quad \text{by~$FT \gg F \frac{1}{\eta} \gg k$}.
\end{eqnarray*}
 The total number of samples is 
 \begin{equation*}
 O\left( \mathit{DR} \cdot B \log (\frac{k}{\delta})\right) = O(B \log_{C} (FT) \log (k/\delta)).
 \end{equation*}
 The sample duration of Algorithm $\mathsf{LocateKSignal}$ is $ O(\frac{\log\frac{k}{\delta} }{\eta})$.

In conclusion, we can show that for frequency where neither $E_{\mathit{coll}}$ nor
$E_{\mathit{off}}$ holds, we recover an $f'$ with $\abs{f-f'} \lesssim
\frac{1}{T \rho}$ as long as $\rho >C$, with an arbitrarily large constant probability. 
\end{proof}

\restate{lem:whole_of_stage}
\begin{proof}
  Let $H$ denote the set of frequencies $f$ for which neither 
  $E_{\mathit{coll}}(f)$ nor $E_{\mathit{off}}(f)$ holds.  For each such $f$, let $v =
  \wh{x^*}(f)$ and denote
  \begin{equation}\label{eq:mu_equal_expectation}
  \mu^2(f) = \E_{a}[|\widehat{u}_j - ve^{a\sigma 2\pi f\i} |^2]
  \end{equation}
  and $\rho^2(f) = |v|^2 / \mu^2(f)$, as in
  Lemma~\ref{lem:correctness_of_locateouter}.  We have that, if
  $\rho^2(f) > C^2$, then with an arbitrarily large constant probability
  one of the recovered $f' \in L$ has $\abs{f' - f} \lesssim
  \frac{1}{T\rho}$.
  If this happens, then $\mathsf{OneStage}$ will estimate $v$ using $v'
  = \widehat{u}_je^{-a\sigma 2\pi f'\i}$. By triangle inequality,
  \begin{equation}\label{eq:first_bound_for_v_minus_vprime}
  \abs{v' - v}^2 \lesssim \abs{v}^2\abs{e^{a 2\pi \sigma (f' - f)\i} - 1}^2 + |\widehat{u}_j - ve^{a\sigma 2\pi f\i} |^2.
  \end{equation}
  Since $a\sigma \leq T$ and $|f'-f| \lesssim \frac{1}{T\rho}$, then the first term of $\mathrm{RHS}$ of Equation (\ref{eq:first_bound_for_v_minus_vprime}) have
  \begin{equation*}
   \abs{v}^2\abs{e^{a 2\pi \sigma (f' - f)\i} - 1}^2 \lesssim |v|^2 |a\sigma (f'-f)|^2.
  \end{equation*}
  For the second term of $\mathrm{RHS}$ of Equation (\ref{eq:first_bound_for_v_minus_vprime}). Using Equation (\ref{eq:mu_equal_expectation}), we have 
  \begin{equation*}
  |\widehat{u}_j - ve^{a\sigma 2\pi f\i} |^2  \lesssim  \mu^2(f),
  \end{equation*}
  with arbitrarily large constant probability. Combining the bounds for those two terms gives
  \begin{equation*}
  \abs{v' - v}^2 \lesssim  \abs{v}^2 \abs{a\sigma(f' - f)}^2 + \mu^2(f),
  \end{equation*}
  with arbitrarily large constant probability.  Since $a\sigma \leq
  T$, the first term is $\abs{v}^2/\rho^2 = \mu^2$, for
  \[
  \abs{v' - v}^2 \lesssim \mu^2(f).
  \]

  On the other hand, if $\rho^2(f) < C^2$, then $\abs{v} 
  = \rho(f) \mu(f) \lesssim C \mu(f)$ so regardless of the frequency $f'$ recovered, the estimate
  $v'$ will have
  \[
  \abs{v' - v}^2 \lesssim C^2 \mu^2(f).
  \]
  with arbitrarily large constant probability.

  Combining with Lemma~\ref{lem:two_close_signal}, we get for any $f \in H$ that
  the recovered $f', v'$ will have
  \[
  \frac{1}{T} \int_0^T \left|v' e^{2\pi f' t \i } - v e^{2\pi ft\i } \right|^2 \mathrm{d} t  \lesssim C^2 \mu^2(f).
  \]
  with arbitrarily large constant probability.  Let $S \subset H$ be
  the set of frequencies for which this happens.  We can choose our
  permutation $\pi$ to match frequencies in $S$ to their nearest
  approximation.
   By
  Lemma~\ref{lem:hbinserror}, this means that
  \[
  \E_{\sigma, b}\left[ \sum_{i\in S} \frac{1}{T} \int_0^T \left|v_i' e^{2\pi f_i' t \i } - v_{\pi(i)} e^{2\pi f_{\pi(i)}t\i } \right|^2 \mathrm{d} t \right] \lesssim C^2 \N^2.
  \]
  as desired.
\end{proof}

\section{Proofs for combining multiple stages}

We first prove that the median of a bunch of estimates of a frequency
has small error if most of the estimates have small error.
\begin{lemma}\label{lem:medians}
  Let $(v_i, f_i)$ be a set of tones for $i \in S$.  Define $v'$
  and $f'$ to be the (coordinate-wise) median of the $(v_i, f_i)$.
  Then for any $(v^*, f^*)$ we have
  \[
  \frac{1}{T} \int_0^T \left|v^* e^{2\pi f^* t \i } - v' e^{2\pi f't\i } \right|^2 \mathrm{d} t \lesssim \median_i \frac{1}{T} \int_0^T \left|v^* e^{2\pi f^* t \i } - v_{i} e^{2\pi f_it\i } \right|^2 \mathrm{d} t.
  \]
\end{lemma}
\begin{proof}
  By Lemma~\ref{lem:two_close_signal}  we have that
  \[
  \frac{1}{T} \int_0^T \left|v^* e^{2\pi f^* t \i } - v' e^{2\pi f't\i } \right|^2 \mathrm{d} t \eqsim \abs{v^*}^2\min(1, T^2 \abs{f^* - f'}^2) + \abs{v^* - v'}^2.
  \]
  Using that $v'$ is taken as a two dimensional median, it suffices to
  show: if $x^{(1)}, x^{(2)}, \dotsc \in \R^3$ then $x' = \underset{i}{\median} ~ x^{(i)}$ has
  \begin{align}
    \norm{x'}_2^2 \lesssim \median_i \norm{x^{(i)}}_2^2.\label{eq:medianmoo}
  \end{align}
  This follows because in each of the three coordinates $j$, we have
  \[
  (x'_j)^2 = (\median_i x^{(i)}_j)^2\leq \median_i (x^{(i)}_j)^2 \leq \median_i \norm{x^{(i)}}_2^2.
  \]
  so summing over the three coordinates gives~\eqref{eq:medianmoo}, as desired.

  Therefore, for two dimensional median and one dimension median, we have
  \begin{equation}\label{eq:magnitude_median}
  |v^* - v'|^2 \lesssim \median_i |v^* - v_i|^2
  \end{equation}
  and
  \begin{equation*}
  |f^* - f'|^2 \lesssim \median_i |f^* - f_i|^2.
  \end{equation*}
 Moreover, 
\begin{eqnarray}\label{eq:frequency_median}
|v^*|^2 \cdot \min(1,T^2 |f^* -f'|^2)& \lesssim& |v^*|^2 \cdot \min(1,T^2 \median_i |f^* -f_i|^2 ) \notag\\ 
& = & |v^*|^2 \cdot \min(1, \median_i T^2  |f^* -f_i|^2 ) \notag \\
& = & |v^*|^2  \cdot \median_i \min(1, T^2  |f^* -f_i|^2 )
\end{eqnarray}
Combining Equation (\ref{eq:magnitude_median}) and (\ref{eq:frequency_median}), we have
\begin{eqnarray*}
|v^* - v'|^2 + |v^*|^2 \cdot \min(1,T^2 |f^* -f'|^2) & \lesssim& \median_i |v^* - v_i|^2 +  \median_i |v^*|^2 \cdot \min(1,T^2 |f^* -f_i|^2 ) \\
& = &  \median_i |v^* - v_i|^2 +  |v^*|^2 \cdot \min(1,T^2 |f^* -f_i|^2 ).
\end{eqnarray*}
which completes the proof.
\end{proof}

\begin{figure}
\begin{tikzpicture}[>=stealth',shorten >=1pt,auto,node distance=3cm,
  thick,red node/.style={scale=.4,circle,fill=red,draw,font=\sffamily\bfseries},blue node/.style={scale=.4,circle,fill=blue,draw,font=\sffamily\bfseries},green node/.style={scale=.4,circle,fill=green,draw,font=\sffamily\bfseries},yellow node/.style={scale=.4,circle,fill=yellow,draw,font=\sffamily\bfseries},gray node/.style={scale=.4,circle,fill=gray,draw,font=\sffamily\bfseries}]
  [scale=.8,auto=left,every node/.style={thick,->,>=stealth',auto,circle,fill=blue!20}]
\def\x{8.2}
 	\draw [->] (0,0) -- (7,0); 
 	\draw [->] (0,0) -- (0,3);
 	\node at (0,-0.5) {$0$};
 	\node at (7,-0.5) {$F$};

	\node[gray node] (gray1) at (0.20,1.30) {};
	\node[gray node] (gray1) at (0.50,0.40) {};
	\node[gray node] (gray1) at (0.70,2.60) {};
	\node[gray node] (gray1) at (2.30,0.70) {};
	\node[gray node] (gray1) at (2.40,1.60) {};
	\node[gray node] (gray1) at (2.50,2.30) {};
	\node[gray node] (gray1) at (4.80,1.10) {};
 	\node[gray node] (gray1) at (4.95,0.20) {};
 	\node[gray node] (gray1) at (5.20,2.40) {};
 	\node[gray node] (gray1) at (6.50,1.10) {};
 	\node[gray node] (gray1) at (6.60,0.20) {};
 	\node[red node] (red1) at (1.05,1.10) {};
 	\node[red node] (red2) at (1.20,1.05) {};
 	\node[red node] (red3) at (1.15,1.35) {};
 	\node[red node] (red4) at (1.40,1.15) {};
 	\node[red node] (red5) at (1.35,1.00) {};
 	\node[red node] (red6) at (1.10,0.95) {};

 	\node[red node] (red7) at (0.90,1.20){};
 	\node[red node] (red8) at (1.60,1.10){};

 	\node[green node] (green1) at (3.05,0.50) {};
 	\node[green node] (green2) at (3.20,0.65) {};
 	\node[green node] (green3) at (3.15,0.85) {};
 	\node[green node] (green4) at (3.40,0.75) {};
 	\node[green node] (green5) at (3.35,0.70) {};
 	\node[green node] (green6) at (3.10,0.85) {};

	\node[green node] (green6) at (2.90,0.60){};
 	\node[green node] (green6) at (2.75,0.80){};
 	\node[green node] (green6) at (3.55,0.70){};
 	\node[green node] (green6) at (3.55,0.90){};
 	
 	\node[blue node] (blue1) at (4.05,2.60) {};
 	\node[blue node] (blue2) at (4.20,2.55) {};
 	\node[blue node] (blue3) at (4.15,2.35) {};
 	\node[blue node] (blue4) at (4.40,2.55) {};
 	\node[blue node] (blue5) at (4.35,2.40) {};
 	\node[blue node] (blue6) at (4.10,2.85) {};

 	\node[blue node] (blue7) at (4.50,2.60) {};
 	\node[blue node] (blue8) at (4.60,2.20) {};
 	\node[blue node] (blue9) at (3.95,2.40) {};

 	\node[yellow node] (yellow1) at (5.55,1.60) {};
 	\node[yellow node] (yellow2) at (5.70,1.45) {};
 	\node[yellow node] (yellow3) at (5.65,1.95) {};
 	\node[yellow node] (yellow4) at (5.90,1.55) {};
 	\node[yellow node] (yellow5) at (5.95,1.80) {};
 	\node[yellow node] (yellow6) at (5.60,1.85) {};
 	\node[yellow node] (yellow7) at (5.45,1.80) {};
 	\node[yellow node] (yellow8) at (5.40,1.85) {};
 	\node[yellow node] (yellow9) at (6.10,1.90) {};
 	\node[yellow node] (yellow10) at (6.05,1.55) {};

	\node[] at (7.3,2.5) {$\mathsf{OneStage}$};
 	\draw [->] (6.5,2) -- (8.5,2);

 	\draw [->] (\x+0,0) -- (\x+7,0); 
 	\draw [->] (\x+0,0) -- (\x+0,3);
 	\node at (\x+0,-0.5) {$0$};
 	\node at (\x+7,-0.5) {$F$};

	\node[gray node] (gray1) at (\x+0.20,1.30) {};
	\node[gray node] (gray1) at (\x+0.50,0.40) {};
	\node[gray node] (gray1) at (\x+0.70,2.60) {};
	\node[gray node] (gray1) at (\x+2.30,0.70) {};
	\node[gray node] (gray1) at (\x+2.40,1.60) {};
	\node[gray node] (gray1) at (\x+2.50,2.30) {};
	\node[gray node] (gray1) at (\x+4.80,1.10) {};
 	\node[gray node] (gray1) at (\x+4.95,0.20) {};
 	\node[gray node] (gray1) at (\x+5.20,2.40) {};
 	\node[gray node] (gray1) at (\x+6.50,1.10) {};
 	\node[gray node] (gray1) at (\x+6.60,0.20) {};
 	
 	\node[red node] (red1) at (\x+1.05,1.10) {};
 	\node[red node] (red2) at (\x+1.20,1.05) {};
 	\node[red node] (red3) at (\x+1.15,1.35) {};
 	\node[red node] (red4) at (\x+1.40,1.15) {};
 	\node[red node] (red5) at (\x+1.35,1.00) {};
 	\node[red node] (red6) at (\x+1.10,0.95) {};

	\node[red node] (red7) at (\x+0.90,1.20){};
 	\node[red node] (red8) at (\x+1.60,1.10){};
 	
 	\node[green node] (green1) at (\x+3.05,0.50) {};
 	\node[green node] (green2) at (\x+3.20,0.65) {};
 	\node[green node] (green3) at (\x+3.15,0.85) {};
 	\node[green node] (green4) at (\x+3.40,0.75) {};
 	\node[green node] (green5) at (\x+3.35,0.70) {};
 	\node[green node] (green6) at (\x+3.10,0.85) {};

 	\node[green node] (green6) at (\x+2.90,0.60){};
 	\node[green node] (green6) at (\x+2.75,0.80){};
 	\node[green node] (green6) at (\x+3.55,0.70){};
 	\node[green node] (green6) at (\x+3.55,0.90){};

 	\node[blue node] (blue1) at (\x+4.05,2.60) {};
 	\node[blue node] (blue2) at (\x+4.20,2.55) {};
 	\node[blue node] (blue3) at (\x+4.15,2.35) {};
 	\node[blue node] (blue4) at (\x+4.40,2.55) {};
 	\node[blue node] (blue5) at (\x+4.35,2.40) {};
 	\node[blue node] (blue6) at (\x+4.10,2.85) {};

 	\node[blue node] (blue7) at (\x+4.50,2.60) {};
 	\node[blue node] (blue8) at (\x+4.60,2.20) {};
 	\node[blue node] (blue9) at (\x+3.95,2.40) {};

 	\node[yellow node] (yellow1) at (\x+5.55,1.60) {};
 	\node[yellow node] (yellow2) at (\x+5.70,1.45) {};
 	\node[yellow node] (yellow3) at (\x+5.65,1.95) {};
 	\node[yellow node] (yellow4) at (\x+5.90,1.55) {};
 	\node[yellow node] (yellow5) at (\x+5.95,1.80) {};
 	\node[yellow node] (yellow6) at (\x+5.60,1.85) {};

 	\node[yellow node] (yellow7) at (\x+5.45,1.80) {};
 	\node[yellow node] (yellow8) at (\x+5.40,1.85) {};
 	\node[yellow node] (yellow9) at (\x+6.10,1.90) {};
 	\node[yellow node] (yellow10) at (\x+6.05,1.55) {};

 	\draw [dashed] (\x+1.0,2) -- (\x+1.0,-0.2);
 	\draw [dashed] (\x+1.5,2) -- (\x+1.5,-0.2);
 	\node[] at (\x+1.2,-0.5) {$c\eta$};

 	\node[] at (\x+2.3,-0.5) {$\gtrsim \eta$};
 	
 	\draw [dashed] (\x+3.0,2) -- (\x+3.0,-0.2);
 	\draw [dashed] (\x+3.5,2) -- (\x+3.5,-0.2);
 	\node [] at (\x+3.2,-0.5) {$c\eta$};

	\draw [dashed] (\x+4.0,2.5) -- (\x+4.0,-0.2);
 	\draw [dashed] (\x+4.5,2.5) -- (\x+4.5,-0.2);
 	\node [] at (\x+4.2,-0.5) {$c\eta$};

	\node [] at (\x+5.1,-0.5) {$\gtrsim \eta$};
	
 	\draw [dashed] (\x+5.5,2) -- (\x+5.5,-0.2);
 	\draw [dashed] (\x+6.0,2) -- (\x+6.0,-0.2);
 	\node [] at (\x+5.7,-0.5) {$c\eta$};
 	
\end{tikzpicture}
\caption{We demonstrate the algorithm for merging various stages($R=O(\log k)$) on 2-dimensional data. Note that the true data should be 3-dimensional, since for each tone $(v_i,f_i)$, $v_i \in \mathbb{C}$ and $f_i \in \mathbb{R}$. The $x$-axis represents the frequency and the $y$-axis represents the real part of magnitude. }
\end{figure}
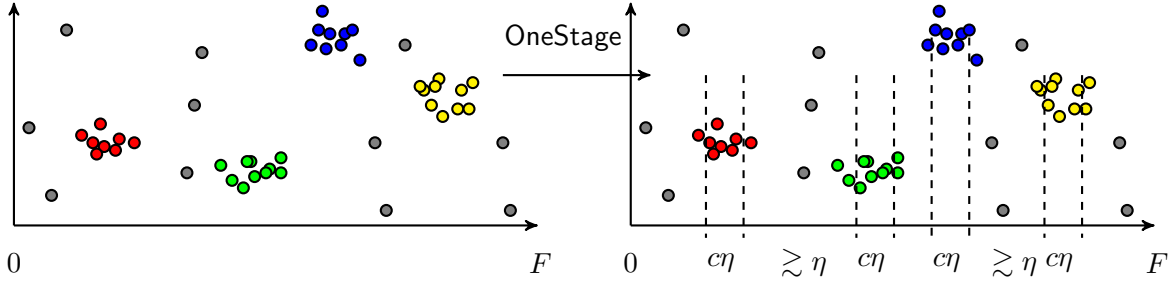

\restate{lem:onemerged}
\begin{proof}
  Our only goal here is to recover the actual tones well, not to
  worry about spurious tones.

  Suppose the number of stages we perform is $R = O(\log k)$.  The
  algorithm for merging the various stages is to scan over a $c\eta$
  size region for small $c$, and take the median (in both frequency
  and magnitude) over $3c\eta$ region around that $c\eta$ 
  if there are at least $\frac{6}{10}R$ results in that $c\eta$
  region. If so, the algorithm will jump to the first right point that
  is at least $\eta$ far away from current region and look for the next $c\eta$ region. 
   Because the minimum separation between frequencies for a
  given stage is $\Omega(\eta)$, this will have minimum separation
  $\eta$ in the output, and because there are $O(k)$ tones
  output at each stage so will this method.  What remains is to show
  that the total error is small.

  We say a stage is ``good'' if the term inside the expectation of
  Lemma~\ref{lem:whole_of_stage} is less than 10 times its
  expectation, as happens with $9/10$ probability:
\begin{equation*}
\mathsf{Pr} \left( \sum_{i\in S} \frac{1}{T} \int_0^T \left|v_i' e^{2\pi f_i' t \i } - v_{\pi(i)} e^{2\pi f_{\pi(i)}t\i } \right|^2 \mathrm{d} t  \leq 10 \underset{\sigma,b}{\mathbb{E}}\left[ \sum_{i\in S} \frac{1}{T} \int_0^T \left|v_i' e^{2\pi f_i' t \i } - v_{\pi(i)} e^{2\pi f_{\pi(i)}t\i } \right|^2 \mathrm{d} t  \right]  \right) \geq \frac{9}{10}.
\end{equation*}
    We say a frequency
  $f_i$ is ``successful'' in a given stage if the stage is good and
  $i$ lies in $S$ for that stage.  Therefore a frequency is successful
  in each stage with at least $8/10$ probability.  For sufficiently
  large $R = O(\log k)$, this will cause all $k$ frequencies to be
  successful more than $\frac{7}{10}R$ times with high probability.
  Suppose this happens.

  Let $\mu^2 (f_i)$ denote the error of $(v_i,f_i)$.
  By Lemma~\ref{lem:whole_of_stage}, the total error over all good
  stages and every successful recovery of a tone in a good stage
  is $O(C^2 \N^2 R)$.  We define $\mu^2(f)$ to be the $6/10R$ worst amount
  of error in the recovery of $f$ over all stages.  Because there are
  $R/10$ worse successful recoveries for each $f$, we have that
  $\sum_i \mu^2(f_i) \lesssim C^2 \N^2$. 


  We will match each $f_i$ with a recovered frequency $f'_i$ with cost
  at most $\mu^2(f_i)$.  If $\mu^2(f_i) \gtrsim \abs{v_i}^2$, then we
  can set an arbitrary $f'_i$ with $v'_i = 0$.  Otherwise, more that
  $\frac{6}{10}R$ successful recoveries of $f_i$ also yield $f'$ that
  are within $O(\frac{1}{T}) \ll c \eta$ of $f_i$.  Thus the algorithm
  for merging tones will find enough tones to report
  something.  What it reports will be the median of at most $R$
  values, $6/10$ of which have less error than $\mu^2(f)$.  Therefore
  by Lemma~\ref{lem:medians} the reported frequency and magnitude will
  have error $O(\mu^2(f))$.  This suffices to get the result.
\end{proof}

\restate{lem:prune_twice}
\begin{proof}
  By Lemma~\ref{lem:onemerged} and Lemma~\ref{lem:two_close_signal}, the
  first run gives us a set of $k'=O(k)$ pairs $\{ (v_i',f_i')\}$ such
  that they may be indexed by using permutation $\pi$ such that the first $k$ are a good
  approximation to $\{(v_i, f_i)\}$ in the sense that
  \begin{equation*}
    \sum_{i=1}^k (|v_{\pi(i)}|^2 + |v'_i|^2)\cdot \min (1,T^2\abs{f_{\pi(i)}-f'_i}^2) + |v_{\pi(i)} -v'_i|^2 \lesssim C^2 \N^2.
  \end{equation*}
  We may as well index to match tones to their nearest match in frequency
  (because the separation for both $f'_i$ and $f_i$ is at least
  $\Omega(\eta) > 1/T$).  That is, we may index such that $\abs{f_j' -
    f_i} \gtrsim \eta$ for any $j \neq i$.

  Now, for indices where $\abs{f_i' - f_{\pi(i)}} > 1/T$, one would do better
  by setting the corresponding $v_i$ to zero.  In particular, suppose
  you knew the true frequencies $f_{\pi(i)}$ and only took the $(v_i', f_i')$
  where $f_i'$ is close to $f_{\pi(i)}$.  That is, there exists some permutation $\pi$,
  for the subset $S \subseteq [k]$
  containing $\{i : \abs{f_{\pi(i)} - f_i'} \leq c/T\}$ for any $c \gtrsim 1$, we have
  \begin{align}
    \sum_{i \in S} ((|v_{\pi(i)}|^2 + |v'_i|^2)\cdot \min
    (1,T^2\abs{f_{\pi(i)}-f'_i}^2) + |v_{\pi(i)} -v'_i|^2) + \sum_{i \in [k]
      \setminus S} (\abs{v_{\pi(i)}}^2 + \abs{v_i'}^2) \lesssim C^2 \N^2. \label{eq:sumS}
  \end{align}
  The problem is that we do not know the set $S$, so we can not throw out
  the other frequencies.  However, we can fake it by running the
  algorithm again on the signal.  In particular, we apply Lemma~\ref{lem:onemerged} 
  again to sparse recovery of the signal defined by
  \[
  \{(v_i, f_i) \mid i \in [k]\} \cup \{(0, f'_i) \mid i \in \{k+1, \dotsc, k'\}\}.
  \]
  That is, we pretend the ``signal'' has terms at the other $f'_i$ for
  $k < i \leq k'$, but with magnitude zero.  This is an identical
  signal in time domain, so it's really just an analytical tool; for
  the analysis, it has $k'=O(k)$ sparsity and $\Omega(\eta)$
  separation, so Lemma~\ref{lem:onemerged} applies again and gets a set of $k''=O(k)$
  pairs $\{(v''_j, f''_j)\}$ such that, for the subset $S^* \subset
  [k']$ containing the indices $i$ where $\abs{f''_j - f'_i}
  \leq c/T$ for some $j\in [k'']$. 
  In other words, there exists some permutation $\tau$ such that
  for the subset $S^* \subset [k'] $ containing $\{i: |f''_{\tau(i)} - f'_i| \leq c/T \}$.
  We have by analogy to~\eqref{eq:sumS} that
\begin{eqnarray*}
C^2 \N^2 & \gtrsim & \sum_{i \in S^* \cap [k]} ((|v_{\pi(i)}|^2 + |v''_{\tau(i)}|^2)\cdot \min (1,T^2\abs{f_{\pi(i)}-f''_{\tau(i)}}^2) + |v_{\pi(i)} -v''_{\tau(i)}|^2) \\
&+& \sum_{i \in S^* \setminus [k]}( (|0|^2 + |v''_{\tau(i)}|^2 ) \cdot \min (1,T^2\abs{f'_i-f''_{\tau(i)}}^2) + |v''_{\tau(i)} - 0|^2) \\
& +&\sum_{i \in [k] \setminus S^* }  (\abs{v_{\pi(i)}}^2 + \abs{v''_{\tau(i)}}^2) + \sum_{i \in [k'] \setminus S^* \setminus [k]}  (\abs{v''_{\tau(i)}}^2) \\
& \gtrsim & \sum_{i \in S^* \cap [k]} ((|v_{\pi(i)}|^2 + |v''_{\tau(i)}|^2)\cdot \min (1,T^2\abs{f_{\pi(i)}-f''_{\tau(i)}}^2) + |v_{\pi(i)} -v''_{\tau(i)}|^2) \\
&+& \sum_{i \in S^* \setminus [k]}|v''_{\tau(i)}|^2  + \sum_{i \in [k] \setminus S^*}  \abs{v_{\pi(i)}}^2, \\
\end{eqnarray*}
where $\pi$ and $\tau$ are two permutations and $|S^*| = k^* = O(k)$.  This last term is precisely the desired total error for the set of
  tones $\{(v''_{\tau(i)}, f''_{\tau(i)}) : i \in S^*\}$, giving the result.

\end{proof}

To prove Theorem \ref{thm:get_eq2}, we still need the following ``local'' Lemma.
\begin{lemma}\label{lem:three_close_signal}
For any three tones $(v_{\pi(i)}, f_{\pi(i)})$, $(v^*_i,f^*_i)$ and $(v'_i,f'_i)$, if $|v'_i| \geq |v^*_i|$ then, 
\begin{eqnarray*}
&&(|v'_i|^2 + |v_{\pi(i)} |^2 ) \cdot \min(1, T^2 |f'_i- f_{\pi(i)}|^2) + |v'_i - v_{\pi(i)}|^2 \\
&\lesssim&  (|v^*_i|^2 + |v_{\pi(i)} |^2 ) \cdot \min(1, T^2 |f^*_i- f_{\pi(i)}|^2) + |v^*_i - v_{\pi(i)}|^2  + |v'_i|^2 .\\
\end{eqnarray*}
\end{lemma}
\begin{proof}

First, we can show an upper bound for $\mathrm{LHS}$. Using inequality $\min(1,T^2 |f'_i - f_{\pi(i)} |^2 ) \leq 1$, 
\begin{equation*}
\mathrm{LHS} \leq |v'_i|^2 + |v_{\pi(i)}|^2  + | v'_i - v_{\pi(i)} |^2.
\end{equation*}
By triangle inequality,
\begin{equation*}
| v'_i - v_{\pi(i)} |^2 \leq 2 |v'_i|^2 + 2|v_{\pi(i)}|^2.
\end{equation*}
Thus, we obtain,
\begin{equation*}
\mathrm{LHS} \lesssim |v'_i|^2 + |v_{\pi(i)}|^2 .
\end{equation*}
Second, we can show a lower bound for $\mathrm{RHS}$. Since first term of $\mathrm{RHS}$ is nonnegative, then
\begin{equation*}
\mathrm{RHS} \geq |v_i^* - v_{\pi(i)}|^2 + |v'_i|^2.
\end{equation*}
Using $|v'_i| \geq |v^*_i|$, we have
\begin{equation*}
\mathrm{RHS} \gtrsim |v'_i|^2 +2 |v^*_i|^2 +2|v_i^* - v_{\pi(i)}|^2.
\end{equation*}
By triangle inequality,
\begin{equation*}
2 |v^*_i|^2 +2|v_i^* - v_{\pi(i)}|^2 \geq | v_{\pi(i)}|^2.
\end{equation*}
Then, we prove the lower bound for $\mathrm{RHS}$,
\begin{equation*}
\mathrm{RHS} \gtrsim |v'_i|^2 +|v_{\pi(i)}|^2.
\end{equation*}
Combining the lower bound of $\mathrm{RHS}$ and the upper bound of $\mathrm{LHS}$ completes the proof.

\end{proof}

By plugging Lemma \ref{lem:two_close_signal} into Lemma \ref{lem:three_close_signal}, we have 
\begin{corollary}\label{cor:three_close_signal}
For any three tones $(v_{\pi(i)}, f_{\pi(i)})$, $(v^*_i,f^*_i)$ and $(v'_i,f'_i)$, if $|v'_i| \geq |v^*_i|$, then
\begin{equation*}
\frac{1}{T} \int_0^T \left| v'_i e^{2\pi f'_i t \i} - v_{\pi(i)}e^{2\pi f_{\pi(i)} t \i} \right|^2 \mathrm{d} t \lesssim \frac{1}{T} \int_0^T \left| v^*_i e^{2\pi f^*_i t \i} - v_{\pi(i)}e^{2\pi f_{\pi(i)} t \i} \right|^2 \mathrm{d} t + |v'_i|^2.
\end{equation*}
\end{corollary}

We use Corollary \ref{cor:three_close_signal} to present the proof of Theorem \ref{thm:get_eq2},
\restate{thm:get_eq2}
\begin{proof}
Let $\{(v^*_i,f^*_i)\}_{i=1,2,\cdots,k''}$ denote the set of tones returned by Lemma \ref{lem:prune_twice}, where $k''= O(k)$ and $k'' > k$. For any $i\in[k]$, tone $(v^*_i,f^*_i)$ was mapped to tone $(v_{\pi(i)}, f_{\pi(i)})$ in Lemma \ref{lem:prune_twice}. Let $\{(v'_i, f'_i)\}_{i=1,2,\cdots,k}$  denote a subset of $\{(v^*_i,f^*_i)\}_{i=1,2,\cdots,k''}$ that satisfies the following two conditions (1) for any $i\in [k]$, $|v'_i|$ is  one of the top-k largest magnitude tones; (2) for any $i\in [k]$, $|v'_i| \geq |v^*_i|$. By Lemma \ref{lem:prune_twice}, we also know that $\min_{i\neq j} |f'_i -f'_j|\gtrsim \eta$.


For any $i\in [k]$, we consider these three tones $(v_{\pi(i)},f_{\pi(i)}), (v^*_i,f^*_i), (v'_i, f'_i)$. If $ (v^*_i,f^*_i) \neq (v'_i, f'_i)$, then applying Corollary \ref{cor:three_close_signal} we have

\begin{equation*}
\frac{1}{T} \int_0^T \left| v'_i e^{2\pi f'_i t \i} - v_{\pi(i)}e^{2\pi f_{\pi(i)} t \i} \right|^2 \mathrm{d} t \lesssim \frac{1}{T} \int_0^T \left| v^*_i e^{2\pi f^*_i t \i} - v_{\pi(i)}e^{2\pi f_{\pi(i)} t \i} \right|^2 \mathrm{d} t + |v'_i|^2.
\end{equation*}
Otherwise $ (v^*_i,f^*_i) = (v'_i, f'_i)$, we also have
\begin{eqnarray*}
\frac{1}{T} \int_0^T \left| v'_i e^{2\pi f'_i t \i} - v_{\pi(i)}e^{2\pi f_{\pi(i)} t \i} \right|^2 \mathrm{d} t & = & \frac{1}{T} \int_0^T \left| v^*_i e^{2\pi f^*_i t \i} - v_{\pi(i)}e^{2\pi f_{\pi(i)} t \i} \right|^2 \mathrm{d} t,
\end{eqnarray*}
which means there exists some universal constant $\alpha$ such that $\forall i \in [k]$, if $(v^*_i,f^*_i) = (v'_i, f'_i)$ then
\begin{equation*}
\frac{1}{T} \int_0^T \left| v'_i e^{2\pi f'_i t \i} - v_{\pi(i)}e^{2\pi f_{\pi(i)} t \i} \right|^2 \mathrm{d} t \leq \alpha \cdot \left( \frac{1}{T} \int_0^T \left| v^*_i e^{2\pi f^*_i t \i} - v_{\pi(i)}e^{2\pi f_{\pi(i)} t \i} \right|^2 \mathrm{d} t + |v'_i|^2 \right)
\end{equation*}
holds. Otherwise $(v^*_i,f^*_i) \neq (v'_i, f'_i)$, then
\begin{eqnarray*}
\frac{1}{T} \int_0^T \left| v'_i e^{2\pi f'_i t \i} - v_{\pi(i)}e^{2\pi f_{\pi(i)} t \i} \right|^2 \mathrm{d} t & = & \alpha \frac{1}{T} \int_0^T \left| v^*_i e^{2\pi f^*_i t \i} - v_{\pi(i)}e^{2\pi f_{\pi(i)} t \i} \right|^2 \mathrm{d} t
\end{eqnarray*}
holds. Let $S$ denote the set of indices $i$ such that $(v^*_i,f^*_i) \neq (v'_i, f'_i)$.
Taking the summation from $i=1$ to $i=k$,
\begin{eqnarray*}
\sum_{i=1}^k\frac{1}{T} \int_0^T \left| v'_i e^{2\pi f'_i t \i} - v_{\pi(i)}e^{2\pi f_{\pi(i)} t \i} \right|^2 \mathrm{d} t & \leq & \sum_{i\in S} \alpha \cdot \left( \frac{1}{T} \int_0^T \left| v^*_i e^{2\pi f^*_i t \i} - v_{\pi(i)}e^{2\pi f_{\pi(i)} t \i} \right|^2 \mathrm{d} t + |v'_i|^2 \right) \\
& + &  \sum_{i\in [k] \setminus S} \alpha \cdot \left( \frac{1}{T} \int_0^T \left| v^*_i e^{2\pi f^*_i t \i} - v_{\pi(i)}e^{2\pi f_{\pi(i)} t \i} \right|^2 \mathrm{d} t \right),
\end{eqnarray*}
furthermore, we have
\begin{eqnarray*}
 \sum_{i=1}^k\frac{1}{T} \int_0^T \left| v'_i e^{2\pi f'_i t \i} - v_{\pi(i)}e^{2\pi f_{\pi(i)} t \i} \right|^2 \mathrm{d} t &\leq& \alpha \sum_{i=1}^k   \left( \frac{1}{T} \int_0^T \left| v^*_i e^{2\pi f^*_i t \i} - v_{\pi(i)}e^{2\pi f_{\pi(i)} t \i} \right|^2 \mathrm{d} t \right) + \alpha \sum_{i\in S} | v'_i |^2.
 \end{eqnarray*}
To finish the proof, we need to show that $ \sum_{i\in S} | v'_i |^2 \leq \sum_{i=k+1}^{k''} |v^*_i|^2$. The point is, for any $i\in S$, we know that $(v'_i, f'_i) \neq (v^*_i, f^*_i)$ which implies that  $(v'_i, f'_i) \notin \{ (v^*_i, f^*_i) \}_{i=1}^{k} $. Thus,
 \begin{eqnarray*}
& & \sum_{i=1}^k\frac{1}{T} \int_0^T \left| v'_i e^{2\pi f'_i t \i} - v_{\pi(i)}e^{2\pi f_{\pi(i)} t \i} \right|^2 \mathrm{d} t \\
&\leq& \alpha \sum_{i=1}^k   \left( \frac{1}{T} \int_0^T \left| v^*_i e^{2\pi f^*_i t \i} - v_{\pi(i)}e^{2\pi f_{\pi(i)} t \i} \right|^2 \mathrm{d} t \right) + \alpha  \sum_{i=k+1}^{k''} |v^*_i|^2 \\
&\lesssim&  \sum_{i=1}^k   \left( \frac{1}{T} \int_0^T \left| v^*_i e^{2\pi f^*_i t \i} - v_{\pi(i)}e^{2\pi f_{\pi(i)} t \i} \right|^2 \mathrm{d} t \right) +  \sum_{i=k+1}^{k''} |v^*_i|^2 \\
&\lesssim& C^2 \N^2,
\end{eqnarray*}
where $k''=O(k)$ is defined in Lemma \ref{lem:prune_twice} and the last inequality follows by using Equation (\ref{eq:prune_twice}) in Lemma \ref{lem:prune_twice}.






\end{proof}

\section{Proofs for converting~\eqref{eq:2} into~\eqref{eq:3}}
In this section, we show that as long as the sample duration $T$ is
sufficiently large, it is possible to convert Equation (\ref{eq:2}) to
Equation (\ref{eq:3}). First, we show an auxiliary lemma,
Lemma~\ref{lem:bound_from_infty_to_infty}, which bounds an integral
that will appear in the analysis.

We will show that
\begin{equation*}
\int^{+\infty}_{-\infty} \min(T,\frac{1}{|f_i-\theta|} ) \cdot \min(T,\frac{1}{|f_j-\theta|}) \mathrm{d} \theta \lesssim \frac{\log (T |f_i-f_j|)}{ |f_i - f_j |}.
\end{equation*}
for $f_j - f_i \geq 2/T$.  We split this into two pieces.

\begin{claim}\label{cla:bound_from_Ti_to_Ti}
Given two frequencies $f_i,f_j$ and $f_j - f_i \geq \frac{2}{T}$, we have
\begin{equation*}
\int_{f_i-\frac{1}{T} }^{f_i } \min(T, \frac{1}{|f_i-\theta|}) \cdot \frac{1}{|f_j - \theta |} \mathrm{d} \theta \lesssim \frac{1}{f_j -f_i}.
\end{equation*}
\end{claim}
\begin{proof}
By $f_i - \frac{1}{T} <\theta < f_i$, we have
\begin{equation*}
\mathrm{LHS} = \int^{f_i }_{f_i - \frac{1}{T} } T \cdot \frac{1}{f_j -\theta} \mathrm{d} \theta.
\end{equation*}
Since $\frac{1}{f_j-\theta}  \eqsim \frac{1}{f_j -f_i}$ for all $\theta \in [f_i-\frac{1}{T}, f_i]$, 
\begin{equation*}
\mathrm{LHS} \lesssim \int_{f_i-\frac{1}{T}}^{f_i} \frac{T}{ f_j - f_i} \mathrm{d} \theta = \frac{1}{f_j -f_i}.
\end{equation*}
\end{proof}

\begin{claim}\label{cla:bound_from_infty_to_Ti}
Given two frequencies $f_i,f_j$ and $f_j - f_i\geq \frac{2}{T}$, we have
\begin{equation*}
\int^{f_i-\frac{1}{T} }_{-\infty } \min(T, \frac{1}{|f_i-\theta|}) \cdot \frac{1}{|f_j - \theta |} \mathrm{d} \theta \lesssim \frac{\log(T |f_j-f_i| )}{f_j -f_i}.
\end{equation*}
\end{claim}
\begin{proof}
By $\theta < f_i -\frac{1}{T} < f_j$, we have that
\begin{eqnarray*}
\mathrm{LHS} & = & \int_{-\infty}^{f_i - \frac{1}{T} } \frac{1}{f_i - \theta } \cdot \frac{1}{f_j - \theta } \mathrm{d} \theta\\
&=& \frac{1}{f_j-f_i} \int_{-\infty}^{f_i - \frac{1}{T}} \frac{f_j-f_i}{ (f_i-\theta)(f_j -\theta)} \mathrm{d} \theta \\
& = & \frac{1}{f_j-f_i} \int_{-\infty}^{f_i - \frac{1}{T}} \frac{1}{ f_i-\theta} -\frac{1}{f_j -\theta} \mathrm{d} \theta \\
&= & -\frac{1}{f_j - f_i} \log \frac{f_i -f_i+\frac{1}{T} }{ f_j - f_i + \frac{1}{T} } \\
&= & -\frac{1}{f_j - f_i} \log \frac{1 }{ T(f_j - f_i) + 1 } \\
&\lesssim& \frac{\log (T(f_j-f_i)) }{f_j -f_i} . \\
\end{eqnarray*}\\
\end{proof}

\begin{lemma}\label{lem:bound_from_infty_to_infty}
Given two frequencies $f_i,f_j$ and $f_j - f_i\geq \frac{2}{T}$, we have 
\begin{equation*}
\int^{+\infty}_{-\infty} \min(T,\frac{1}{|f_i-\theta|} ) \cdot \min(T,\frac{1}{|f_j-\theta|}) \mathrm{d} \theta \lesssim \frac{\log (T |f_i-f_j|)}{ |f_i - f_j |}.
\end{equation*}
\end{lemma}

\begin{proof}
By symmetry, we have
\begin{equation*}
\mathrm{LHS} = 2 \int_{-\infty}^{\frac{f_i+f_j}{2}} \min(T,\frac{1}{|f_i-\theta|} ) \cdot \min(T,\frac{1}{|f_j-\theta|}) \mathrm{d} \theta.
\end{equation*}
Since $T > \frac{1}{f_j -\theta}$ when $\theta < \frac{f_i+f_j}{2}$, 
\begin{equation*}
\mathrm{LHS} \leq 2 \int_{-\infty}^{\frac{f_i+f_j}{2}} \min(T,\frac{1}{|f_i-\theta|} ) \cdot \frac{1}{|f_j-\theta|} \mathrm{d} \theta.
\end{equation*}
We also observe that $\frac{1}{|f_j - \theta |} \eqsim \frac{1}{|f_j-f_i|}$ for all $\theta \in [f_i - \frac{f_j-f_i}{2}, f_i + \frac{f_j-f_i}{2}]$,
\begin{equation*}
\int_{f_i - \frac{f_j -f_i}{2} }^{\frac{f_j+f_i}{2}} \min(T,\frac{1}{|f_i-\theta|} ) \cdot \frac{1}{|f_j-\theta|} \mathrm{d} \theta \eqsim \int_{f_i - \frac{f_j -f_i}{2} }^{f_i} \min(T,\frac{1}{|f_i-\theta|} ) \cdot \frac{1}{|f_j-\theta|} \mathrm{d} \theta.
\end{equation*}
Thus, we get
\begin{equation*}
\mathrm{LHS} \lesssim \int_{-\infty}^{f_i} \min (T, \frac{1}{|f_i-\theta|}) \cdot \frac{1}{|f_j-\theta|} \mathrm{d} \theta.
\end{equation*}
Plugging Claim \ref{cla:bound_from_Ti_to_Ti} and \ref{cla:bound_from_infty_to_Ti} into the above formula completes the proof.
\end{proof}

\begin{lemma}\label{lem:sum_of_absolute}
For any $i$, let $a_i(t) = v_i e^{2\pi f_i t \i} - v'_i e^{2\pi f'_i t \i}$, then for $i\neq j$, 
\begin{equation}
\frac{1}{T} \int_0^T a_i(t) \overline{ a_j(t) } \mathrm{d} t \lesssim \frac{\log ( \Delta f_{i,j} T)}{ \Delta f_{i,j} T} \cdot \left( \frac{1}{T} \int_0^T |a_i(t)|^2 \mathrm{d} t \cdot \frac{1}{T} \int_0^T |a_j(t)|^2 \mathrm{d} t \right)^{\frac{1}{2}},
\end{equation}
where $\Delta f_{i,j}=\min ( |f_i-f_j|, |f_i-f'_j|, |f'_i - f_j|, |f_i-f'_j| )$.
\end{lemma}

\begin{figure}[!t]
\begin{tikzpicture}
\def\PI{2*3.14159265359}
\def\s{2}
\def\dy{4}
    \node at (2,1) {$F(\theta)$};
    \node at (-8,-0.5) {$-\infty$};
    \node at (8,-0.5) {$+\infty$};
    \draw [domain=-0.05:-8,variable=\t,smooth,samples=500]
    	plot (\t,{sin(\s*\t r)/\t});
    \draw [domain=0.05:8,variable=\t,smooth,samples=500]
    	plot (\t,{sin(\s*\t r)/\t});
    \draw [blue] (-8,0) -- (8,0);
    \draw [blue] (0,-1.5) -- (0,2.5);
	\node at (2,1-\dy) {$F'(\theta)$};
	\node at (-8,-0.5-\dy) {$-\infty$};
    \node at (8,-0.5-\dy) {$+\infty$};
      \draw [domain=-0.05:-8,variable=\t,smooth,samples=500]
    	plot (\t,{-\dy+(\s*\t*cos(\s*\t r) - sin(\s*\t r))/(\t*\t) });
    \draw [domain=0.05:8,variable=\t,smooth,samples=500]
    	plot (\t,{-\dy+(\s*\t*cos(\s*\t r) - sin(\s*\t r))/(\t*\t) });
    \draw [blue] (-8,0-\dy) -- (8,0-\dy);
    \draw [blue] (0,-1.5-\dy) -- (0,2-\dy);
\end{tikzpicture}\caption{$F(\theta)$ is a $\sinc$ function and has derivate function $F'(\theta)$.}
\end{figure}
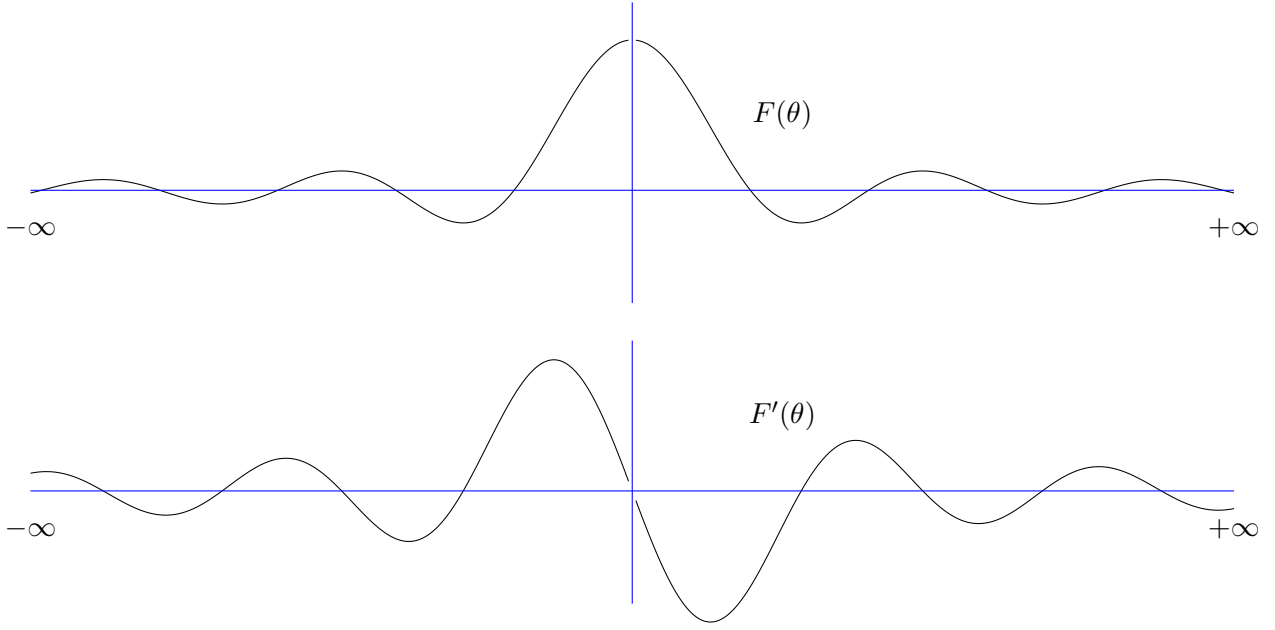

\begin{proof}
Let $\nu_i = |f_i - f'_i|$ and $\nu_j = |f_j - f'_j|$. Define $\| a_i\| = \sqrt{ \frac{1}{T} \int_0^T |a_i(t) |^2 \mathrm{d} t}$.
We define $f(t)$ and $F(\theta)$ to be a rectangle and $\sinc$ function respectively:

\begin{equation*}
f(t) = \left\{
\begin{array}{l l}
     1 & \quad \text{if $ 0 \leq t\leq T$}\\
     0 & \quad \text{otherwise}
  \end{array}
\right.
\end{equation*}

\begin{equation*}
F(\theta) = \frac{\sin( 2\pi \theta T )}{2\pi \theta}
\end{equation*}
where $F = \widehat{f}$.

$F(\theta)$ has the derivative,
\begin{equation*}
F'(\theta) = \frac{ 2\pi\theta T \cos (2\pi\theta T) - \sin (2\pi \theta T)}{2\pi \theta^2},
\end{equation*}
which means that
\begin{equation*}
|F'(\theta) | \lesssim \left\{
\begin{array}{l l}
     T^2 & \quad \text{if $ \theta \leq 1/T$}\\
     T/|\theta| & \quad \text{otherwise}
  \end{array}
\right.
\end{equation*}
Let $y_i(t) = a_i(t) \cdot f(t)$, then
\begin{eqnarray*}
\widehat{y}_i(\theta) &= &\widehat{a}_i(\theta) * \widehat{f}(\theta) \\
 &= &\widehat{a}_i(\theta) * F(\theta) ~\text{by}~F=\widehat{f}\\
& = & v_i F(f_i - \theta) - v'_{i} F(f'_{i} -\theta)\\
& = & (v_i - v'_{i} ) F(f_i -\theta) + v'_{i} (f_i - f_i') \cdot F'(x-\theta) \text{~some $x \in [f_i,f'_{i} ] $}.
\end{eqnarray*}

We split into two cases.  First, if $\nu_i \leq \frac{1}{T}$, then
\begin{eqnarray}\label{eq:bound_y_i_theta}
|\widehat{y}_i(\theta)|& \lesssim & (|v_i - v'_i| + \nu_i T | v'_i|) \cdot \min (T, \frac{1}{|f_i -\theta|}) \notag \\
& \lesssim & \| a_i \| \cdot \min(T,\frac{1}{|f_i -\theta|}) \quad \text{by~Lemma \ref{lem:two_close_signal}},
\end{eqnarray}
where the first line holds for both $f_i > f_i'$ and $f_i \leq f_i'$ since the triangle inequality.  Therefore
\begin{eqnarray}\label{eq:bound_diagonal_by_integral_two}
& & \frac{1}{T} \int_0^T a_i(t) \overline{a_j(t)} \mathrm{d} t \notag \\
& = & \frac{1}{T} \int_{-\infty}^{\infty} y_i(t) \overline{y_j(t)} \mathrm{d} t ~\text{by}~ y_i(t) = a_i(t) \cdot f(t) ~\text{and}~ f(t)=1 ~\text{if}~ t\in [0,T] \notag \\
& = & \frac{1}{T} \int_{-\infty}^{\infty} \widehat{y}_i(\theta) \overline{\widehat{y}_j(\theta)} \mathrm{d} \theta ~\text{by~the~property~of~Fourier~Transform}\notag \\
& \lesssim & \frac{1}{T} \|a_i\| \|a_j\| \underbrace{\int_{-\infty}^{+\infty} \min(T,\frac{1}{|f_i-\theta|}) \cdot \min(T,\frac{1}{|f_j-\theta|}) \mathrm{d} \theta}_{C} ~\text{by~Equation~(\ref{eq:bound_y_i_theta})}.
\end{eqnarray}
Using Lemma \ref{lem:bound_from_infty_to_infty}, we have following bound for term $C$,

\begin{equation*}
\int_{-\infty}^{+\infty} \min(T, \frac{1}{|f_i-\theta|}) \cdot \min(T,\frac{1}{|f_j-\theta|}) \mathrm{d} \theta \lesssim \frac{\log T|f_j-f_i|}{|f_j - f_i| }. 
\end{equation*}

\begin{figure}
\centering
\begin{tikzpicture}
\def\PI{2*3.14159265359}
\def\T{4}
\def\freqi{4}
\def\freqj{9}
    \draw [domain=\freqi+1/\T:\freqj+3,variable=\t,smooth,samples=100]
    	plot (\t,{ 1/(\t-\freqi)/\T });

    \draw [domain=\freqi-1/\T:\freqi-3,variable=\t,smooth,samples=100]
    	plot (\t,{ 1/(-\t+\freqi)/\T });
    \draw (\freqi-1/\T,1) --(\freqi+1/\T,1);
    	
    \draw [domain=\freqj-1/\T:\freqi-3,variable=\t,smooth,samples=100]
    	plot (\t,{ 1/(-\t+\freqj)/\T });

    \draw [domain=\freqj+1/\T:\freqj+3,variable=\t,smooth,samples=100]
    	plot (\t,{ 1/(\t-\freqj)/\T });

    	\draw (\freqj-1/\T,1) --(\freqj+1/\T,1);

    \node at (1,-0.5) {$-\infty$};
    \node at (12,-0.5) {$+\infty$};
    \draw [blue] (1,0) -- (12,0);
    \draw (\freqi,0) -- (\freqi,-0.2);
    \node at (\freqi,-0.5) {$f_i$};
    \draw (\freqj,0) -- (\freqj,-0.2);
    \node at (\freqj,-0.5) {$f_j$};

    \draw (\freqj+1/\T,0) -- (\freqj+1/\T,-0.2);
    \node at (\freqj+1/\T+0.5,-0.8) {$f_j+\frac{1}{T}$};

     \draw (\freqj-1/\T,0) -- (\freqj-1/\T,-0.2);
    \node at (\freqj-1/\T-0.5,-0.8) {$f_j-\frac{1}{T}$};

     \draw (\freqi+1/\T,0) -- (\freqi+1/\T,-0.2);
    \node at (\freqi+1/\T+0.5,-0.8) {$f_i+\frac{1}{T}$};

     \draw (\freqi-1/\T,0) -- (\freqi-1/\T,-0.2);
    \node at (\freqi-1/\T-0.5,-0.8) {$f_i-\frac{1}{T}$};

   	\node at (\freqi, 1.5) {$\min(T,\frac{1}{|\theta-f_i|})$};
   	\node at (\freqj, 1.5) {$\min(T,\frac{1}{|\theta-f_j|})$};
\end{tikzpicture}\caption{$\int_{-\infty}^{+\infty} \min(T,\frac{1}{|\theta-f_i|}) \cdot \min (T,\frac{1}{|\theta-f_j|}) \mathrm{d} \theta$}
\end{figure}
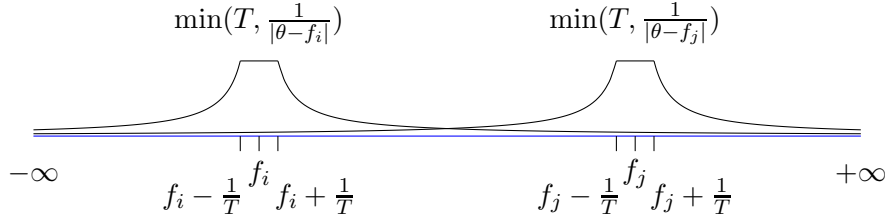


This gives the result for $\nu_i \leq \frac{1}{T}$.  In the alternate case, we have $\nu_i > \frac{1}{T}$, then
\begin{eqnarray*}
|\widehat{y}_i(\theta) | & \lesssim & v_i \cdot \min(T,\frac{1}{|f_i-\theta|} ) + v'_i \cdot \min(T,\frac{1}{|f'_i - \theta|}) \\
& \lesssim & \| a_i \| \cdot \left( \min(T, \frac{1}{|f_i-\theta|}) + \min(T, \frac{1}{ | f'_i- \theta |}) \right) \quad \text{by~Lemma~\ref{lem:two_close_signal}} .
\end{eqnarray*}
By similar reason for Equation (\ref{eq:bound_diagonal_by_integral_two}), we have 

\begin{eqnarray*}
& & \frac{1}{T} \int_0^T a_i(t) \overline{a_j(t)} \mathrm{d} t \notag \\
& \lesssim & \frac{1}{T} \|a_i\| \|a_j\| \underbrace{\int_{-\infty}^{+\infty} \min(T,\frac{1}{|f_i-\theta|}) \cdot \min(T,\frac{1}{|f_j-\theta|}) \mathrm{d} \theta}_{C_1} \\
& + & \frac{1}{T} \|a_i\| \|a_j\| \underbrace{\int_{-\infty}^{+\infty} \min(T,\frac{1}{|f_i-\theta|}) \cdot \min(T,\frac{1}{|f'_j-\theta|}) \mathrm{d} \theta}_{C_2} \\
& + & \frac{1}{T} \|a_i\| \|a_j\| \underbrace{\int_{-\infty}^{+\infty} \min(T,\frac{1}{|f'_i-\theta|}) \cdot \min(T,\frac{1}{|f_j-\theta|}) \mathrm{d} \theta}_{C_3} \\
& + & \frac{1}{T} \|a_i\| \|a_j\| \underbrace{\int_{-\infty}^{+\infty} \min(T,\frac{1}{|f'_i-\theta|}) \cdot \min(T,\frac{1}{|f'_j-\theta|}) \mathrm{d} \theta}_{C_4} .
\end{eqnarray*}
Applying  Lemma \ref{lem:bound_from_infty_to_infty} on the term $C_1$, $C_2$, $C_3$ and $C_4$ respectively, 
\begin{equation*}
\frac{1}{T} \int_0^T a_i(t) \overline{a_j(t) } \lesssim \| a_i\| \| a_j\| \frac{\log(T\Delta f_{i,j})}{\Delta f_{i,j}},
\end{equation*}
where $\Delta f_{i,j} = \min(|f_i-f_j|, |f_i-f'_j|, |f'_i-f_j|, |f'_i - f'_j|)$. 

\end{proof}


\begin{lemma}\label{lem:squares_of_sum_less_than_sum_of_squares}
  Let $\{(v_i,f_i)\}$ and $\{(v'_i,f'_i)\}$ be two sets of $k$ tones
  for which $\min_{i\neq j} |f_i-f_j| \geq \eta$ and $\min_{i\neq j}
  |f'_i - f'_j| \geq \eta$ for some $\eta > 0$.  Suppose that $T >
  C/\eta$ for a sufficiently large constant $C$.  Then these sets can
  be indexed such that
\begin{equation}
\frac{1}{T} \int_0^T | \sum_{i=1}^k (v'_i e^{2\pi \i f'_i t}  - v_i e^{2\pi \i f_i t}) |^2 \mathrm{d} t \leq (1+ O(\frac{\log(k\eta T) \log(k) }{\eta T})) \sum_{i=1}^k \frac{1}{T} \int_0^T  |  v'_i e^{2\pi \i f'_i t}  - v_i e^{2\pi \i f_i t} |^2 \mathrm{d} t.
\end{equation}
\end{lemma}
\begin{proof}
For simplicity, let $a_i(t) =  v_i e^{2\pi f_i t\i} - v'_{i} e^{2\pi f'_{i} t \i }$.
Let's express the square of summations by diagonal term and off-diagonal term, and then bound them separately.
\begin{eqnarray}\label{eq:sum_of_diagonal_nondiagonal}
& &\int_0^T \left|\sum_{i=1}^{k} a_i(t) \right|^2 \mathrm{d} t \notag \\
& = & \int_0^T \left (\sum_{i=1}^{k} a_i(t) \right) \left(\sum_{i=1}^{k} \overline {a_i(t) } \right)   \mathrm{d} t \notag \\
& = & \int_0^T \sum_{i=1}^k \underbrace{ a_i(t) \overline{a_i(t)} }_{\mathrm{ diagonal} } +  \sum_{i\neq j}^k  \underbrace{ a_i(t) \overline{a_{j}(t)} }_{\mathrm{off\text{-}diagonal}} \mathrm{d} t .
\end{eqnarray}






Using the result of Lemma \ref{lem:sum_of_absolute} and $a^2+b^2 \geq 2ab$, we can upper bound the off-diagonal term,
 \begin{eqnarray}
& &\int_0^T a_i(t) \overline {a_j(t)} \mathrm{d} t \notag \\
&\lesssim & \frac{\log( \Delta f_{i,j} T)}{ \Delta f_{i,j} T} \cdot \sqrt{ \int_0^T |a_i(t)|^2 \mathrm{d} t  \int_0^T |a_j(t)|^2 \mathrm{d} t  } \quad \text{by~Lemma~\ref{lem:sum_of_absolute}}\notag \\
&\lesssim & \frac{\log(\Delta f_{i,j} T)}{\Delta f_{i,j} T} \cdot  \left( \int_0^T |a_i(t)|^2 \mathrm{d} t + \int_0^T |a_j(t)|^2 \mathrm{d} t \right) \quad \text{by~$2ab\leq a^2+b^2$} \notag .
\end{eqnarray}
where $\Delta f_{i,j} = \min(|f_i-f_j|, |f_i-f'_j|, |f'_i-f_j|, |f'_i
- f'_j|)$.  If we index such that any $f_i$ is matched with any $f_j'$
with $\abs{f_i - f_j'} < \eta/3$ -- which is possible, since at most
one such $f_j'$ will exist by the separation among the $f_j'$, and
that $f_j'$ will be within $\eta/3$ of at most on $f_i$ -- then we
have $\Delta f_{i, j} \gtrsim \abs{f_i - f_j}$.  If we order the $f_i$
in increasing order, then in fact $\Delta f_{i, j} \gtrsim \eta \abs{i
  - j}$.

If $T > C/\eta$ for a sufficiently large constant $C$, this means that
$\Delta f_{i, j} T \gtrsim \abs{i - j} \eta T \geq e$.  Since $\frac{\log
  x}{x}$ is decreasing on the region, this implies
\[
\frac{\log(\Delta f_{i,j} T)}{\Delta f_{i,j} T} \lesssim \frac{\log(\abs{i - j} \eta T)}{\abs{i-j}\eta T}.
\]
Thus, we have
\begin{eqnarray}\label{eq:bound_for_diagonal_term}
& &\int_0^T a_i(t) \overline {a_j(t)} \mathrm{d} t \lesssim  \frac{\log( |i-j|\eta T)}{  |i-j| \eta T } \cdot \left( \int_0^T |a_i(t)|^2 \mathrm{d} t + \int_0^T |a_j(t)|^2 \mathrm{d} t \right) \quad  .
\end{eqnarray}

Finally, we have 
\begin{eqnarray*}
& &\int_0^T \left|\sum_{i=1}^{k} a_i(t) \right|^2 \mathrm{d} t - \sum_{i=1}^k \int_0^T  |a_i(t)|^2 \mathrm{d} t \\
& = &  \sum_{i\neq j}^k \int_0^T a_i(t) \overline{a_{j}(t)}  \mathrm{d} t \quad \text{by~Equation (\ref{eq:sum_of_diagonal_nondiagonal})} \\
&\lesssim &  \sum_{i\neq j}^k \frac{\log(|i-j| \eta T)}{|i -j|\eta T} \cdot \left( \int_0^T |a_i(t)|^2 \mathrm{d} t + \int_0^T |a_j(t)|^2 \mathrm{d} t \right) \quad \text{by~Equation~(\ref{eq:bound_for_diagonal_term}) } \\
& \leq  &  \frac{\log(k\eta T)}{\eta T} \sum_{i=1}^k \sum_{j\neq i}^k \frac{1}{|i-j|} \cdot \left( \int_0^T |a_i(t)|^2 \mathrm{d} t + \int_0^T |a_j(t)|^2 \mathrm{d} t \right) \\
& =  & 2\frac{\log(k\eta T)}{\eta T} \sum_{i=1}^k \sum_{j\neq i}^k \frac{1}{|i-j|} \int_0^T |a_i(t)|^2 \mathrm{d} t \quad \text{by~symmetry} \\
& \lesssim  &  \frac{\log(k\eta T)}{\eta T}   \sum_{i=1}^k \int_0^T |a_i(t)|^2 \mathrm{d} t \sum_{j\neq i}^k \frac{1}{|i-j|} \\
& \lesssim & \frac{\log(k\eta T) \log (k) }{\eta T} \sum_{i=1}^k \int_0^T |a_i(t)|^2 \mathrm{d} t \quad \text{by~$\sum_{i=1}^{k} \frac{1}{i} \eqsim \log(k)$ }.
\end{eqnarray*}
Thus, we complete the proof.
\end{proof}

\restate{lem:choose_longer_duration}
\begin{proof}
Directly follows by Lemma \ref{lem:squares_of_sum_less_than_sum_of_squares}.
\end{proof}

\section{Lower Bound}
\restate{lem:lower}
\begin{proof}

  Suppose this were possible, and consider two one-sparse signals $y$
  and $y'$ containing tones $(v, f)$ and $(v, f')$,
  respectively.  By Lemma~\ref{lem:two_close_signal},
  \[
  \int_0^T \abs{y(t) - y(t)}^2 \mathrm{d}t \lesssim   \abs{v}^2T^2\abs{f - f'}^2.
  \]
  Consider recovery of the signal $x(t) = y(t)$, and suppose it
  outputs some frequency $f^*$.  This must simultaneously be a good
  recovery for the decomposition $(x^*, g) = (y, 0)$ and $(x^*, g) =
  (y', y-y')$.  These have noise levels $\N^2$ bounded by $\delta
  \abs{v}^2$ and $\delta \abs{v}^2 +O(\abs{v}^2T^2\abs{f - f'}^2)$,
  respectively.  By the assumption of good recovery, and the triangle
  inequality, we require
  \[
  c\frac{2\sqrt{\delta \abs{v}^2} + \sqrt{O(\abs{v}^2T^2\abs{f - f'}^2)}}{Tv} \gtrsim \abs{f - f'}
  \]
  or
  \[
  c \cdot O(\sqrt{\frac{\delta}{T\abs{f - f'}}} + 1) \geq 1.
  \]
  Because $\delta$ may be chosen arbitrarily small, we can choose a
  small constant $c$ such that this is a contradiction.
\end{proof}

\begin{algorithm}
\caption{Continuous Fourier Sparse Recovery}\label{alg:noisy_k_sparse1}
\begin{algorithmic}[1]
\Procedure{$\mathsf{ContinuousFourierSparseRecovery}$}{$x,k,\delta,\alpha,C, F, T,\eta$ }
----- The $\mathrm{Main ~ Algorithm}$
	\State $F$ is the upper bound of frequency, $T$ is the sample duration, $C$ is the approximation factor.
	\State $\delta,\alpha$ are the parameters associated with Hash function.
	\State $c= 1/10$ and $b=8/10$
  	\State $R_1 \leftarrow \mathsf{NoisyKSparseCFFT}(x,k,\delta, \alpha,C, F) $
 	\State $S_1 \leftarrow \mathsf{MergedStages}(R_1,O(k\log k), \eta,c,b)$
 	\State $R_2 \leftarrow \mathsf{NoisyKSparseCFFT}(x,O(k),\delta, \alpha,C, F, T) $
 	\State $S_2 \leftarrow \mathsf{MergedStages}(R_2,O(k\log k), \Omega(\eta),c,b)$
  	\State $S \leftarrow S_1 \cap S_2$, which means only keeping the tones that $S_1$ agrees with $S_2$ by Lemma \ref{lem:prune_twice}.
  	\State $S^* \leftarrow \mathsf{Prune}(S,k)$, which means only keeping the top-$k$ largest magnitude tones.
  	\State \Return $S^*$ 
\EndProcedure
\Procedure{$\mathsf{NoisyKSparseCFFT}$}{$x,k,\delta,\alpha,C,F,T$}
  \State Let $B = k/\epsilon$.
  \For{ $c=1 \to \log(k)$}
  	\State Choose $\sigma$ uniformly at random from $[\frac{1}{B \eta}, \frac{2}{B \eta}]$.
  	\State Choose $b$ uniformly at random from $[0,\frac{2\pi \lceil F/\eta \rceil}{\sigma B}]$.
  	\State $R_c\leftarrow \mathsf{OneStage}(x,B,\delta,\alpha,\sigma,b,C,F,T) $
  \EndFor
  \State \Return $(R_1,R_2,\cdots, R_{\log(k)})$.
\EndProcedure
\Procedure{$\mathsf{MergedStages}$}{$R,m,\eta,c,b$}
	\State $R$ is a list of $m$ tones $(v'_i, f'_i)$
	\State $c$ is some constant $<1$.
	\State $b$ is some constant $<1$.
	\State Sort list $R$ based on $f'_i$.
	\State Building the 1D range search $\mathsf{Tree}$ based on $m$ points by regarding each frequency $f'_i$ as a 1D point on a line where $x_i = f'_i$.
	\State $S\leftarrow \emptyset$, $i \leftarrow 0$
	\While {$i < m$}
		\If {$\mathsf{Tree.Count}(f'_i, f'_i+ c\eta) \geq b\log k$}
			\State $f \leftarrow \text{median} ~\{ ~f'_j ~| ~f'_j \in [f'_i - c\eta, f'_i+ 2c\eta]  \}$
			\State $v  \leftarrow \text{median} ~\{ ~v'_j ~| ~f'_j \in [f'_i-c\eta, f'_i+ 2c\eta]  \}$
			\State $S \leftarrow S \cup (f,v)$
			\State $i \leftarrow \mathsf{Tree.Search}(f'_i+2c\eta+ \eta/2)$, which means walk to the first point that is on the right of $f'_i+2c\eta+ \eta/2$
		\Else
			\State $i \leftarrow i+1$
		\EndIf
	\EndWhile
	\State \Return $S$
\EndProcedure
\end{algorithmic}
\end{algorithm}

\begin{algorithm}
\caption{Continuous Fourier Sparse Recovery}\label{alg:noisy_k_sparse2}
\begin{algorithmic}[1]
\Procedure{$\mathsf{HashToBins}$}{$x,P_{\sigma,a,b},B,\delta,\alpha$}
  \State Compute $\widehat{y}_{j F/B}$ for $j\in [B]$, where $y = G_{B,\alpha,\delta} \cdot (P_{\sigma,a,b}x)$
  \State \Return $\widehat{u}$ given by $\widehat{u}_j =  \widehat{y}_{j F/B}$
\EndProcedure
\Procedure{$\mathsf{OneStage}$}{$x,B,\delta,\alpha,\sigma,b,C,F,T$}
\State $L\leftarrow \mathsf{LocateKSignal}(x,B,\delta,\alpha,\sigma,b,C,F,T)$
\State Choose $a \in [0, 1]$ uniformly at random.
\State $\widehat{u} \leftarrow \mathsf{HashToBins}(x,P_{\sigma,a,b}, B, \delta,\alpha)$
\State \Return $\{(\wh{u}_{h_{\sigma, b}(f')} e^{-2\pi\sigma a f' \i}, f')$ for $f' \in L$ if not $E_{\mathit{off}}(f')\}$.
\EndProcedure
\Procedure{$\mathsf{LocateKSignal}$}{$x,B,\delta,\alpha,\sigma,b,C,F,T$}
	\State Set $t \eqsim \log (FT)$, $t'= t/(c_n+1)$, $D \eqsim  \log_{t'} (FT) $ , $R_{loc} \eqsim \log_{C} (tC)$, $l^{(1)}=F/2$.
	\For{ $i\in [D - 1]$}
		\State $\Delta l \eqsim F/(t')^{i-1} $, $s = \frac{1}{\sqrt{C}}$, $\widehat{\beta} = \frac{ t s}{2\sigma \Delta l}$
		\State $l^{(i+1)} \leftarrow \mathsf{LocateInner}(x,B,\delta,\alpha,\sigma,b,\widehat{\beta},l^{(i)}, \Delta l,t, R_{loc}, {\bf false})$.
	\EndFor
	\State Set $s=1/C$, $t\eqsim \log(FT)/s$, $\Delta l \eqsim st/T $, $\widehat{\beta} = \frac{ t s}{2\sigma \Delta l}$ , $R_{loc} \eqsim \log_{C} (tC)$
	\State $l^{(*)} \quad~\leftarrow~ \mathsf{LocateInner}(x,B,\delta,\alpha,\sigma,b,\widehat{\beta},l^{(D)}, \Delta l,t, R_{loc}, {\bf true} )$.
	\State \Return  $l^{(*)}$.
\EndProcedure
\Procedure{$\mathsf{LocateInner}$}{$x,B,\delta,\sigma,b,\widehat{\beta},l,\Delta l ,t,R_{loc},last$}
	\State Let $v_{j,q} = 0$ for $(j,q)\in [B] \times [t]$.
	\For {$r\in [R_{loc}]$}
		\State Choose $\gamma\in [\frac{1}{2},1]$ uniformly at random.
		\State Choose $\beta \in [\frac{1}{2} \widehat{\beta}, 1 \widehat{\beta}]$ uniformly at random.
		\State $\widehat{u} \leftarrow \mathsf{HashToBins}(x,P_{\sigma,\gamma,b}, B, \delta,\alpha)$.
		\State $\widehat{u}' \leftarrow \mathsf{HashToBins}(x,P_{\sigma,\gamma+\beta,b}, B, \delta,\alpha)$.
		\For {$j\in [B]$}
			\For {$i \in [m]$}
				\State $\theta_{j,i}^r = \frac{1}{2\pi\sigma\beta} ( \phi(\widehat{u}_j/\widehat{u'}_j) +2\pi s_i), s_i\in [ \sigma\beta(l_j -\Delta l/2) ,  \sigma\beta(l_j + \Delta l/2)  ] \cap \mathbb{Z}_+$
				\State $f_{j,i}^r = \theta_{j,i}^r + b \pmod F$
				\State suppose $f_{j,i}^r$ belongs to $\text{region}(j,q)$, 
				\State add a vote to both region($j,q$) and two neighbors nearby that region, e.g. region($j,q-1$) and region($j,q+1$)
			\EndFor
		\EndFor
	\EndFor
	\For {$j \in [B]$}
		\State $q_j^* \leftarrow \{q | v_{j,q}>\frac{R_{loc}}{2} \}$
		\If {$last~={\bf true} $}
			\State $l_j^* \leftarrow \text{median} \{ f_{j,i}^r | f_{j,i}^r \in \text{region}(j,q_j^*), i\in [f], r\in [R_{loc}]\}$
		\Else
			\State $l_j^* \leftarrow \mathrm{center~~of~region}(j, q_j^*)$
		\EndIf
	\EndFor
	\State \Return $ l^*$
\EndProcedure
\end{algorithmic}
\end{algorithm}



\end{document}